\documentclass[11pt]{article}

\newsavebox{\foobox}
\newcommand{\slantbox}[2][0]{\mbox{%
        \sbox{\foobox}{#2}%
        \hskip\wd\foobox
        \pdfsave
        \pdfsetmatrix{1 0 #1 1}%
        \llap{\usebox{\foobox}}%
        \pdfrestore
}}
\newcommand\unslant[2][-.25]{\slantbox[#1]{$#2$}}

\newcommand{\mpi}{\text{\unslant[-.18]\pi}}

\newcommand{\CB}{\mathcal{B}}

\newcommand{\BE}{\mathbb{E}}

\newcommand{\CH}{\mathcal{H}}
\newcommand{\CI}{\mathcal{I}}

\newcommand{\CL}{\mathcal{L}}
\newcommand{\CM}{\mathcal{M}}

\newcommand{\CO}{\mathcal{O}}

\newcommand{\BP}{\mathbb{P}}

\newcommand{\BQ}{\mathbb{Q}}
\newcommand{\CR}{\mathcal{R}}

\newcommand{\CT}{\mathcal{T}}
\newcommand{\CU}{\mathcal{U}}

\newcommand{\sA}{\mathsf{A}}
\newcommand{\sB}{\mathsf{B}}

\newcommand*{\tr}{\mathrm{Tr}}

\newcommand*{\poly}{\mathrm{Poly}}

 \newcommand{\lV}{\lVert}
 \newcommand{\rV}{\rVert}
\newcommand{\vertiii}[1]{{\left\vert\kern-0.25ex\left\vert\kern-0.25ex\left\vert #1 
    \right\vert\kern-0.25ex\right\vert\kern-0.25ex\right\vert}}
\newcommand{\vertiiiNoLR}[1]{{\bigg\vert\kern-0.25ex\bigg\vert\kern-0.25ex\bigg\vert #1 
    \bigg\vert\kern-0.25ex\bigg\vert\kern-0.25ex\bigg\vert}}
\newcommand{\lvertiii}{\bigg\vert\kern-0.25ex\bigg\vert\kern-0.25ex\bigg\vert }
\newcommand{\rvertiii}{\bigg\vert\kern-0.25ex\bigg\vert\kern-0.25ex\bigg\vert }
\newcommand{\norm}[1]{\Vert {#1} \Vert}
\newcommand{\lnorm}[1]{\left\Vert {#1} \right\Vert}

\newcommand{\lnormp}[2]{\lnorm{#1}_{#2}}

\newcommand{\Span}[1]{\mathrm{Span}\{{#1}\}}

\newcommand{\comm}[2]{\left[ #1, #2\right]}
\newcommand{\lr}[1]{\left(#1\right)}
\newcommand{\mlr}[1]{\left[#1\right]}
\newcommand{\glr}[1]{\left\{#1\right\}}
\newcommand{\alr}[1]{\left\langle #1\right\rangle}
\newcommand{\cor}{\mathrm{Cor} }

\newcommand{\labs}[1]{\left\vert {#1} \right\vert}
\newcommand{\e}{\mathrm{e}}

\newcommand{\ri}{\mathrm{i}}
\newcommand{\rd}{\mathrm{d}}

\usepackage{braket}
\usepackage{bm,bbm}

\usepackage[left=2cm, right=2cm, top=2.5cm, bottom=2.5cm]{geometry}
\usepackage[x11names]{xcolor}
\usepackage{fancyhdr, amssymb, cancel, amsmath, graphicx, pgfplots, tikz}
\usepackage{isomath}
\usepackage{comment}
\usetikzlibrary{shadows}

\allowdisplaybreaks

\newcommand{\stylecolor}{IndianRed3}

\usepackage[labelfont={bf,sf, color=\stylecolor}, margin={1.5cm,0cm}]{caption}

\usepackage[colorlinks=true, urlcolor=\stylecolor, linkcolor=\stylecolor, citecolor=\stylecolor, hyperindex=true, linktocpage=true]{hyperref}

\usepackage{amsthm}

\usepackage[explicit]{titlesec}

\newcommand*\sectionlabel{}
\titleformat{\section}
  {\gdef\sectionlabel{}
   \Large\bfseries\scshape}
  {\gdef\sectionlabel{\thesection }}{0pt}
  {\begin{tikzpicture}[remember picture,overlay]
       \end{tikzpicture}
  }
\titlespacing*{\section}{0pt}{0pt}{0pt}

\newcommand*\subsectionlabel{}
\titleformat{\subsection}
  {\gdef\subsectionlabel{}
   \large\bfseries\scshape}
  {\gdef\subsectionlabel{\thesubsection  }}{0pt}
  {\begin{tikzpicture}[remember picture]
    	\draw (-0.15, 0) node[left] {\color{\stylecolor} \textsf{\subsectionlabel}};
	\draw (0.15, 0) node[right] {\color{\stylecolor} \textsf{#1}};
	\fill[color=\stylecolor] (-0.05, -0.23) rectangle (0.05, 0.23);
       \end{tikzpicture}
  }
\titlespacing*{\subsection}{-4pt}{10pt}{0pt}

\newcommand*\subsubsectionlabel{}
\titleformat{\subsubsection}
  {\gdef\subsubsectionlabel{}
   \bfseries\scshape}
  {\gdef\subsubsectionlabel{\thesubsubsection.\ \  }}{0pt}
  {\begin{tikzpicture}[remember picture]
    	\draw (0, 0) node[left] {\color{\stylecolor} \textsf{\subsubsectionlabel}};
	\draw (0, 0) node[right] {\color{\stylecolor} \textsf{#1}};
       \end{tikzpicture}
  }
\titlespacing*{\subsubsection}{-4pt}{7pt}{0pt}

\pgfplotsset{every axis legend/.append style={at={(1.02,1)},anchor=north west}}

\usepackage[framemethod=tikz]{mdframed}

\begin{document}

\pgfplotsset{every axis legend/.append style={at={(1.02,1)},anchor=north west}}

\usetikzlibrary{calc, arrows}
\numberwithin{equation}{section}
   
\newcounter{theor}
\numberwithin{theor}{section}
   \renewcommand{\thetheor}{\thesection.\arabic{theor}}
\newenvironment{theor}[1][]{%
 \refstepcounter{theor}%
  \ifstrempty{#1}%
  {\mdfsetup{%
     frametitle={%
        {\strut \color{\stylecolor} \textsf{Theorem~\thetheor}}}}%
   }%
  {\mdfsetup{%
     frametitle={%
        {\strut \color{\stylecolor} \textsf{Theorem~\thetheor:~#1}}}}%
   }%
   \mdfsetup{innertopmargin=0pt,linecolor=\stylecolor, frametitlerule=false,
             linewidth=1pt, backgroundcolor=\stylecolor!15!white, topline=true, frametitlebackgroundcolor=\stylecolor!15!white}
   \begin{mdframed}[]\relax%
   }{\end{mdframed}}
   
\newenvironment{exam}[1][]{%
 \refstepcounter{theor}%
  \ifstrempty{#1}%
  {\mdfsetup{%
     frametitle={%
        {\strut \color{\stylecolor} \textsf{Example~\thetheor}}}}%
   }%
  {\mdfsetup{%
     frametitle={%
        {\strut \color{\stylecolor} \textsf{Example~\thetheor:~#1}}}}%
   }%
   \mdfsetup{innertopmargin=0pt,linecolor=\stylecolor, frametitlerule=false,
             linewidth=1pt, backgroundcolor=\stylecolor!15!white, topline=true, frametitlebackgroundcolor=\stylecolor!15!white}
   \begin{mdframed}[]\relax%
   }{\end{mdframed}}
   
   \newenvironment{corol}[1][]{%
 \refstepcounter{theor}%
  \ifstrempty{#1}%
  {\mdfsetup{%
     frametitle={%
        {\strut \color{\stylecolor} \textsf{Corollary~\thetheor}}}}%
   }%
  {\mdfsetup{%
     frametitle={%
        {\strut \color{\stylecolor} \textsf{Corollary~\thetheor:~#1}}}}%
   }%
   \mdfsetup{innertopmargin=0pt,linecolor=\stylecolor, frametitlerule=false,
             linewidth=1pt, backgroundcolor=\stylecolor!15!white, topline=true, frametitlebackgroundcolor=\stylecolor!15!white}
   \begin{mdframed}[]\relax%
   }{\end{mdframed}}
   
   \newenvironment{prop}[1][]{%
 \refstepcounter{theor}%
  \ifstrempty{#1}%
  {\mdfsetup{%
     frametitle={%
        {\strut \color{\stylecolor} \textsf{Proposition~\thetheor}}}}%
   }%
  {\mdfsetup{%
     frametitle={%
        {\strut \color{\stylecolor} \textsf{Proposition~\thetheor:~#1}}}}%
   }%
   \mdfsetup{innertopmargin=0pt,linecolor=\stylecolor, frametitlerule=false,
             linewidth=1pt, backgroundcolor=\stylecolor!15!white, topline=true, frametitlebackgroundcolor=\stylecolor!15!white}
   \begin{mdframed}[]\relax%
   }{\end{mdframed}}
   
      \newenvironment{lma}[1][]{%
 \refstepcounter{theor}%
  \ifstrempty{#1}%
  {\mdfsetup{%
     frametitle={%
        {\strut \color{\stylecolor} \textsf{Lemma~\thetheor}}}}%
   }%
  {\mdfsetup{%
     frametitle={%
        {\strut \color{\stylecolor} \textsf{Lemma~\thetheor:~#1}}}}%
   }%
   \mdfsetup{innertopmargin=0pt,linecolor=\stylecolor, frametitlerule=false,
             linewidth=1pt, backgroundcolor=\stylecolor!15!white, topline=true, frametitlebackgroundcolor=\stylecolor!15!white}
   \begin{mdframed}[]\relax%
   }{\end{mdframed}}

    \newenvironment{defn}[1][]{%
 \refstepcounter{theor}%
  \ifstrempty{#1}%
  {\mdfsetup{%
     frametitle={%
        {\strut \color{\stylecolor} \textsf{Definition~\thetheor}}}}%
   }%
  {\mdfsetup{%
     frametitle={%
        {\strut \color{\stylecolor} \textsf{Definition~\thetheor:~#1}}}}%
   }%
   \mdfsetup{innertopmargin=0pt,linecolor=\stylecolor, frametitlerule=false,
             linewidth=1pt, backgroundcolor=\stylecolor!15!white, topline=true, frametitlebackgroundcolor=\stylecolor!15!white}
   \begin{mdframed}[]\relax%
   }{\end{mdframed}}
   
   \newenvironment{theorNB}[1][]{%
 \refstepcounter{theor}%
  \ifstrempty{#1}%
  {\mdfsetup{%
     frametitle={%
        {\strut \color{\stylecolor} \textsf{Theorem~\thetheor}}}}%
   }%
  {\mdfsetup{%
     frametitle={%
        {\strut \color{\stylecolor} \textsf{Theorem~\thetheor:~#1}}}}%
   }%
   \mdfsetup{innertopmargin=0pt,linecolor=\stylecolor, frametitlerule=false,
             linewidth=1pt, backgroundcolor=\stylecolor!15!white, topline=true, frametitlebackgroundcolor=\stylecolor!15!white,nobreak=true}
   \begin{mdframed}[]\relax%
   }{\end{mdframed}}
   
   \newenvironment{corolNB}[1][]{%
 \refstepcounter{theor}%
  \ifstrempty{#1}%
  {\mdfsetup{%
     frametitle={%
        {\strut \color{\stylecolor} \textsf{Corollary~\thetheor}}}}%
   }%
  {\mdfsetup{%
     frametitle={%
        {\strut \color{\stylecolor} \textsf{Corollary~\thetheor:~#1}}}}%
   }%
   \mdfsetup{innertopmargin=0pt,linecolor=\stylecolor, frametitlerule=false,
             linewidth=1pt, backgroundcolor=\stylecolor!15!white, topline=true, frametitlebackgroundcolor=\stylecolor!15!white,nobreak=true}
   \begin{mdframed}[]\relax%
   }{\end{mdframed}}
   
   \newenvironment{propNB}[1][]{%
 \refstepcounter{theor}%
  \ifstrempty{#1}%
  {\mdfsetup{%
     frametitle={%
        {\strut \color{\stylecolor} \textsf{Proposition~\thetheor}}}}%
   }%
  {\mdfsetup{%
     frametitle={%
        {\strut \color{\stylecolor} \textsf{Proposition~\thetheor:~#1}}}}%
   }%
   \mdfsetup{innertopmargin=0pt,linecolor=\stylecolor, frametitlerule=false,
             linewidth=1pt, backgroundcolor=\stylecolor!15!white, topline=true, frametitlebackgroundcolor=\stylecolor!15!white,nobreak=true}
   \begin{mdframed}[]\relax%
   }{\end{mdframed}}
   
      \newenvironment{lmaNB}[1][]{%
 \refstepcounter{theor}%
  \ifstrempty{#1}%
  {\mdfsetup{%
     frametitle={%
        {\strut \color{\stylecolor} \textsf{Lemma~\thetheor}}}}%
   }%
  {\mdfsetup{%
     frametitle={%
        {\strut \color{\stylecolor} \textsf{Lemma~\thetheor:~#1}}}}%
   }%
   \mdfsetup{innertopmargin=0pt,linecolor=\stylecolor, frametitlerule=false,
             linewidth=1pt, backgroundcolor=\stylecolor!15!white, topline=true, frametitlebackgroundcolor=\stylecolor!15!white,nobreak=true}
   \begin{mdframed}[]\relax%
   }{\end{mdframed}}


\pagestyle{fancy}
\renewcommand{\headrulewidth}{0pt}
\fancyhead{}

\fancyfoot{}
\fancyfoot[C] {\textsf{\textbf{\thepage}}}

\begin{equation*}
\begin{tikzpicture}
\draw (\textwidth, 0) node[text width = \textwidth, right] {\color{white} easter egg};
\end{tikzpicture}
\end{equation*}

\begin{equation*}
\begin{tikzpicture}
\draw (0.5\textwidth, -3) node[text width = \textwidth] {\huge  \textsf{\textbf{Speed limits and locality in many-body quantum \linebreak \vspace{0.07in} dynamics}} };
\end{tikzpicture}
\end{equation*}
\begin{equation*}
\begin{tikzpicture}
\draw (0.5\textwidth, 0.1) node[text width=\textwidth] {\large \color{black} \textsf{Chi-Fang (Anthony) Chen$^{{\color{\stylecolor} \mathsf{a}}}$, Andrew Lucas$^{{\color{\stylecolor} \mathsf{b}}}$, Chao Yin$^{{\color{\stylecolor} \mathsf{b}}}$}};
\draw (0.5\textwidth, -0.5) node[text width=\textwidth] {\small $^{{\color{\stylecolor} \mathsf{a}}}$\textsf{Institute for Quantum Information and Matter,
California Institute of Technology, Pasadena, CA, 91125 USA}};
\draw (0.5\textwidth, -1.0) node[text width=\textwidth] {\small $^{{\color{\stylecolor} \mathsf{b}}}$\textsf{Department of Physics and Center for Theory of Quantum Matter, University of Colorado, Boulder, CO 80309 USA}};
\end{tikzpicture}
\end{equation*}
\begin{equation*}
\begin{tikzpicture}
\draw (0, -13.15) node[right, text width=0.5\paperwidth] { \texttt{chifang@caltech.edu, andrew.j.lucas@colorado.edu, chao.yin@colorado.edu}};
\draw (\textwidth, -13.1) node[left] {\textsf{\today}};
\end{tikzpicture}
\end{equation*}
\begin{equation*}
\begin{tikzpicture}
\draw[very thick, color=\stylecolor] (0.0\textwidth, -5.75) -- (0.99\textwidth, -5.75);
\draw (0.12\textwidth, -6.25) node[left] {\color{\stylecolor}  \textsf{\textbf{Abstract:}}};
\draw (0.53\textwidth, -6) node[below, text width=0.8\textwidth, text justified] {\small  
We review the mathematical speed limits on quantum information processing in many-body systems. After the proof of the Lieb-Robinson Theorem in 1972, the past two decades have seen substantial developments in its application to other questions, such as the simulatability of quantum systems on classical or quantum computers, the generation of entanglement, and even the properties of ground states of gapped systems.  Moreover, Lieb-Robinson bounds have been extended in non-trivial ways, to demonstrate speed limits in systems with power-law interactions or interacting bosons, and even to prove notions of locality that arise in cartoon models for quantum gravity with all-to-all interactions.  We overview the progress which has occurred, highlight the most promising results and techniques, and discuss some central outstanding questions which remain open. To help bring newcomers to the field up to speed, we provide self-contained proofs of the field's most essential results.
 };
\end{tikzpicture}
\end{equation*}

\tableofcontents

\begin{equation*}
\begin{tikzpicture}
\draw[very thick, color=\stylecolor] (0.0\textwidth, -5.75) -- (0.99\textwidth, -5.75);
\end{tikzpicture}
\end{equation*}

\titleformat{\section}
  {\gdef\sectionlabel{}
   \Large\bfseries\scshape}
  {\gdef\sectionlabel{\thesection }}{0pt}
  {\begin{tikzpicture}[remember picture]
	\draw (0.2, 0) node[right] {\color{\stylecolor} \textsf{#1}};
	\fill[color=\stylecolor]  (0,0.37) rectangle (-0.7, -0.37);
	\draw (0.0, 0) node[left, fill=\stylecolor] {\color{white} \textsf{\sectionlabel}};
       \end{tikzpicture}
  }
\titlespacing*{\section}{0pt}{20pt}{5pt}

\section{Introduction}
We are all familiar with the idea that there are ``speed limits" on physical dynamical processes.  For example, in special relativity, no two parties can send information faster than the speed of light $c$. This simple observation allows us to reconcile our notion of causality with Einstein's observation that there is no absolute time.  Consider events $A$ and $B$ separated (to one observer, in flat spacetime) by time $\mathrm{\Delta}t = t_B- t_A$ and distance $\mathrm{\Delta}r = |\mathbf{x}_B-\mathbf{x}_A|$.  If $\mathrm{\Delta}t>0$ for this observer, then \emph{all observers} will find (in their own frame) 
\begin{equation}
\mathrm{\Delta}t^\prime>0\quad \text{if and only if} \quad \mathrm{\Delta}r < c\mathrm{\Delta}t.
\end{equation} It is no stretch to say that this speed limit on information underpins our confidence in the theory of relativity, and thus in our understanding of the universe. 

Of course, for more ``human scale" problems, often the speed of light is effectively infinite: $c\approx \infty$.  Still, there can be important \emph{emergent} speed limits on information transmission.  For example, auditory signals propagate at $v_{\mathrm{sound}}\sim 10^{-6}c$.  Listening to an orchestra, it is not important that some information content of the music might \emph{in principle} be transmitted at the speed of light $c$, because the medium through which information is transmitted is subject to a stricter non-relativistic speed limit.

This review article discusses how such stricter non-relativistic speed limits arise in quantum many-body systems.  One way this can arise is when the system is effectively described by a model of particles interacting on a lattice, which is nearly always the appropriate description of a many-body system in condensed matter or atomic physics.  On the lattice, non-interacting particles have dispersion relations with the schematic form $\epsilon(k) \sim J \cos (ka)$, where $J$ is an energy scale and $a$ is the lattice spacing.  The velocity of information is bounded by \begin{equation}
    v_{\mathrm{group}} = \frac{1}{\hbar} \sup_k \left| \frac{\partial \epsilon}{\partial k} \right| \sim \frac{Ja}{\hbar}. \label{eq:Jahbar}
\end{equation}
In typical condensed matter systems, this velocity is roughly $10^{-3}c \lesssim v_{\mathrm{group}} \lesssim 10^{-2}c$.

The particular focus of this review is on the remarkable fact that one can \emph{prove} such speed limits in a huge number of physically realistic lattice models (and, with caveats, in some continuum settings as well!).  Because the relevant literature intersects heavily with physics, mathematics, and quantum information, much of the review may be more formal and precise than a typical physics review article.  However, the subject does not require intense training in modern mathematics, but rather the creative use of simple mathematics which should be familiar to any physicist who has studied quantum mechanics. We will begin gently and guide the reader through the formal proofs of central results in the field while highlighting some more technical but important extensions in recent years, which will almost always be stated without (a full) proof.  

We remark in passing that sometimes the phrase ``quantum speed limit" is used to refer to the Heisenberg energy-time uncertainty principle \cite{heisenberg1,heisenberg2}: namely, if we look at the solution to some time-independent Schr\"odinger equation \begin{equation}
    |\psi(t)\rangle = \sum_{\alpha} c_\alpha \mathrm{e}^{-\mathrm{i}E_\alpha t/\hbar}|\alpha\rangle,
\end{equation}
with $|\alpha\rangle$ the eigenstates of a Hamiltonian, we can show that if $\mathrm{\Delta}E$ is the largest difference between two $E_\alpha$s represented in the above sum, $|\psi(t)\rangle$ and $|\psi(0)\rangle$ cannot be orthogonal before a time~\begin{equation}\label{eq:dtdE>h}
    \mathrm{\Delta}t \gtrsim \frac{\hbar}{\mathrm{\Delta}E}.
\end{equation}
The purpose of our review is to explain why even in the thermodynamic limit when $\mathrm{\Delta}E \rightarrow \infty$ is extensively large, there are still meaningful notions of ``speed limits" deriving from spatial locality.

For the remainder of the introduction, we tell a historical story that places work on Lieb-Robinson bounds into its broader context; starting in Section \ref{sec:LR}, we begin our formal but friendly tour through the mathematical physics of quantum speed limits.

\subsection{The EPR paradox}
Probably the first serious consideration of information propagation in quantum mechanics was in the classic paper \cite{epr} by Einstein, Podolsky, and Rosen.  The resulting EPR paradox goes as follows:  suppose Alice and Bob have two qubits prepared in a Bell pair \cite{Bell}: the wave function of their ``universe" is \begin{equation}
    |\psi\rangle = \frac{|0\rangle_{\sf A}|0\rangle_{\sf B}+ |1\rangle_{\sf A}|1\rangle_{\sf B}}{\sqrt{2}}
\end{equation}
where $|0\rangle,|1\rangle$ represent the up and down states (in the $z$-direction) of a (two-level) spin-$\frac{1}{2}$ system (i.e. qubit), and $\sf A/\sf B$ denote the qubit held by Alice/Bob.   Now suppose Alice measures her qubit to be $|0\rangle$.  No matter how far Bob is from Alice, instantly Bob's qubit is also $|0\rangle$.   EPR believed this must communicate information, and thus quantum mechanics was not compatible with relativity and locality.

Actually, no information has been communicated in this process.  In particular, Bob and Alice only know their measurements agree (and confirm the state was $|\psi\rangle$) after sending classical signals to each other announcing their measurement outcomes: these classical signals travel at most at the speed of light $c$.\footnote{They must also do this experiment many times to confirm that $|\psi\rangle$ was entangled (while outcomes may differ, they always agree)!}  EPR's ``paradox" is fully compatible with relativity and locality.

\subsection{The Lieb-Robinson Theorem}
To further constrain quantum mechanics via locality, we now claim that the Bell state $|\psi\rangle$ is not easy to prepare.  As made precise in Section \ref{sec:corr}, suppose Alice and Bob are separated by distance $L$ with two unentangled qubits $|00\rangle$, then if \begin{equation}
    \mathrm{e}^{-\mathrm{i}Ht}|00\rangle \otimes |\text{rest}\rangle = |\psi\rangle \otimes |\text{rest}^\prime\rangle,
\end{equation}
with ``rest" denoting any additional qubits in the universe. Then, in any lattice model the evolution time $t$ must obey \begin{equation}
    t \ge \frac{L}{v_{\mathrm{LR}}}, \label{eq:sec1tLv}
\end{equation}
where $v_{\mathrm{LR}}$ -- the Lieb-Robinson velocity -- is an $L,t$-independent constant.  This bound holds even in a non-relativistic spin chain where $c=\infty$.  Intuitively, $v_{\mathrm{LR}} \sim Ja/\hbar$ as in (\ref{eq:Jahbar}).

The way the Lieb-Robinson Theorem is stated, as we will do when we prove it in Section \ref{sec:LR}, is not in terms of such a Bell pair experiment.  In 1972, Lieb and Robinson instead thought more abstractly \cite{Lieb1972}, in terms of the operator norms of commutators of Heisenberg-evolved operators.  They proved that for any local spin model, there exists a constant $\mu >0$ such that
\begin{equation}
    \lVert [A(t),B]\rVert \le \mathrm{e}^{\mu (v_{\mathrm{LR}}t-L)}
\end{equation}
where $A,B$ denote Pauli matrices of Alice and Bob respectively, and $A(t) = \mathrm{e}^{\mathrm{i}Ht}A\mathrm{e}^{-\mathrm{i}Ht}$.  In Section \ref{sec:corr} we will, in some detail, explain the connection between this abstract bound, and the concrete constraint on preparing a Bell pair (along with many other entangled states!).  In a nutshell, only when the commutator is large $\mathrm{O}(1)$ can a Bell pair be formed out of $|00\rangle$, which implies (\ref{eq:sec1tLv}).

\subsection{Applications of the Lieb-Robinson bound}
The broader physics and quantum information community only seemed to become aware of the Lieb-Robinson Theorem in the past two decades.  Perhaps the most immediate reason why the Lieb-Robinson bound became better known is that in a beautiful pair of papers \cite{Hastings_koma,nachtergaele06}, it was shown that the Lieb-Robinson bound implies the finite correlation length of any gapped ground state.  This took a curious fact about the dynamics of Heisenberg-evolved operators and connected it to a completely different kind of question of broad and independent interest in the quantum matter community (Section~\ref{sec:gap})!

The story above focuses on the ground state properties of a many-body quantum system.  But, over the past two decades, an increasingly large fraction of the theoretical physics community has begun to focus on the \emph{dynamics} of large quantum systems.  Some of this focus has arisen due to developments such as a theory of many-body localization \cite{Basko_2006,vadim2007,Nandkishore:2014kca,Imbrie_2016,abaninreview,Suntajs:2019lmb}, quantum scars \cite{scar_exp17,Moudgalya:2021xlu}, or prethermalization (discussed in Section \ref{sec:preth}), which suggest that dynamics can be more complicated than textbook hydrodynamics\footnote{Yet even this subject is undergoing a revival of interest! See, e.g., the review \cite{Lucas:2017idv}.}.  Just as much has arisen out of the rapid and impressive developments in experimental quantum simulation  (and the baby steps toward quantum information processing and computation in the lab).   In a quantum simulator, one often studies highly excited states and, in principle, desires control over much of Hilbert space -- not merely the ground state!  Lieb-Robinson bounds then limit how quickly interesting operations can be done (using unitary quantum mechanics alone) in such a simulator.  The consequences of a Lieb-Robinson bound on information spreading and correlations have been observed in actual experiments on cold atomic gases \cite{cheneau}.  

Looking forward, if someday a large-scale quantum computer is built, Lieb-Robinson bounds provide very non-trivial constraints on how efficiently such a computer could operate.  When laying out physical qubits in a two-dimensional chip, one cannot perform quantum state transfer at the speed of light $c$ -- the effect is limited by the emergent Lieb-Robinson velocity! The implications of Lieb-Robinson bounds on the resources required to prepare interesting entangled states will be discussed at length in Section \ref{sec:corr}.

Lastly, the increase in \emph{classical} computational power has been enormous since 1972, and papers now routinely use numerical simulation to model complex quantum dynamics, at least in small systems and at short times; in recent years, researchers have begun to consider potential usage of a quantum computer for the same task. In Section \ref{sec:compute}, we explain how the Lieb-Robinson Theorem gives guarantees for classical simulations accuracy, and, at the same, gives rise to provably efficient quantum algorithms.   

The increased attention to Lieb-Robinson bounds has also led to significant developments and extensions of the original Lieb-Robinson bounds to new settings.  Indeed, an unfortunate reality is that many-body quantum systems realized in the lab usually do \emph{not} merely consist of spins interacting with nearest neighbors.  Charged particles interact with $1/r$ interactions, dipolar objects have $1/r^3$ interactions, and even genuinely neutral objects have $1/r^6$ van der Waals interactions. For many years, there was a theoretical effort to extend the Lieb-Robinson Theorem to systems with power-law interactions; this recently-resolved question will be the focus of Section \ref{sec:power-law}.  Many other systems have interacting bosonic degrees of freedom, which introduces an additional subtlety (Section~\ref{sec:density}).  

Lastly, there is a profound (and not fully understood) connection between the physics captured by Lieb-Robinson bounds and the holographic theory of quantum gravity.  We will briefly describe this story in Section \ref{sec:all-to-all}.   It is likely that a full resolution of these questions will require powerful generalizations of Lieb-Robinson bounds to other ``norms for operator"; we will explain this perspective in detail in Section~\ref{sec:frobenius}.

\subsection{Outline}
Our review is organized into roughly two parts.   In the first part, we introduce more basic content, focusing on minimal lattice models with nearest-neighbor interactions and finite-dimensional Hilbert spaces, to give the reader a sense of the broad scope and implications of Lieb-Robinson bounds.  Section \ref{sec:math} gives a lightning review of useful mathematical definitions, propositions, and conventions.  Section \ref{sec:LR} motivates and proves a standard Lieb-Robinson bound for dynamics on a lattice.  The next three sections all provide key applications of this technique to different problems: the simulatability of quantum dynamics (Section \ref{sec:compute}), bounds on entanglement and correlations (Section \ref{sec:corr}), and the proof that ground states of one-dimensional gapped systems have area law entanglement (Section \ref{sec:gap}).
We also include discussions of bounds on thermalization in Section \ref{sec:thermalization}, although this is a less developed area of the field.

The second part of the review focuses on more recent extensions of the Lieb-Robinson bound away from local spin models.  In Section \ref{sec:frobenius}, we introduce the notion of quantum operator growth and the Frobenius light cone, which have become of recent interest in studies of chaotic many-body dynamics (but prove mathematically interesting as well). Section \ref{sec:all-to-all} describes all-to-all interacting systems, with no strict spatial locality, but still a ``computer science" notion of $k$-locality (which can be used to constrain operator growth and chaos).   Section \ref{sec:power-law} describes the extension of Lieb-Robinson bounds to systems with power-law interactions, while Section \ref{sec:density} extends Lieb-Robinson bounds to systems at a finite charge or energy density, and to certain bosonic systems. Lastly, our final section describes our perspective on important open problems.

Earlier reviews that discuss locality bounds include \cite{hastingsreview1,simsreview,eisertreview,hastingsreview2,nacht_rev19,bound_noneq_rev22}, including their tests in experiment \cite{exp_test_LRB22}. Other recent reviews \cite{Xu:2022vko,Fisher:2022qey} discuss many-body chaos and operator growth. Our review is complementary: we provide a self-contained introduction to this subject (which can appear rather formidable to an outsider!), but also illustrate, with some depth, how the Lieb-Robinson bounds can be applied very broadly. We have found, working in this field, that often one mathematical technique will find surprising applications to multiple types of problems previously thought unrelated.  We hope that our review, organized around a few key mathematical results and their many applications, will inspire future scientists and mathematicians to uncover new results, for many years to come.

\section{Mathematical preliminaries}\label{sec:math}
We begin by reviewing our conventions and important mathematical facts.  The reader may skim this section and refer to it as appropriate throughout the review.
\subsection{Notation}
The prevailing notation is summarized by the following, admittedly scattered, collection of facts and definitions.

We denote the complex numbers by $\mathbb{C}$, the real numbers by $\mathbb{R}$, the non-negative real numbers by $\mathbb{R}^+ = \lbrace t\in \mathbb{R} : t\ge 0\rbrace$, and the integers by $\mathbb{Z}$  (non-negative ones by $\mathbb{Z}^+$). Sets, subsets, and parties (Alice/Bob) will always be denoted with uppercase serif font: $\mathsf{A,B}$, etc. The complement of a subset $\sf A$ is denoted as $\mathsf{A}^{\mathrm{c}}$, which consists of all elements not in $\mathsf{A}$.  (The set of all elements that should be considered will be clear from the context.)

We use roman/upright font $ \e , \ri , \mpi,$ to denote mathematical constants (such as $\ri = \sqrt{-1}$). Variable names which are problem-specific will be italic (e.g. site/vertex $i$ in a lattice/graph).  States in quantum mechanics are described by vectors in a Hilbert space (denoted with $\mathcal{H}$); this Hilbert space will mostly be finite-dimensional.  Vectors in Hilbert space are described in Dirac's bra-ket notation: $|a\rangle$ or $|\psi\rangle$.  Note that we will always use lower-case letters for their arguments.  Uppercase italic letters ($A$, $B$) are reserved for operators on a quantum Hilbert space, which one can often think of as matrices.  We will reserve the letter $I$ for the identity operator: \begin{equation}
    I|\psi\rangle = |\psi\rangle,
\end{equation}
while $H$ is reserved for the Hamiltonian acting on a quantum system, generating time-evolution via the Schr\"odinger equation: \begin{equation}
    H(t)|\psi(t)\rangle = \mathrm{i}\frac{\mathrm{d}}{\mathrm{d}t}|\psi(t)\rangle. \label{eq:schrodinger}
\end{equation}
Here and henceforth, we set $\hbar=1$.  An entry of matrix $A$ will be denoted as $A_{ij} = \langle i|A|j\rangle$.

Sometimes, it is useful to think of the operators acting on $\mathcal{H}$ as vectors themselves.  Mathematically this can be stated as follows: the set of linear transformations, which is denoted as $\mathrm{End}(\mathcal{H})$, is itself a vector space.\footnote{The ``End" stands for endomorphism, which is a mathematical generalization of the notion of linear transformation to more complicated structures.} When we wish to highlight this fact (starting in Section \ref{sec:frobenius}), we will use bra-ket notation with parentheses to write down operators: for example, we write $A$ as $|A)$.  This notation is especially useful when we wish to use the natural notion of Hilbert-Schmidt inner product on this space: \begin{align}
(A|B):= \frac{\tr(A^\dagger B)}{\tr[I]}. \label{eq:innerproduct}
\end{align}
The symbol $:=$ denotes ``the left-hand side of this expression is defined by the right"; the symbol $=:$ would mean the right is defined by the left.
In Section \ref{sec:density}, we will find it useful to consider alternative inner products on this space.  Lastly, in the Heisenberg picture of quantum mechanics, the Schr\"odinger equation (\ref{eq:schrodinger}) generalizes to \begin{equation}
    \frac{\mathrm{d}}{\mathrm{d}t}|A) = \mathcal{L}|A), \label{eq:schrodingerheisenberg}
\end{equation}
with the Liouvillian $\mathcal{L}$ defined as \begin{align}
     \ri [H, A ]=: \CL|A). 
 \end{align}
 Without using this bra-ket notation and assuming $H$ is time-independent, we can integrate (\ref{eq:schrodingerheisenberg}) to find
\begin{align}
    A(t):= \e^{\ri H t}A \e^{-\ri H t} = \mathrm{e}^{\mathcal{L}t}A. \label{eq:Aoft}
\end{align}
It is almost always the case in this review that a result derived for $t$-independent $H$ also holds for $t$-dependent $H$; so we will often write (\ref{eq:Aoft}) even when the result does hold for $t$-dependent $H$ (where $\mathrm{e}^{-\mathrm{i}Ht}$ becomes a time-ordered path integral).\footnote{In the mathematics literature one often finds the notation $A(t) = \tau_t(A)$; we stick to the physics notation $A(t)$.}
 In general, we will denote ``super-operators" (transformations on the space of operators) by curly fonts $\mathcal{A},\mathcal{B}$, etc.

We use $D$ to denote the dimension of a finite-dimensional Hilbert space (often, this will be exponentially large in the number of qubits or degrees of freedom $N$). The spatial dimension is denoted by $d$ (defined in \eqref{eq:d_dim}), which is never particularly large.

Expectation values in quantum mechanics may be denoted with the following shorthand notation: with states, $\alr{A}_\psi :=\langle\psi|A|\psi\rangle$, while with density matrices, $\alr{A}_\rho := \tr(\rho A )$. The symbol $\mathbb{E}[\cdots]$ will denote expectation value with respect to a \emph{classical probability distribution}, not a quantum state, while $\mathbb{P}[\cdots]$ will denote the \emph{classical} probability of an event.

We say that $f(x) = \mathrm{O}(g(x))$ if there is a constant $c<\infty$ such that $|f(x)| \le c |g(x)|$ for all $x$ and that  $f(x) = \mathrm{\Omega}(g(x))$ if there is a constant $c >0 $ such that $|f(x)| \ge c |g(x)|$ for all $x$.

\subsection{Qubits}\label{sec:qubits}
We will often discuss qubits: two-level systems with states $|0\rangle$ and $|1\rangle$. We define the Pauli matrices \begin{subequations}
\begin{align}
    X &= |0\rangle\langle 1| + |1\rangle\langle 0|, \\
    Y &= -\mathrm{i}|0\rangle\langle 1| + \mathrm{i} |1\rangle\langle 0|,\\
    Z &= |0\rangle\langle 0| - |1\rangle\langle 1|.
\end{align}\end{subequations}
Of course, the Hilbert space is often the tensor product of $N$ two-level systems (so the Hilbert space has dimension $D=2^N$).  The most general operator on this system is the tensor product of the operators on each of the $N$ two-level systems.  On a single two-level system, the most general operator can be expanded in the Pauli basis: $A=a_0 I + a_1 X + a_2 Y + a_3 Z$, where all coefficients $a_i$s are real if and only if $A$ is Hermitian.  Therefore we conclude that the most general operator on the $2^N$-dimensional Hilbert space is a sum of $4^N$ different possible \textbf{Pauli strings}.   Denoting $(I,X,Y,Z)  = X^a$ for $a=0,\ldots,3$: \begin{equation}
    A = \sum_{a_1\cdots a_N = 0}^3 A_{a_1\cdots a_N} X^{a_1}\otimes \cdots \otimes X^{a_N}, \quad \text{or}\quad |A) = \sum_{a_1\cdots a_N = 0}^3 A_{a_1\cdots a_N} |a_1\cdots a_N).
\end{equation}
The following fact will prove very useful in Section \ref{sec:frobenius}: 
\begin{prop}[Pauli strings are orthonormal] \label{prop:pauliorthogonal}
Using the inner product (\ref{eq:innerproduct}), \begin{equation}
    (a_1^\prime \cdots a_N^\prime| a_1\cdots a_N) = \mathbb{I}(a_1^\prime=a_1,\ldots, a_N^\prime=a_N).
\end{equation}
\end{prop}
The symbol $\mathbb{I}[\cdots]$ denotes the indicator function, which is 1 if its argument is true, and 0 if its argument is false. As is standard, we will denote with $X_i^a$ the Pauli matrix $X^a$ acting on qubit $i$, tensored with identity $I$ on all other qubits.

\subsection{Graphs and local Hamiltonians}\label{sec:graphreview}
Many-body quantum systems are interesting precisely because the Hamiltonian is usually expressible in a simple way. For most of this review, we will focus on Hamiltonians between $N$ qubits that can be expressed as \begin{equation}
    H(t) = \sum_{i=1}^N \sum_{a=1}^3 h_i^a(t) X_i^a + \sum_{i<j = 1}^N \sum_{a,b=1}^3 h_{ij}^{ab}(t)X_i^a X_j^b. \label{eq:2localH}
\end{equation}
Because $H$ contains terms that have at most two non-identity Paulis at a time, we call $H$ \textbf{2-local}.

Each 2-local Hamiltonian can be naturally described via an undirected graph $\mathsf{G}=(\mathsf{V},\mathsf{E})$, where $\mathsf{V}$ is a vertex set and $\mathsf{E}$ is a collection of two element subsets of $\mathsf{V}$.  We place a qubit (or more generally, a degree of freedom) on each vertex $v\in\mathsf{V}$; we place an edge $\lbrace u,v\rbrace \in \mathsf{E}$ if and only if $h^{ab}_{uv}(t) \ne 0$ for some $a,b$ and $t$ (except for in Section \ref{sec:power-law}).  In many physics problems, one actually has a graph $\mathsf{G}$ which is pre-specified - only couplings between certain qubits are permitted, so it will be natural for us to describe notions of locality and Lieb-Robinson bounds in terms of graphs.  A useful notation is to keep track of what edges are adjacent to each vertex: for $v\in \mathsf{V}$,
\begin{equation}
\partial v = \lbrace e\in\mathsf{E} : v\in e\rbrace \subset \mathsf{E}. 
\end{equation} 
More generally, for any set $\mathsf{A}\subseteq \mathsf{V}$, \begin{equation}
   \partial \mathsf{A} = \lbrace e\in\mathsf{E}: |e\cap \mathsf{A}|=1\rbrace
\end{equation}
contains the edges that connect vertices in $\mathsf{A}$ to those outside of $\mathsf{A}$.
The \textbf{Manhattan distance} (for us, just ``distance") $\mathsf{d}(u,v)$ on a graph is defined as the fewest number of edges that can be traversed to get between any two vertices. This distance measure obeys the triangle inequality: \begin{equation}
    \mathsf{d}(u,v) \le \mathsf{d}(u,w) + \mathsf{d}(w,v), \quad \text{ for any } w\in\mathsf{V}. \label{eq:triangleinequality}
\end{equation}
The diameter of a set $\mathsf{S}\subseteq \mathsf{V}$ is defined as \begin{equation}
    \mathrm{diam}(\mathsf{S}) = \max_{u,v\in\mathsf{S}}\mathsf{d}(u,v).
\end{equation}
We say that the graph is $d$-dimensional (as measured by constant $C>0$) if \begin{equation}
    |\lbrace u\in \mathsf{V} : \mathsf{d}(u,v)\le r\rbrace| \le 1+Cr^d \label{eq:d_dim},
\end{equation}
for any $v\in\mathsf{V}$; this is, of course, satisfied by the $d$-dimensional lattice $\mathbb{Z}^d$ with nearest-neighbor connectivity. 

Naturally, one may consider $k$-local Hamiltonians where each of the interaction terms is $k$-local (e.g., $X_1X_2X_3$ is a $3$-local term). Physically, we often take $k$ as constant independent of the system size $N$.  We will not focus too much on $k$-local Hamiltonians in this review: many of the techniques we have described for $k=2$ generalize somewhat straightforwardly. We think an elegant way to generalize what is described in this review to $k$-local problems is the factor graph construction: see \cite{chen2019operator} for details.

\subsection{Operator norms and identities}

Operator dynamics involve complicated, high-dimensional objects. Nevertheless, we may capture its sizes using prevailing choices of matrix norms. The \textbf{operator norm}
\begin{align}
    \norm{A} &:= \sup_{\ket{\psi}, \ket{\phi}} \frac{\labs{\bra{\phi}A\ket{\psi}}}{\sqrt{\braket{\phi|\phi}\braket{\psi|\psi}}}
\end{align}
controls the matrix element between any possible states and is equivalent to the maximal singular value; the \textbf{Frobenius norm} 
is the root-mean-square of singular values
\begin{align}
    \norm{A}_{\mathrm{F}} &:= \sqrt{\frac{\tr(A^{\dagger}A)}{\tr(I)}}\label{eq:frobeniusnorm}
\end{align}
which corresponds to its strength on random states (see~\eqref{eq:gettingfrobenius}).
Intuitively, the operator norm and Frobenius norm respectively capture the strength of the operator over the \textit{worst} and \textit{average} inputs states. Naturally, the above are special cases of the Schatten $p$-norms at $p=\infty$ and $p=2$
\begin{align}    
    \norm{A}_p &:= 
    \left(\frac{\tr((A^{\dagger}A)^{p/2})}{\tr[I]} \right)^{1/p} \quad \text{for each } 1 \le p \le \infty.
\end{align}
In the above, we have normalized all norms such that $\norm{I}_p = 1$ as it appears natural in our discussions. Perhaps the most useful fact about these norms is the triangle inequality:\begin{equation}
    \norm{A+B}_p\le \norm{A}_p+\norm{B}_p
\end{equation} and H\"older's inequality. 
\begin{prop}[H\"older's inequality for Schatten norms]\label{prop:holder}
For any square matrices $A$ and $B$, we have
\begin{align}
    \norm{AB}_r \le \norm{A}_p\norm{B}_q \quad \text{whenever}\quad \frac{1}{r} = \frac{1}{p}+ \frac{1}{q}.
\end{align}
\end{prop}
In particular, setting $p=\infty$ yields the submultiplicativty of operator norm $\norm{AB}_r \le \norm{A} \cdot\norm{B}_r$. 

Further, since our operator is defined on a Hilbert space with a tensor product structure, it is often helpful to think about projection superoperators that isolate components of the operator according to locality. 
\begin{defn}[Super-projectors]  
\label{defn:super_projector}
For any set $\sf A$, define the projection superoperator $\mathbb{P}_{\mathsf{A}}$ that annihilates operators acting trivially on $\sf A$ by \begin{equation}
    \mathbb{P}_{\mathsf{A}}|a_1\cdots a_N) := \mathbb{I}(a_j\ne 0, \text{ for some $j\in\mathsf{A}$})\ |a_1\cdots a_N).
\end{equation}
Alternatively, let
\begin{align}
 \overline{\mathbb{P}}_{\mathsf{A}}:= \mathcal{I} - \mathbb{P}_{\mathsf{A}^c}, \quad \text{then}     \quad \overline{\mathbb{P}}_{\mathsf{A}} A = A \quad \text{if and only if} \quad A = A_{\mathsf{A}} \otimes I_{\mathsf{A}^c}
\end{align}
where $\CI$ is the identity super-operator. We will say ``operator $A$ is supported on set $\sf A$'' or ``operator $A$ acts non-trivially only on set $\sf A$'' if $\overline{\mathbb{P}}_{\mathsf{A}} A = A$. 
\end{defn}
It is worth distinguishing the functionality of the two super-projectors. The super-projector $\mathbb{P}_{\mathsf{A}}$ isolates the components of the operator that do not vanish in commutators (Section \ref{sec:LRbounds}); for any operator $A$ supported on $\mathsf{A}$ and any operator $B$, we have \begin{equation}
    [A,B]= [A,\mathbb{P}_{\mathsf{A}}B].
\end{equation}
Nicely, this superoperator cannot increase any Schatten norm by too much.
\begin{prop}[Super-projectors and norms]\label{prop:proj_norms}
    The projection $\mathbb{P}_{\mathsf{A}}$ at most increases the Schatten $p$-norm by 
\begin{align}
    \norm{ \mathbb{P}_{\mathsf{A}}B }_{p} \le 2 \norm{B}_{p}.
\end{align}
In particular, it is an actual orthogonal projector in the Hilbert-Schmidt inner product such that
\begin{align}
    \norm{ \mathbb{P}_{\mathsf{A}}B }_{\mathrm{F}} \le \norm{B}_{\mathrm{F}}.
\end{align}
\end{prop}
The super-projector $\overline{\mathbb{P}}_{\mathsf{S}}$ has an elegant representation in terms of Haar averages: 
\begin{prop}[Haar representation for super-projector]\label{prop:haar_rep_projector}
For any set $\mathsf{S}$ and operator $A$, we can present the super-projector as 
\begin{equation}
\overline{\mathbb{P}}_{\mathsf{S}}A=\int [\mathrm{d} U]_{\mathsf{S}^c} U^\dagger A U, \label{eq:barPS}
\end{equation}
where $[\mathrm{d} U]_{\mathsf{S}^c}$ is the Haar measure for unitaries supported on set $\mathsf{S}^c$. By the triangle inequality,
\begin{align}
    \norm{ \overline{\mathbb{P}}_{\mathsf{S}}A }_{p} \le \norm{A}_{p}.
\end{align}
\end{prop}

The following identity will prove immensely useful for us when we prove Lieb-Robinson bounds:
   \begin{prop}[Duhamel's identity] \label{prop:duhamel}
For any square matrices $A$ and $B$ of the same dimension, \begin{equation}
    \mathrm{e}^{(A+B)t } = \mathrm{e}^{At} + \int\limits_0^t \mathrm{d}s \mathrm{e}^{(A+B)(t-s)}B\mathrm{e}^{As}. \label{eq:duhamel}
\end{equation}
\end{prop}
\begin{proof}
Call the right hand side $C(t)$, and observe that $\mathrm{d}C/\mathrm{d}t = (A+B)C$.  Explicitly evaluating the right-hand side, we find $C(0)=I$.  Then, solving the differential equation leads to (\ref{eq:duhamel}).
\end{proof}

\subsection{Remark on $C^*$-algebras}\label{sec:cstaralgebra}
In the mathematics literature, one often discusses Lieb-Robinson bounds using the formalism of $C^*$-algebras.  In a nutshell, the idea is that in an unbounded graph (e.g., $\mathsf{V}=\mathbb{Z}$: the one-dimensional lattice!), many-body states are not precisely defined, yet local operators are.   E.g. it is not possible to specify the state $|\cdots 000\cdots \rangle$ and always correctly set the boundary conditions at infinity; however, we can always discuss local operators, such as Pauli matrices $X_j$ acting on site $j\in\mathsf{V}$, whether or not $\mathsf{V}$ is a finite set.  

$C^*$-algebras provide a rigorous language \cite{bratteli} for discussing the objects that do have precise definitions: bounded operators $O$ which are supported on finite subsets of a (possibly infinite) vertex set $\mathsf{V}$.   The key observation is that commutators of bounded local operators are also bounded local operators: this closure, and the fact that time translation is generated by commutators with a local Hamiltonian, suggests that $C^*$-algebras are a rigorous way to discuss the limit of infinite system size.    We feel, however, that this mathematical structure often distracts from the crucial intuition and ingredients behind making powerful Lieb-Robinson bounds relevant to concrete physics problems, and will not focus on it in this review.  

\section{Lieb-Robinson bounds}\label{sec:LR}
We now turn to the core section of this review, where we introduce the famous Lieb-Robinson Theorem. Assuming only spatial locality, the Lieb-Robinson bounds constrain a large class of many-body quantum dynamics. Though it will take many pages to fully show it, these Lieb-Robinson bounds are remarkably versatile, especially since \textit{exact} dynamics of many-body Hamiltonians are generally analytically intractable and problem-specific.  In contrast, Lieb-Robinson bounds allow us to make statements about all local systems.

\subsection{Warm-up: particle on a line}
\label{sec:warmup_oneparticle}
 Before we dive into the many-body problem, it is instructive to consider a single-particle problem on a 1d lattice ($r \in \mathbb{Z}$) with the Hamiltonian \begin{equation}
    H = -h\sum_{r\in\mathbb{Z}}\left( |r\rangle \langle r+1| + |r\rangle \langle r-1| \right). \label{eq:Hwarmup}
\end{equation}
This is essentially a discrete-space Schrodinger equation, or in computer science literature, the (continuous-time) \textit{quantum walk} \cite{Venegas-Andraca:2012zkr}.
We are interested in the Schr\"odinger picture wave function (presented in the $\ket{r}$ basis) \begin{equation}
    \psi(r,t) := \langle r| \mathrm{e}^{-\mathrm{i}Ht}|0\rangle \quad \text{for} \quad r\in \mathbb{Z}^+, t\in \mathbb{R}^+.
\end{equation}
The amplitude (squared) gives us the probability of the particle being on site $r$. The time evolution can be rewritten in the form of Schr\"odinger equation by inserting a complete basis:
\begin{align} \label{eq:schrodinger31}
    \frac{\mathrm{d}}{\mathrm{d}t} \psi(r,t) &= \bra{r} - \ri H \e^{-\ri H t}\ket{0} =\sum_{r'} \bra{r} - \ri H \ket{r'} \bra{r'}\e^{-\ri H t}\ket{0} = -\ri \sum_{r'} H_{rr'} \psi(r',t).
\end{align}
There are many approaches to solving this problem. One is to Taylor expand in time $t$:
\begin{align}
    \mathrm{e}^{-\mathrm{i}Ht}|0\rangle &= |0\rangle - \mathrm{i}ht \left[ |-1\rangle + |1\rangle\right] - \frac{(ht)^2}{2}\left[ |-2\rangle + 2|0\rangle + |2\rangle\right] + \frac{\mathrm{i}(ht)^3}{6}\left[ |-3\rangle + 3|-1\rangle +3|1\rangle + |3\rangle\right] + \cdots. \label{eq:easyeiHt}
\end{align}
The coefficients in the above expression -- at each order in $t$ -- are binomial coefficients, equivalent to those that count the number of random walks analogous to a random walk (Figure~\ref{fig:1_particle_pascal}).   However, if we are interested in $\psi(r,t)$, the interference between terms at different orders can, a priori, be important:
\begin{align}
    \psi(r,t) = \frac{(-\mathrm{i}ht)^r}{r!} + \binom{r+2}{1}\frac{(-\mathrm{i}ht)^{r+2}}{(r+2)!}+\cdots.
\end{align}
For this particular problem, we can find the exact solution of the oscillatory sum via exact diagonalization~\eqref{eq:exact_diag}. However, in the spirit of a Lieb-Robinson bound, let us only look for an inequality.  Then the argument can greatly simplify.  We take absolute values around the Schr\"odinger equation (\ref{eq:schrodinger31}):
\begin{align}\label{eq:single_particle_ODE}
    \frac{\mathrm{d}}{\mathrm{d}t} \labs{\psi(r,t)} &\le \sum_{r'} \labs{H_{rr'}} \labs{\psi(r',t)}.
\end{align}
For this system of ordinary differential \emph{inequalities} with initial conditions, we have the general exponential bound
\begin{align}
    \labs{\psi(r,t)} \le \sum_{r'} (\e^{A t} )_{rr'}\labs{\psi(r',0)}=  \bra{r}\e^{A t}\ket{0} \quad \text{where} \quad A_{rr'} = \labs{H_{rr'}} \ge 0,
\end{align}
which is the vector version of Gronwall's inequality: 
\begin{align}
    \frac{\mathrm{d}}{\mathrm{d} t} u(t) = f(t) u(t) \quad \text{implies that} \quad \labs{u(s)} \le \labs{u(0)} \exp\left[\int\limits_0^s\labs{f(t)}\mathrm{d}t\right].
\end{align}

 Intuitively, the exponential bound amounts to ignoring the phases in the Schr\"odinger equation and just adding up all of the terms in (\ref{eq:easyeiHt}) \textit{coherently}. This also explicitly reduces the bound to a combinatorial problem defined by the weighted adjacency matrix $A_{rr'}$ (which is entry-wise positive). The exponential conveniently generates all paths connecting sites $0$ to $r$:
\begin{align}
    \bra{r}\e^{A t}\ket{0} &= \sum_{\ell = 0}^\infty \frac{(h|t|)^{\ell}}{\ell!} \cdot \#(\text{paths from $0$ to $r$ with length $\ell$})\notag \\
    & = \sum_{m=0}^{\infty} \frac{(ht)^{r+2m}}{(r+2m)!}\binom{r+2m}{m}
    \le\frac{(2ht)^r}{r!}\cdot \sum_{m=0}^{\infty} \frac{(2ht)^{2m}r!}{(r+2m)!}
    \le \frac{(2ht)^r}{r!} \frac{1}{1-(2ht/r)^2}. 
    \label{eq:1d_sum_paths}
\end{align}
The second line gives precisely the coefficients of Pascal's triangle~\eqref{eq:easyeiHt} without phases; the inequality uses the following convenient bound on binomial coefficients and then that $(r+2m)!/r! \ge r^{2m}$  to obtain the exponential:
\begin{prop}[Bounds on binomial coefficients]
For positive integers $a>b \ge 0$, we have
\begin{align}
    \binom{a}{b} \le 2^a.
\end{align}    
\end{prop} 
To further simplify, apply Stirling's approximation for the factorial (this will be suitably strong throughout the paper).
\begin{propNB}[Stirling's approximation]
For non-negative integers $\ell \ge 0$, we have
\begin{equation}
 \left(\frac{\ell}{\mathrm{e}}\right)^{\ell}  < \ell!.
\end{equation}
\end{propNB}
Since the coefficient of a wave function is bounded by $|\psi(r,t)|\le 1$, we can simply choose the smaller of 1 and~\eqref{eq:1d_sum_paths} as our bound; Eq.~\eqref{eq:1d_sum_paths} is only meaningful when $2\e ht/r \le 1$. We deduce that \begin{equation}
    |\psi(r,t)| \le \min\left(1,\frac{1}{1-\mathrm{e}^{-2}} \left(\frac{2\mathrm{e}ht}{r}\right)^r \right). \label{eq:C31}
\end{equation}
In other words, we obtain an \emph{emergent speed limit}:  asymptotically, the particle cannot travel faster than the speed \begin{equation}
    v_{\mathrm{LR}} = 2\mathrm{e}h \quad \text{for}\quad r \rightarrow \infty.  \label{eq:vlrsec31}
\end{equation}
We are calling this emergent velocity $v_{\mathrm{LR}}$, in analogy with the Lieb-Robinson velocities we will soon introduce. From the single-particle problem to the many-body problem, we will see the Lieb-Robinson bounds have a similar combinatorial flavor: counting (weighted) paths on the lattice. Since the exact dynamics of the many-body evolution now has exponentially larger dimensions and is thus much harder to solve, the bounds that generalize our argument above can become much more important!

\begin{figure}[t]
\centering
\includegraphics[width=.6\textwidth]{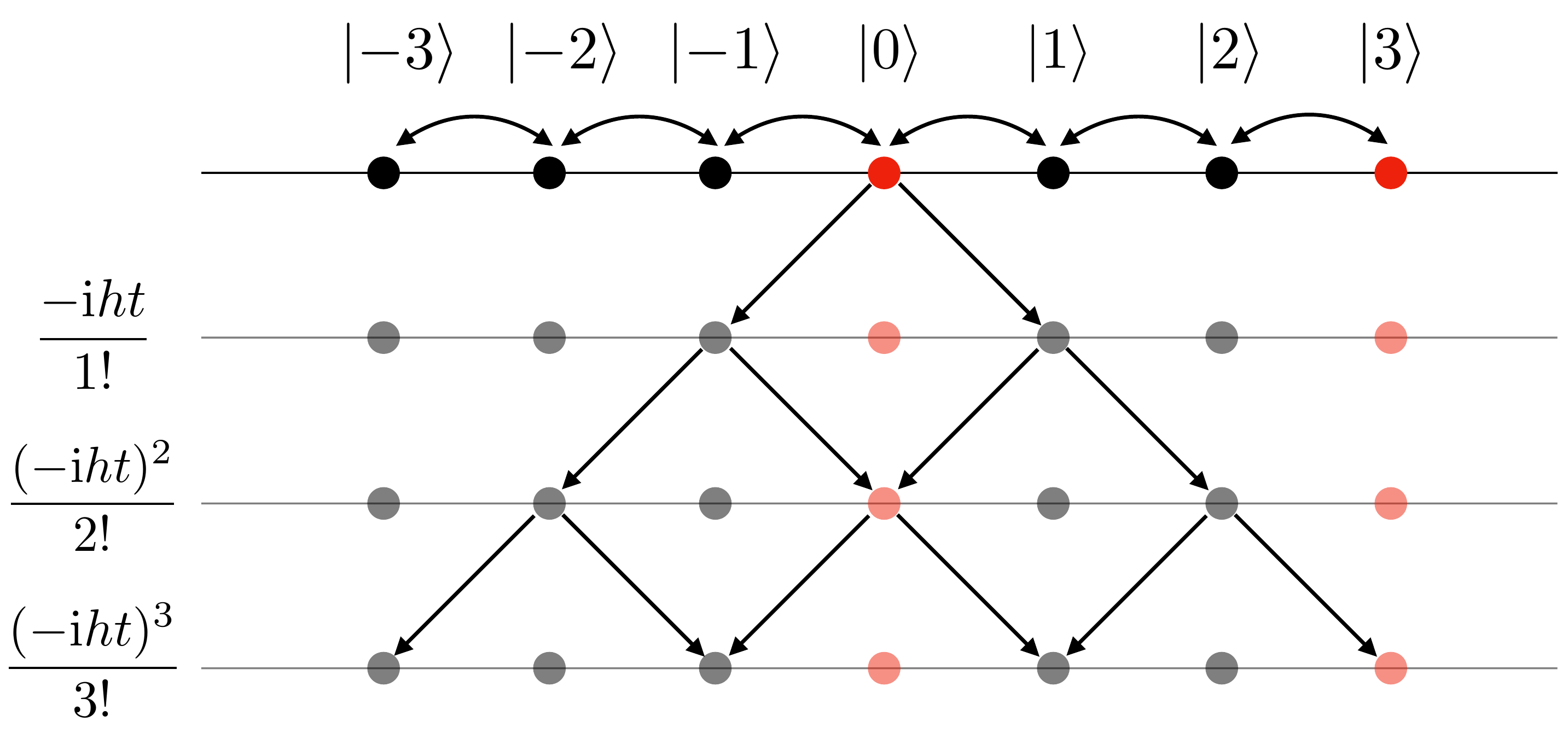}
\caption{Quantum walk of a particle on a line with the nearest neighbor Hamiltonian (\ref{eq:Hwarmup}), starting at $\ket{0}$. The Taylor expansion gives precisely Pascal's triangle (up to phases).}
\label{fig:1_particle_pascal}
\end{figure}
Before we move on, let us comment on the tightness of these naive bounds. The reader may find the triangle inequality approach ``wasteful" -- perhaps one can do a better counting to get a smaller $v_{\mathrm{LR}}$? In Section \ref{sec:babyQW}, we will refine this approach by further utilizing the probabilistic interpretation (constraints) of $|C(r,t)|^2$, and show that this is indeed possible.

For concreteness, let us now compare with the exact diagonalization results; 
the eigenvectors of the Hamiltonian $H$ are the non-normalized plane wave states \begin{equation}
    |k\rangle := \sum_{r\in\mathbb{Z}}\mathrm{e}^{\mathrm{i}kr}|r\rangle,
\end{equation}with eigenvalues \begin{equation}
    H|k\rangle = E_k |k\rangle \quad \text{where} \quad E_k = -2h\cos(k).\label{eq:exact_diag}
\end{equation}
The maximal group velocity in the problem is \begin{equation}
    \frac{\partial E_k}{\partial k} = 2h\sin k \le 2h < v_{\mathrm{LR}}. \label{eq:vgsec31}
\end{equation}
How can anything travel faster than the fastest particle in the system?  We will return to this issue in Section~\ref{sec:babyQW}; for now, we simply remark that our estimate's $r$-dependence is so tight that we cannot improve on the factor $\mathrm{e}$ in (\ref{eq:C31}), although the factor of 2 can be removed using the methods of Section~\ref{sec:self_avoid}: see (\ref{eq:algebraicfrobenius}) and nearby discussion for more.
Moreover, we can solve the problem exactly using these eigenstates.  Define the Fourier transform  \begin{equation}
    \tilde{\psi}(k,t) := \sum_{r} \mathrm{e}^{\mathrm{i}kr} \psi(r,t).
\end{equation}
Using the Schr\"odinger equation, we find \begin{equation}
    \partial_t \tilde{\psi}(k,t) = 2\mathrm{i}h\cos(k)\tilde{\psi}(k,t),
\end{equation}
which can be solved given our initial condition \begin{equation}
    \tilde{\psi}(k,t) = \mathrm{e}^{2\mathrm{i}ht \cos k}.
\end{equation}
Rewrite in terms of the $n^{\mathrm{th}}$ order Bessel function $\mathrm{J}_n$\begin{equation}
    \psi(r,t) = \int\limits_0^{2\mpi}\frac{\mathrm{d}k}{2\mpi } \mathrm{e}^{-ikr} \tilde{\psi}(k,t) = (-\mathrm{i})^{\ell} \mathrm{J}_r(2ht).
\end{equation}
Using the Bessel function asymptotics, we confirm (\ref{eq:C31}) when $t \ll r$. 

\subsection{Lieb-Robinson bounds}\label{sec:LRbounds}
Now, let us turn to many-body quantum mechanics.  Unlike above, it will now prove more natural to discuss the time evolution of \emph{operators}, rather than states.  The reason was described in Section \ref{sec:cstaralgebra}: in the thermodynamic (large particle number limit), a quantum state is an extremely complicated object.   Not only does this make it annoying to discuss, but it also makes it \emph{fragile}:  small local perturbations can completely orthogonalize a quantum state \cite{ortho_cata67}.   For example, if $|\theta\rangle = \cos\theta |0\rangle + \sin\theta |1\rangle$,  \begin{equation}
    \langle 0\cdots 0|\theta\cdots \theta \rangle = \cos^N\theta 
\end{equation}
if there are $N$ qubits, even though $|\theta\cdots \theta \rangle$ is obtained from $|0\cdots 0\rangle$ by a simple sequence of $N$ single-qubit rotations.  As this is the type of perturbation that will arise in quantum mechanics when we look at time evolution generated by unitary $\mathrm{e}^{-\mathrm{i}Ht}$, we will need a different notion of locality.  

With operator growth, we find such a notion:  a local operator is robust to all but the perturbations that arise near its starting location.  This follows from the trivial (but \emph{extremely} important) fact that: \begin{prop}[Locality gives commuting operators]
If a quantum many-body system is defined on set $\mathsf{V}$, and $\mathsf{A},\mathsf{B}\subset \mathsf{V}$ obey $\mathsf{A}\cap \mathsf{B}=\emptyset$, then operators $A$ and $B$, which act non-trivially only on qudits in $\mathsf{A}$ and $\mathsf{B}$ respectively, commute:  $[A,B]=0$.
\end{prop}  
Hence, Lieb-Robinson bounds will capture the dynamics of operators.  In particular, we will evaluate commutators of the form $[A(t),B]$.  By this proposition, the commutator is not zero (as an \emph{operator}) only when $A(t)$ has grown enough to act non-trivially (not as the identity) in $\mathsf{B}$. Formally, we often consider the following quantity
\begin{align}
   C_{\mathsf{AB}}(t) :=  \sup_{A, B : \lVert A\rVert =\lVert B\rVert =1} \lVert[A(t),B]\rVert\quad \text{for any sets}\quad \mathsf{A},\mathsf{B}\subset \mathsf{V}. 
\end{align}
\subsubsection{Warm-up: one-dimensional chain}\label{sec:warmup_LR}

\begin{figure}[t]
\centering
\includegraphics[width=.9\textwidth]{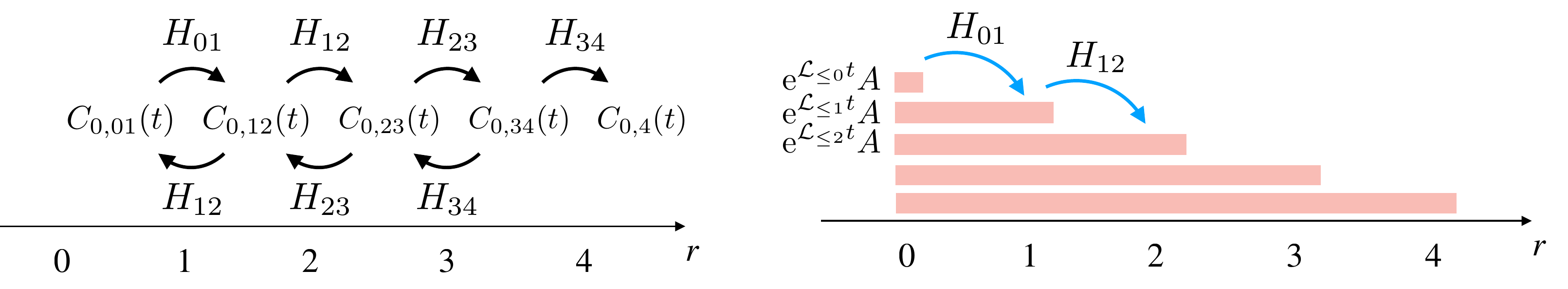}
\caption{ 
(Left) The system of ordinary differential inequalities between the commutator quantities $C_{0,e}$ and $C_{0,r}$ for $r=4$. However, this picture is unsatisfying: the commutator quantity $C_{0,e}$ does not yield a ``local operator'' interpretation; the operator growth contribution comes from both directions. 
(Right) An alternative picture for operator growth with an emphasis on local approximation (Example~\ref{ex:1d_local_approx}).  The growth of $\e^{\CL_{\le r} t} A$ is attributed to $\e^{\CL_{\le r-1} t} A$ and the interaction $H_{r-1,r}$. This is closely related to Lieb-Robinson bounds based on self-avoiding paths (Theorem~\ref{thm:self-avoiding}). 
}
\label{fig:recursive_1d}
\end{figure}

It is instructive to begin by studying a one-dimensional spin chain with  nearest-neighbor interactions.  In the language of Section \ref{sec:math}, this is a 2-local Hamiltonian with interaction graph $\mathsf{G} = (\mathsf{V},\mathsf{E})$, with $\mathsf{V}=\mathbb{Z}$ and $\mathsf{E} = \lbrace \lbrace n,n+1\rbrace: n\in\mathsf{V}\rbrace$.  One could write, assuming that the interactions have bounded norms,
\begin{align}
    H = \sum_{r\in \mathsf{V}} H_{r,r+1} = \sum_{e\in\mathsf{E}} H_e \quad \text{such that} \quad \norm{H_e} \le h.\label{eq:1d_spin_chain}
\end{align}
We will calculate the commutator between different operators at different times:
\begin{equation}
   C_{0r}(t) :=  \sup_{A_0, B_r : \lVert A_0\rVert =\lVert B_r\rVert =1} \frac{\lVert[A_0(t),B_r]\rVert }{2}\quad \text{where} \quad A_0(t):= \e^{\ri H t}A_0\e^{-\ri H t}. \label{eq:defC0r}
\end{equation}
 This commutator, albeit abstract, bounds numerous interesting tasks, as we show in Section \ref{sec:corr}: roughly speaking, it tells us the extent to which a perturbation at site $r$ can modify an observable at site $0$ after time $t$.  Bounding $C_{0r}(t)$ becomes more complicated than the single-particle case. Naively, let us try Talyor-expanding the Heisenberg evolution 
\begin{align}
    \e^{\ri H t}A_0\e^{-\ri H t} = A_0+ &\ri [H_{0,1}, A_0]t+ \ri [H_{-1,0}, A_0]t - \sum_{e_2,e_1}[H_{e_2}, [H_{e_1}, A_0]]\frac{t^2}{2!}+\cdots.\label{eq:naive_expansion}
\end{align}

Indeed, the leading order Taylor expansion tells us that the operator only grows ``one step further" (as in Figure~\ref{fig:1_particle_pascal}) as a consequence of the spatial locality of the Hamiltonian. However, the high-order terms include all possible chains of non-vanishing commutators, which grows \textit{factorially} fast because -- unlike for a single particle -- the operator acts on more sites as we commute with $H$ more times, and \emph{any} of these terms can cause a non-vanishing commutator at a later order.  For example, the following fourth-order term is allowed: $[H_{1,2},[H_{-1,0},[H_{1,2},[H_{0,1},A_0]]]]$.    Directly taking absolute values of~\eqref{eq:naive_expansion} will give a divergent sum\footnote{This is explicitly seen by considering imaginary time evolution, as in $\e^{\beta H}A_0\e^{-\beta H }$.  Here the commutator expansion is genuinely less controlled \cite{locality_temp,Avdoshkin:2019trj}.  A very weak notion of locality is only known in 1d spin chains \cite{araki}.} at constant time $t= \mathrm{O}(1)$.   

The key to obtaining the Lieb-Robinson Theorem is unitarity.  Indeed, because of the factors of $\mathrm{i}$ in the exponential, many of the terms in (\ref{eq:naive_expansion}) will \emph{destructively interfere} with each other.  In Section \ref{sec:frobenius}, we will quantitatively use the intuition that Heisenberg operator dynamics amounts to a \emph{rotation} in a high-dimensional space: in this picture, it is particularly intuitive that many of the terms in (\ref{eq:naive_expansion}) are just ``internal rotations" of the operator, that cannot contribute to the commutator $C_{0r}(t)$.  Yet this picture is better suited for the Frobenius light cone, which is a slightly different object than $C_{0r}(t)$ (see Section \ref{sec:frobenius}).
Therefore, a bound on identifying which terms we can \emph{prove} interfere requires some care. The key insight is to use unitarity to move some of the time evolution onto the operator $B_r$.  Defining \begin{equation}
    H_r := \sum_{e\in\mathsf{E} : e\ni r}H_e = H_{r,r+1} + H_{r-1,r},
\end{equation}
we find that
\begin{align}
    \lVert [A_0(t+\epsilon),B_r]\rVert =  \lVert [A_0(t),B_r(-\epsilon)]\rVert &\le \left\lVert [A_0(t),B_r] -\mathrm{i} \epsilon  [A_0(t),[H_r,B_r]]\right\rVert + \mathrm{O}\left(\epsilon^2\right) \notag \\
      &\le \left\lVert [A_0(t), B_r] +  \mathrm{i}\epsilon [[A_0(t),B_r ],H_r] +\mathrm{i} \epsilon  [[H_r,A_0(t)], B_r]   \right\rVert + \mathrm{O}\left(\epsilon^2\right) \notag \\
     &\le \left\lVert [A_0(t), B_r] +\mathrm{i} \epsilon \mathrm{e}^{-\mathrm{i}\epsilon H_r}  [[H_r,A_0(t)], B_r] \mathrm{e}^{\mathrm{i}\epsilon H_r}  \right\rVert + \mathrm{O}\left(\epsilon^2\right) \notag \\
    &\le \lVert [A_0(t), B_r]\rVert + 2\epsilon \sum_{e\in\mathsf{E} : e\ni r} \lVert H_e \rVert \lVert [A_0(t),B'_e]\rVert  + \mathrm{O}\left(\epsilon^2\right). \label{eq:326}
\end{align}
The first inequality Taylor-expands the Heisenberg equation of motion, using that the commutator $[H,B_r]$ can only survive if a term $H_e$ in $H$ overlaps with site $r$ ($r\in e$). The second inequality uses the Jacobi identity, while the third line conjugates the argument of the norm by an overall unitary to cancel the second term in the second line.  This unitary rotation only causes additional errors at $\mathrm{O}(\epsilon^2)$.  The fourth line uses the triangle inequality, along with $\lVert B_r\rVert=1$ and the definition
\begin{align}
 B'_e := H_e / \lVert H_e\rVert \quad \text{such that}\quad \lVert B'_e\rVert=1.
\end{align}
Remarkably, we see that at order $\mathrm{O}(\epsilon)$, we have found a kind of ``recursive" relation: the bound on $C_{0r}(t)$ reduces to a system of differential inequalities.   Upon taking $\epsilon \rightarrow 0$ and defining $\sup_{\norm{A_0} = \norm{B'_e} = 1} \lVert [A_0(t),B'_e]\rVert = C_{0e}(t)$ a la (\ref{eq:defC0r}), we find that 
\begin{subequations}\begin{align}
   \frac{\mathrm{d}}{\mathrm{d}t} C_{0r}(t) &\le 2h\sum_{e: r\in e}  C_{0e}(t),\label{eq:1d_recursion_RS} \\
   \frac{\mathrm{d}}{\mathrm{d}t} C_{0e}(t) &\le 2h\sum_{e \cap e' \ne \emptyset, e \ne e'} C_{0e'}(t). \label{eq:1d_recursion_Rf}
\end{align}\end{subequations}
Intuitively, the growth of bound $C_{0r}$ is bounded by $C_{0e}$: this nearly coincides with the single-particle recursion~\eqref{eq:single_particle_ODE} except that the intermediate bounds are in terms of commutators $C_{0e}$.  In fact, for our warm-up one-dimensional model, this is not a big deal, since \begin{equation}
    \frac{\mathrm{d}}{\mathrm{d}t} C_{0,\lbrace r,r+1\rbrace} \le 2h \left[ C_{0,\lbrace r-1,r\rbrace} + C_{0,\lbrace r+1,r+2\rbrace}\right]; \label{eq:warmup_1d}
\end{equation}
however, the manipulations of (\ref{eq:326}) will hold for general graphs too, so we keep (\ref{eq:1d_recursion_Rf}) general. 

We can explicitly integrate (\ref{eq:warmup_1d}), which is basically identical to (\ref{eq:single_particle_ODE}).  However, due to the relation (\ref{eq:1d_recursion_RS}) between $C_{0e}$ and $C_{0r}$, we will multiply by an overall factor of 2 in our final bound to account for the two terms in (\ref{eq:1d_recursion_RS}).  Following the derivation of (\ref{eq:C31}), we find that
\begin{equation}
    C_{0r}(t) \le \frac{(4ht)^r}{r!} K,
\end{equation}
for constant $K=1/(1-\mathrm{e}^{-2})$.  Note that there is a factor of $4h$ instead of $2h$, which arises from the extra factor of 2 coming from the fourth line of (\ref{eq:326}), which itself comes from bounding the size of a \emph{commutator}. 

In retrospect, it is perhaps surprising that the commutator growth of a many-body quantum system, like a single-particle system, is controlled by counting over paths (instead of branching trees)!  We will make this picture even sharper in Section \ref{sec:self_avoid}.

At late times, the Lieb-Robinson bounds become very weak.  Indeed, in the argument above, it is always true that \begin{equation} \label{eq:LRboundby1}
    \lVert [A_0(t),B_r]\rVert \le 2 \lVert A_0(t)\rVert \lVert B_r\rVert \le 2,
\end{equation}
So Lieb-Robinson bounds are obviously best used outside the ``light cone" (defined by where we know that $C_{0r}(t)<1$.

\subsubsection{Lieb-Robinson bounds on general graphs}
As we have already noted, the manipulations in (\ref{eq:326}) apply to 2-local Hamiltonians on graph $\mathsf{G}=(\mathsf{V},\mathsf{E})$ with arbitrary connectivity.  With a slight generalization to our derivation above, allowing the operators $A$ and $B$ to have support in sets $\mathsf{A}\subset \mathsf{V}$ and $\mathsf{B}\subset \mathsf{V}$,  we obtain the following theorem: 
\begin{theor}[Lieb-Robinson bound on a general graph \cite{Bentsen:2018uph}]\label{thm:sum_of_paths}
For a 2-local Hamiltonian on graph $\mathsf{G}=(\mathsf{V},\mathsf{E})$, we have that
\begin{equation}
    C_{\sA\sB}(t) \le \sum_{\ell = 0}^\infty \frac{(2|t|)^{\ell}}{\ell!} \sum_{\substack{\text{paths } \Gamma\\ \labs{\Gamma} = \ell} } \prod_{j=1}^{\ell} \lVert H_{\Gamma_j}\rVert.
    \label{eq:CABGamma}
\end{equation}
Each path $\Gamma=\lr{\Gamma_1,\cdots,\Gamma_\ell}$ of length $\ell$ is a sequence of edges $\Gamma_j\in \mathsf{E}$ from the set $\sA=: \Gamma_0$ to the set $\sB=:\Gamma_{\ell+1}$, satisfying the connectivity rule \begin{align}\label{eq:Gammaj_j+1}
    \Gamma_j \cap \Gamma_{j+1} \neq \emptyset,\quad \Gamma_j \neq \Gamma_{j+1},\quad \forall j=0,\cdots,\ell,
\end{align}
Further, defining a symmetric real matrix indexed on $\mathsf{V}$: 
\begin{equation}\label{eq:huv=Huv}
    h_{uv} := \left\{\begin{array}{cc}
        \norm{H_{\{u,v\}}}, & u\neq v \\
        \sum_{w}\norm{H_{\{u,w\}}}, & u=v
    \end{array}\right.,
\end{equation}
we find that
\begin{equation}\label{eq:C<eth}
    C_{\mathsf{AB}}(t) \le  \sum_{u\in \mathsf{A},v\in\mathsf{B}}\lr{\e^{2|t|h}}_{uv}.
\end{equation}
\end{theor}

\begin{proof}

(\ref{eq:CABGamma}) is derived identically to our calculation in Section \ref{sec:warmup_LR}, so we focus on deriving (\ref{eq:C<eth}) as a simplification of (\ref{eq:CABGamma}). Since the sum over $u,v$ in (\ref{eq:C<eth}) comes straightforwardly from considering all the initial and final points in $\mathsf{A}$ and $\mathsf{B}$, let us focus on Taylor expanding the exponentiated matrix $\exp[2h|t|]_{uv}$, and confirming that it includes all terms in (\ref{eq:CABGamma}) beginning and ending at fixed vertices: 
\begin{align}\label{eq:ehuv}
    \lr{\e^{2|t|h}}_{uv} &= \sum_{\ell=0}^\infty \frac{(2|t|)^{\ell}}{\ell!} \sum_{\substack{ u_1,\cdots, u_{\ell'} \in \mathsf{V}: \\ u\neq u_1, u_1\neq u_2, \cdots, u_{\ell'}\neq v } } \sum_{\substack{p_0,\cdots,p_{\ell'}\in \mathbb{Z}_{\ge 0}: \\ p_0+\cdots+p_{\ell'+1}=\ell-\ell'-1} } h_{uu}^{p_0} h_{uu_1}h_{u_1u_1}^{p_1} h_{u_1u_2}\cdots h_{u_{\ell'}u_{\ell'}}^{p_{\ell'}} h_{u_{\ell'} v} h_{vv}^{p_{\ell'+1}}\nonumber\\
    &= \sum_{\ell=0}^\infty \frac{(2|t|)^{\ell}}{\ell!} \sum_{\substack{ u_1,\cdots, u_{\ell'} \in \mathsf{V}: \\ u\neq u_1, u_1\neq u_2, \cdots, u_{\ell'}\neq v } } \sum_{\substack{p_0,\cdots,p_{\ell'}\in \mathbb{Z}_{\ge 0}: \\ p_0+\cdots+p_{\ell'}=\ell-\ell'-1} } \prod_{j'=0}^{\ell'+1} \lr{\sum_{w\neq u_{j'}} h_{u_{j'}w} }^{p_{j'}} h_{u_{j'}u_{j'+1}},
\end{align}
where $u_0:=u$ and $u_{\ell'+1}:=v$.
We need to verify that each term in \eqref{eq:CABGamma} is included in \eqref{eq:ehuv}. For each $\Gamma$ satisfying \eqref{eq:Gammaj_j+1}, assign a vertex $\{u'_j\}=\Gamma_{j}\cap\Gamma_{j+1}$ to each edge $\Gamma_j$ that grows the path further. For two neighboring edges, we have two possibilities: either $u'_{j-1}=u'_j$ or $u'_{j-1}\neq u'_j$. $\Gamma_j$ looks like a branch for the previous case, and a part of the ``core" path for the latter case, as shown in Figure~\ref{fig:path_self_avoid}. The sequence $(\Gamma_0,\ldots, \Gamma_{\ell+1})$ then corresponds to $(u,\ldots,u, u_1,\ldots, u_1, u_2,\ldots, u_{\ell'},\ldots ,u_{\ell'},v\ldots,v)$, where $u_{j'}$ appears $p_{j'}+1$ multiple times with $\sum_{j'} p_{j'}=\ell-\ell'-1$, and the subsequence $(u',u'_1,\cdots,u'_{\ell'},v)$ obeys $u'_{j-1}\ne u'_j$.  The label $(u',u'_1,\cdots,u'_{\ell'},v)$ and $(p_0,\cdots,p_{\ell'+1})$ then corresponds precisely to the second and the third sum in \eqref{eq:ehuv}. Thus it remains to verify for a fixed label, \begin{equation} \label{eq:336}
    \prod_{j'=0}^{\ell'+1} \lr{\sum_{w\neq u_{j'}} h_{u_{j'}w} }^{p_{j'}} h_{u_{j'}u_{j'+1}} \ge \sum_{\substack{\text{paths } \Gamma \text{ corresponding to} \\ (u,u_1,\cdots,u_{\ell'})\text{ and } (p_0,\cdots,p_{\ell'})} } \prod_{j=1}^{\ell} \lVert H_{\Gamma_j}\rVert.
\end{equation} 
This holds because we can always write any $\Gamma$ included in the right-hand side as \begin{equation}
    \Gamma = \lr{\glr{u,w_{01}},\ldots, \glr{u, w_{0 i_{p_0} } }, \glr{uu_1}, \glr{u_1w_{11}}, \ldots, \glr{u_1, w_{1 i_{p_1} } }, \glr{u_1u_2}, \ldots \glr{u_{\ell'}v}, \lbrace vw_{\ell^\prime+1,1}\rbrace ,\ldots },
\end{equation}
where $\glr{u_1w_{11}}, \cdots, \glr{u_1, w_{1 i_{p_1} } }$ for example, are branches hanging at vertex $u_1$.  The left-hand side of (\ref{eq:336}) includes all such terms, but overcounts them because, e.g., we are allowed to include the sum $h_{u_{j'}}^2$ on the left-hand side of (\ref{eq:336}) in a term with $p_{j'}=2$, but on the right-hand side of (\ref{eq:336}) we cannot count the same edge twice in a row.
\end{proof}
\subsubsection{Examples}
\begin{figure}[t]
\centering
\includegraphics[width=.6\textwidth]{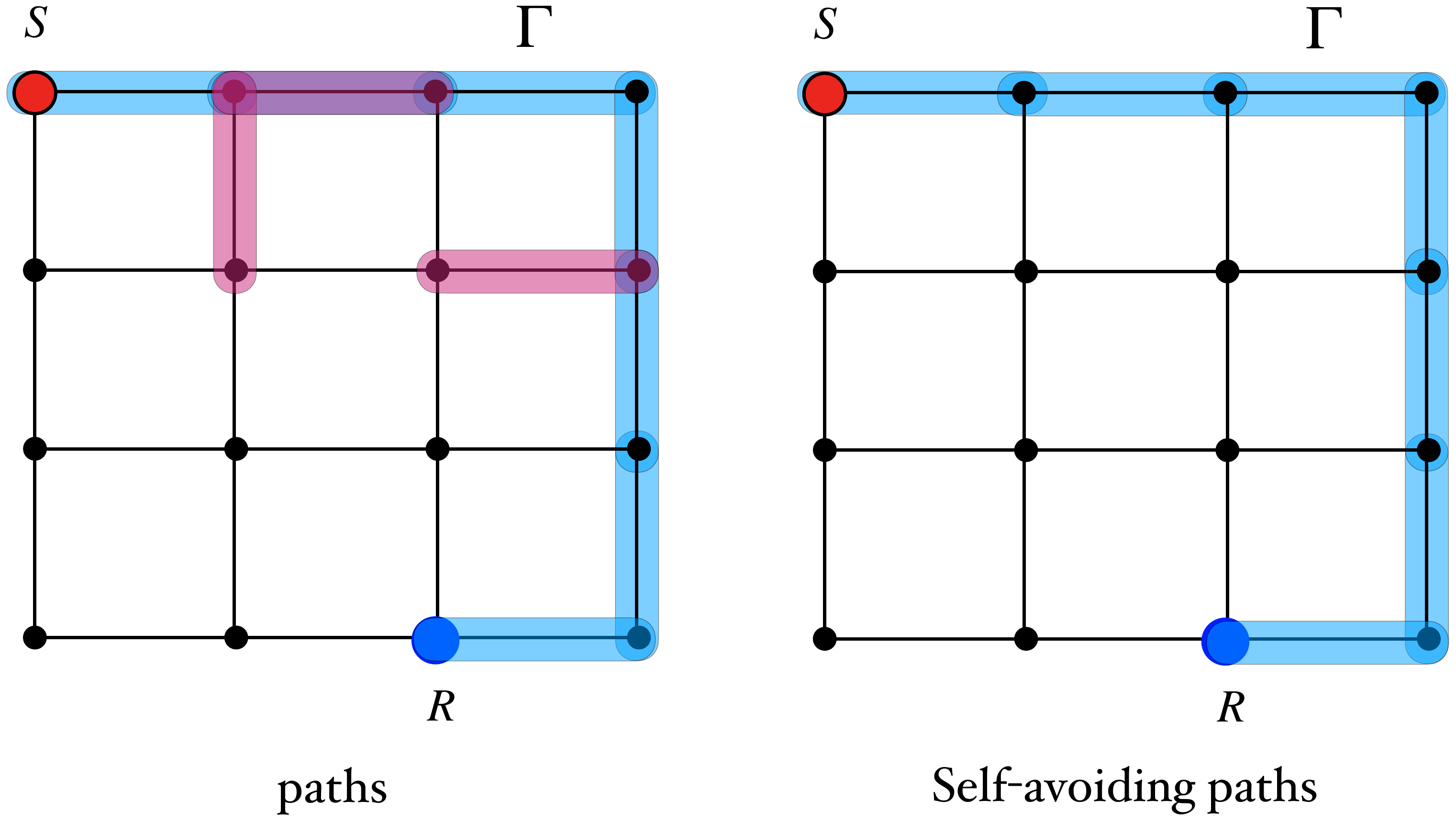}
\caption{(Left) Operator growth is bounded by a sum over paths of Hamiltonian terms.  Terms in blue are the ``core" terms corresponding to the $h_{u_ju_{j+1}}$ in (\ref{eq:ehuv}), while terms in pink come from $h_{u_ju_j}$. (Right) With the more careful expansion of the exponential, the branches and backtracking steps do not contribute to operator growth.}\label{fig:path_self_avoid}
\end{figure}

We will improve the bound in Theorem \ref{thm:sum_of_paths} and thus \eqref{eq:C<eth} in Section \ref{sec:self_avoid}, such that the diagonal elements of $h_{uv}$ can actually be eliminated. However, Theorem \ref{thm:sum_of_paths} is already useful (although not tight) in many examples, as we show here. We will drop the absolute value for time $t$ for notational simplicity.

\begin{prop}[Lieb-Robinson bound on a graph of bounded degree $g$ (loose version)]\label{exam:bound_degree}
Suppose the graph $\mathsf{G}$ has bounded degree $g$, i.e., each vertex connects to at most $g$ edges. If $\norm{H_e}\le h$ for any $e\in \mathsf{E}$, then for two vertices $u,v$ of distance $r=\mathsf{d}(u,v)$ and a constant $0<K<\infty$, \begin{equation}
    C_{uv}(t) \le  K\frac{(4(g-1)ht)^r}{r!}.
\end{equation}
\end{prop}
\begin{proof}
We relax the sum in \eqref{eq:CABGamma} by not demanding the path ends at the fixed vertex $v$:
\begin{align}\label{eq:C<sumxr}
    C_{uv}(t) &\le \sum_{\ell = 0}^\infty \frac{(2ht)^{\ell}}{\ell!} \cdot \#(\text{paths $\Gamma$ starting from $u$ with length $\ell$}) \nonumber\\
    &\le \sum_{\ell = r}^\infty \frac{(2ht)^{\ell}}{\ell!} (2g-2)^\ell = \frac{(4(g-1)ht)^{r}}{r!} \sum_{\ell'=0}^\infty \frac{(4(g-1)ht)^{\ell'} r!}{(r+\ell')!}\nonumber\\
   C_{uv} &\le \min\left(2, \frac{(4(g-1)ht)^r}{r!} \frac{1}{1-2/\mathrm{e}}\right). 
\end{align}
The second line uses the fact that each edge connects to at most $2g-2$ other edges. The last two lines use manipulations analogous to the derivation of \eqref{eq:1d_sum_paths}.  In the last line, we fix the constant $K$ analogously to the discussion above (\ref{eq:C31}); the bound is only meaningful when $4\e(g-1)ht/r \le 2$.
\end{proof}

Proposition \ref{exam:bound_degree} reduces to the calculation of Section \ref{sec:warmup_LR} by taking $g=2$ (which is a one-dimensional lattice with nearest-neighbor interactions), and implies that for a general graph, information propagates under a speed limit proportional to the degree. This scaling is saturated, for example, by translational-invariant free fermions in a $g/2$-dimensional lattice. 

\begin{exam}[Converting factorial to exponential decay]\label{exam:factorialtoexponential}
The tail bound of form $(vt)^r/r!$ is often relaxed to an exponential $C\e^{\mu (vt-r)}$ (with a $\mu$-dependent prefactor $C$), where $\mu$ can be chosen to be arbitrarily large for sufficiently small $t$.  This follows by considering the following chain of inequalities: for $r>1$ \begin{equation}\label{eq:factorialtoexponential}
  \frac{(ct)^r}{r!} \le \mathrm{e}^{-\mu r} \frac{\left( \mathrm{e}^\mu ct\right)^r}{r!} \le \mathrm{e}^{-\mu r} \sum_{n=1}^\infty \frac{\left( \mathrm{e}^\mu ct\right)^n}{n!} = \mathrm{e}^{-\mu r} \left(\mathrm{e}^{\mu vt}-1\right).
\end{equation}
for constant $v=\mu^{-1}\mathrm{e}^\mu c$. The exponential tail bound is both easier to work with when a Lieb-Robinson bound is an intermediate step in a calculation (see many examples in later sections), but is also useful because it allows us to consider exponentially-decaying interactions.  This often arises when one calculates an effective Hamiltonian in the intermediate stages of another proof (see e.g., the discussion in Section \ref{sec:preth}).
\end{exam}

Often, it is desirable to have a good Lieb-Robinson bound for Hamiltonians that have exponentially decaying interactions and are not strictly local.  A strong bound of this kind is given in Theorem \ref{thm:LRexpTail}, as we will use more sophisticated techniques to streamline the proof.  Here we show how to prove such a bound when the vertices $\mathsf{A}$ and $\mathsf{B}$ both consist of a single vertex:
\begin{theor}[Lieb-Robinson bound with exponentially decaying tails]
\label{thm:LRexpTail0}
Let $a,b\in\mathsf{V}$.  For a spatially local Hamiltonian $H=\sum_{\mathsf{S}\subset \mathsf{V}} H_{\mathsf{S}}$ on a graph $\mathsf{V}$ in $d$ spatial dimensions, suppose that for any $u\in\mathsf{V}$, \begin{equation}
    \sum_{\mathsf{S}\ni u} \norm{H_{\mathsf{S}}} \e^{\mu'\, ({\rm diam}(\mathsf{S})-1)} \le h <\infty. \label{eq:expboundnorm}
\end{equation}
We choose to define $H_{\mathsf{S}}\ne 0$ only for connected sets $\mathsf{S}$; this means that terms in $H_{\mathsf{S}}$ may not act non-trivially on all sites within $\mathsf{S}$.
Then for any $0<\mu<\mu'$, there exist constants $c,v>0$ such that for any disjoint $\sA,\sB\subset \mathsf{V}$, \begin{equation}\label{eq:LR_exp0}
    C_{ab}(t) \le c \mathrm{e}^{-\mu \mathsf{d}(a,b)} \left(\mathrm{e}^{\mu vt} -1\right).
\end{equation}
\end{theor}
\begin{proof}
 Lieb-Robinson bounds for exponentially decaying interactions are discussed in \cite{Hastings_koma,nachtergaele06,hastingsreview1,Nachtergaele_2008}.  We start with the following observation: for any $\mu < \mu'$ and $\alpha>d$, there exists a constant $h^\prime$ such that  \begin{equation}
    h^\prime \frac{\mathrm{e}^{-\mu (\mathsf{d}(i,j)-1)}}{\mathsf{d}(i,j)^\alpha} \ge \sum_{\mathsf{S} :\lbrace i,j\rbrace \subseteq \mathsf{S}}  \lVert H_{\mathsf{S}} \rVert. \label{eq:reproducible1}
\end{equation}
Moreover, for any two sites $i,j$ in the vertex set $\mathsf{V}$, \begin{align}\label{eq:reproducible2}
    \sum_{k\in\mathsf{V}\setminus\lbrace i,j\rbrace} \frac{\mathrm{e}^{-\mu( \mathsf{d}(i,k)-1)}}{\mathsf{d}(i,k)^\alpha} \frac{\mathrm{e}^{-\mu ( \mathsf{d}(k,j)-1)}}{\mathsf{d}(k,j)^\alpha} \le \sum_{k\in\mathsf{V}\setminus\lbrace i,j\rbrace} \frac{\mathrm{e}^{-\mu (\mathsf{d}(i,j)-2)}}{\mathsf{d}(i,k)^\alpha\mathsf{d}(j,k)^\alpha} \le K \frac{\mathrm{e}^{-\mu( \mathsf{d}(i,j)-1)}}{\mathsf{d}(i,j)^\alpha}
\end{align}
for some constant $0<K<\infty$.  This latter fact follows from the fact that on a $d$-dimensional lattice, the sum over $k$ converges rapidly at large distances as $\int \mathrm{d}r r^{d-1-2\alpha}$, while at short distances only $\mathrm{O}(n^d)$ sites obey the product $\mathsf{d}(i,k)\mathsf{d}(k,j) \ge n\mathsf{d}(i,j)$.  We have also used the triangle inequality (\ref{eq:triangleinequality}).

Using these facts, we now follow (\ref{eq:huv=Huv}) to write (here $a \in \mathsf{S}_1$ and $b\in \mathsf{S}_\ell$ is implicit):
\begin{align}\label{eq:expsets}
C_{ab}(t) &\le \sum_{\ell=1}^\infty \frac{t^\ell}{\ell!}   \sum_{\Gamma \text{ of length } \ell} \prod_{j=1}^\ell \lVert H_{\Gamma_j}\rVert \le  \sum_{\ell=1}^\infty \frac{t^\ell}{\ell!}  \sum_{\mathsf{S}_1: a\in \mathsf{S}_1} \sum_{\mathsf{S}_2 : \mathsf{S}_1 \cap \mathsf{S}_2 \ne \emptyset} \cdots\sum_{\mathsf{S}_\ell : \mathsf{S}_{\ell-1} \cap \mathsf{S}_\ell \ne \emptyset}  2 \lVert H_{\mathsf{S}_1} \rVert \cdot 2\lVert H_{\mathsf{S}_2} \rVert \cdots 2\lVert H_{\mathsf{S}_\ell}\rVert \notag \\
&\le \sum_{\ell=1}^\infty \frac{(2t)^\ell}{\ell!}  \sum_{v_1\in\mathsf{V}}\sum_{\mathsf{S}_1: \lbrace a,v_1\rbrace \subset \mathsf{S}_1} \sum_{v_2\in\mathsf{V}}\sum_{\mathsf{S}_2: \lbrace v_1,v_2\rbrace \subset \mathsf{S}_2}  \cdots \sum_{\substack{ \mathsf{S}_{\ell-1}: \\ \lbrace v_{\ell-2},v_{\ell-1}\rbrace \subset \mathsf{S}_{\ell-1}} }  \sum_{\mathsf{S}_\ell: \lbrace v_{\ell-1},{b}\rbrace \subset \mathsf{S}_\ell}  \lVert H_{\mathsf{S}_1} \rVert \cdot \lVert H_{\mathsf{S}_2} \rVert \cdots  \lVert H_{\mathsf{S}_\ell}\rVert \notag \\
&\le \sum_{\ell=1}^\infty \frac{(2t)^\ell}{\ell!} \left[ \sum_{v_1\in\mathsf{V}}\sum_{\mathsf{S}_1: \lbrace a,v_1\rbrace \subset \mathsf{S}_1} \sum_{v_2\in\mathsf{V}}\sum_{\mathsf{S}_2: \lbrace v_1,v_2\rbrace \subset \mathsf{S}_2}  \cdots  \sum_{\substack{ \mathsf{S}_{\ell-1}: \\ \lbrace v_{\ell-2},v_{\ell-1}\rbrace \subset \mathsf{S}_{\ell-1}} } \lVert H_{\mathsf{S}_1} \rVert \cdots \lVert H_{\mathsf{S}_{\ell-1}} \rVert \right] \frac{h^\prime \mathrm{e}^{-\mu \mathsf{d}(b,v_{\ell-1})-1)}}{\mathsf{d}(b,v_{\ell-1})^\alpha} \notag \\
&\le \sum_{\ell=1}^\infty \frac{\left(2h^\prime t\mathrm{e}^{\mu}\right)^\ell}{\ell!} \sum_{v_1,\ldots, v_{\ell-1} \in \mathsf{V}}  \frac{\mathrm{e}^{-\mu (\mathsf{d}(a,v_1)+\cdots + \mathsf{d}(v_{\ell-1},b)}}{\mathsf{d}(a,v_1)^\alpha \cdots \mathsf{d}(v_{\ell-1},b)^\alpha} \le \frac{1}{K} \frac{\mathrm{e}^{-\mu \mathsf{d}(a,b)}}{\mathsf{d}(a,b)^\alpha}\sum_{\ell=1}^\infty \frac{\left(2Kh^\prime t\mathrm{e}^{\mu}\right)^\ell}{\ell!}.
\end{align}
In the second line above, we perform the sum over set $\mathsf{S}_\ell$ by first fixing a vertex $v_{\ell-1} \in \mathsf{S}_{\ell-1}$, and then summing over all sets including that site and $b$; we can similarly sum over all intermediate sets by including a site $v_j \in \mathsf{S}_j\cap \mathsf{S}_{j+1}$.  Note that since this choice may not be unique, this line is an overestimate in general.  We do this because now we can carry out the sum over $\mathsf{S}_\ell$ using (\ref{eq:reproducible1}): this is shown in the third line.  In the fourth line, we iterate this argument to reduce the sum to the intermediate vertices $v_1,\ldots, v_{\ell-1}$, which are then bounded using (\ref{eq:reproducible2}). We have also used  the triangle inequality on distances in the exponential.  Since $\mathsf{d}(a,b)\ge 1$, we obtain (\ref{eq:LR_exp0}). 
\end{proof}

\subsection{Self-avoiding paths}\label{sec:self_avoid}
In this section, we describe a simple yet powerful improvement of the proof technique of Theorem \ref{thm:sum_of_paths}, based on the notion of self-avoiding or irreducible paths.
\subsubsection{General results}
In Theorem \ref{thm:sum_of_paths}, we summed over all ``paths" of couplings $H_e$, including those paths that grow branches or backtrack (see Figure~\ref{fig:path_self_avoid}). In fact, a sharper bound is possible by ``ignoring" these ``branches" in Figure~\ref{fig:path_self_avoid}.  One can show that we only need to account for the \textit{self-avoiding} (originally called \emph{irreducible}) paths~\cite{chen2019operator} in Figure \ref{fig:path_self_avoid}.
\begin{theor}[Summing over self-avoiding paths~\cite{chen2019operator}]\label{thm:self-avoiding}
Theorem \ref{thm:sum_of_paths} can be improved to 
\begin{equation}
    C_{\mathsf{AB}}(t) \le \sum_{\ell = 0}^\infty \frac{(2|t|)^{\ell}}{\ell!} \sum_{\substack{\text{self-avoiding paths } \Gamma\\ \labs{\Gamma} = \ell}}\ \prod_{j=1}^{\ell} \lVert H_{\Gamma_j}\rVert
\end{equation}
where each path $\Gamma$ is a sequence of Hamiltonian terms $H_{\Gamma_j}$ from the set $\mathsf{A}$ to the set $\mathsf{B}$, subject to the connectivity rule \begin{equation}
  \Gamma_i \cap \Gamma_j =\emptyset \ \text{if}\ |i-j| > 1.  \label{eq:Self-avoiding}
\end{equation}
\end{theor}

The proof of Theorem~\ref{thm:self-avoiding} is more intricate, and we only provide a sketch of the construction~\cite{chen2019operator}. The key idea is to artfully manipulate the exponential: resumming most terms into new exponentials, and keeping only terms that contribute to operator growth.  The ``irreducible path" is named as the path for which if a single ``irreducible coupling" was dropped, the final commutator would vanish. This resummation is achieved via the following identity:
\begin{align} \label{eq:equivalenceclass}
    \ri [B,A(t)]  &= \CL_{\mathsf{B}} \e^{\CL t}|A) = \CL_{\mathsf{B}} \sum_{\substack{\text{self-avoiding paths } \Gamma}} \cdots \notag \\
& \cdots \ \int\limits_{\substack{t<t_{\ell}<\cdots <t_1\\ \ell = \labs{\Gamma}}} \mathrm{d}t_1\cdots \mathrm{d}t_{\ell}\  \mathrm{e}^{\mathcal{L}(t-t_{\ell})} \mathcal{L}_{\Gamma_{\ell}}  \mathrm{e}^{\mathcal{L}_{\ell-1} (t_{\ell}-t_{\ell-1})} \mathcal{L}_{\Gamma_{\ell-1}}\cdots \mathrm{e}^{\mathcal{L}_1(t_2-t_1)}\mathcal{L}_{\Gamma_1}\mathrm{e}^{\mathcal{L}_0t_1} |A)
\end{align}
 for appropriately chosen intermediate unitaries $\mathrm{e}^{\mathcal{L}_{\ell-1} (t_{\ell}-t_{\ell-1})}$ (depending on the path $\Gamma$). Remarkably, each term in the Taylor expansion of $\e^{\CL t}$ is accounted for in exactly one self-avoiding path. With this identity at hand, we immediately see that the unitaries (containing fictitious terms) do not contribute to commutator growth after employing the triangle inequality: after all, $\lVert \mathrm{e}^{\mathcal{L}_j (t_{j+1}-t_j)} A'\rVert = \lVert A'\rVert$ for any operator $A'$.

\begin{exam}[Irreducible path bounds for the one-dimensional line, and trees]\label{exam:1d_tight}
For the 1d nearest neighbor spin chain~\eqref{eq:1d_spin_chain}, there is a \emph{unique} irreducible path between any two points (since the graph has no cycles (a.k.a. loops) in it).  Hence, Theorem \ref{thm:self-avoiding} implies 
\begin{align}\label{eq:LRB_1d}
     C_{0r}(t) \le \frac{(2h|t|)^r}{r!}.
\end{align}
This matches \textit{exactly} the leading order Taylor expansion, without the higher-order terms in time $t$.   A similar result holds whenever the graph $\mathsf{G}$ is a tree.
\end{exam}
The following corollary follows straightforwardly from the proof of Theorem \ref{thm:self-avoiding}, and is useful since it bounds not only a particular commutator but also the part of the operator that can grow far away at all. 
\begin{corol}[Lieb-Robinson bound for operator to expand]
\label{cor:self-avoiding}
Denote \begin{equation}
    \mathsf{C}_r := \lbrace v\in\mathsf{V} : \mathsf{d}(v,\mathsf{A}) \ge r\rbrace.
\end{equation}
Then for a spatially 2-local Hamiltonian on a graph $\mathsf{V}$ with O(1) maximal degree, there exist O(1) constants $0<c,u<\infty$ such that \begin{equation}
    C_{\mathsf{AC}_r}(t) \le c \frac{(ut)^r}{r!}.
\end{equation}
\end{corol}
To prove this result, one simply sums over all irreducible paths of length $\le r$, which are required to hit any vertex in $\mathsf{C}_r$. The restriction to self-avoiding paths that Theorem \ref{thm:self-avoiding} allows is only quantitatively important in spatially local systems but becomes qualitatively crucial in random all-to-all models (Section~\ref{sec:all-to-all}) and power-law interacting systems (Section~\ref{sec:power-law}).  It also helps in the proof of the following result, which is of high value in the literature:

\begin{theor}[Lieb-Robinson bound with exponentially decaying tails]
\label{thm:LRexpTail}
For $H$ defined as in Theorem \ref{thm:LRexpTail0}, for any $0<\mu < \mu^\prime$,
there exists constants $c,v>0$ such that for any disjoint $\sA,\sB\in \mathsf{V}$, \begin{equation}\label{eq:LR_exp}
    C_{\sA\sB}(t) \le c \cdot \min(|\partial \sA|,|\partial \sB|)\mathrm{e}^{-\mu \mathsf{d}(\sA,\sB)} \left(\mathrm{e}^{\mu vt} -1\right).
\end{equation}
\end{theor}
\begin{proof}[Proof sketch]
Since we have already discussed how to handle exponentially decaying interactions in Theorem \ref{thm:LRexpTail0}, we explain here how to get the \emph{area} prefactor $\min(|\partial \mathsf{A}|,|\partial \mathsf{B}|)$ in (\ref{eq:LR_exp}).  Any self-avoiding path from $\mathsf{A}$ to $\mathsf{B}$ must start by intersecting some site near the boundary $\partial \mathsf{A}$.  For exponentially decaying interactions, the number of such paths can contribute a total weight proportional to $|\partial \mathsf{A}|$. Once we pick the first term in the irreducible path of Theorem \ref{thm:self-avoiding}, then we can use Corollary \ref{cor:self-avoiding} to get the factor of $\mathrm{e}^{-\mu \mathsf{d}(\mathsf{A},\mathsf{B})}$ in (\ref{eq:LR_exp}).
\end{proof}
 
\subsubsection{Optimizing over equivalence classes}

Sometimes in these generalizations, it proves valuable not to choose ``self-avoiding paths" in the way that we have defined above: so long as one can find any exact identity of the form (\ref{eq:equivalenceclass}), the strategy of Theorem \ref{thm:self-avoiding} can be applied.  So far in the literature, this relies on finding an efficient notion of \textbf{equivalence class} on the set of all possible sequences of $\mathcal{L}_e$.  To highlight some ways that this equivalence class construction can be used, let us give two simple examples.

\begin{exam}[1d transverse-field Ising model]
The transverse field Ising model \begin{equation}
    H = -\sum_{r\in \mathbb{Z}} \left[J X_rX_{r+1} + hZ_r\right]
\end{equation}
is a standard integrable model with a storied history \cite{pfeuty}.  Naively applying the Lieb-Robinson bound of Theorem \ref{thm:self-avoiding} to this problem, one arrives at \begin{equation}
    v_{\mathrm{LR}}\lesssim J.
\end{equation}
However, one can do better \cite{commute_graph20}. The actual sequence of coefficients that must occur to grow an operator to the right is $ [hZ_2,[JX_1X_2,[hZ_1,[JX_0X_1,\ldots,]]]]$, since all the $J$ terms commute.
Our equivalence class could simply require that the irreducible coefficients in the sequence correspond to both the $J$ and $h$ terms above, and this leads to \begin{equation}
    v_{\mathrm{LR}} \lesssim \sqrt{Jh}.
\end{equation}
Of course, we could simply drop the $J$ terms if $J\gg h$: the irreducible terms in the sequence of (\ref{eq:equivalenceclass}) could amount to $hX_1,hX_2,\ldots$.  This choice leads to \begin{equation}
    v_{\mathrm{LR}} \lesssim h.
\end{equation}  Clearly, there is a lot of creativity in the equivalence class construction!  \label{example:1dtfim}
\end{exam}
This example suggests it can be important to consider which operators in the Hamiltonian commute with other terms.  While this point was highlighted in \cite{commute_graph20}, they did not use the factor graph and equivalence class-based construction of \cite{chen2019operator}. Marrying these two approaches would be fruitful.

\begin{exam}[1d disordered spin chain]
Consider a Hamiltonian of the form (\ref{eq:1d_spin_chain}), but now where the magnitudes $h_r = \lVert H_{r,r+1}\rVert$ are strongly varying from one site to the next \cite{Baldwin:2022hle}.  Assume for simplicity that the $h_r$ are independent and identically distributed.    In this case, we should generalize Example \ref{example:1dtfim} and look for the weakest links between any two sites we wish to send signals between.  In general, the nature of $v_{\mathrm{LR}}$ on very large distances will depend on the distribution of $J_r$s. Suppose for example that we know the cumulative distribution function \begin{equation}
    F(h) = \mathbb{P}[|h_r| \le h].
\end{equation}
Choosing our equivalence classes to only include the couplings where $h_r \le h_*$ for any $h_*$, we can bound (in the thermodynamic limit) the Lieb-Robinson velocity by: \begin{equation}
    v_{\mathrm{LR}} \le 2\mathrm{e} \times \inf_{h_*}\frac{h_*}{F(h_*)}.
\end{equation}
Notice that $v_{\mathrm{LR}}=0$ if $F(h_*) \sim h_*^\alpha$ for $\alpha<1$, which happens when the probability density for small $h_r$ diverges.
\end{exam}

Theorem \ref{thm:self-avoiding} has recently been generalized to get stronger bounds on the range/size of operators; such a technical achievement is important in proving tighter Lieb-Robinson bounds in problems with decaying interactions \cite{preth22}.

\subsection{Generalization to open systems}
It is straightforward to find Lieb-Robinson bounds in open quantum systems \cite{open_LRB}.  We present a modern version reminiscent of Theorem~\ref{thm:self-avoiding}:
\begin{theor}[ Lieb-Robinson bounds for open systems]
\label{thm:open_LRB}
For local Lindbladian $\mathcal{L}$ (i.e., generator of a completely positive unital map), defined on a graph analogous to Theorem \ref{thm:self-avoiding},
\begin{equation}
    C_{\mathsf{AB}}(t) \le \sum_{\ell = 0}^\infty \frac{(|t|)^{\ell}}{\ell!} \sum_{\substack{\text{self-avoiding paths } \Gamma\\ \labs{\Gamma} = \ell}}\ \prod_{j=1}^{\ell} \lVert \mathcal{L}_{\Gamma_j}\rVert_{\infty-\infty}
\end{equation}
where the super-operator norm is defined by
\begin{align}
    \lVert \mathcal{L}\rVert_{\infty-\infty}: = \sup_{O }\frac{\norm{\CL[O]}}{\norm{O}}.
\end{align}
\end{theor}  
 In an open quantum system, the intermediate terms are not unitary but still satisfy $\lVert \mathrm{e}^{\mathcal{L}_k t} \rVert_{\infty-\infty} \le 1$~\cite{wolf2012quantum}: the proof of Theorem \ref{thm:self-avoiding} immediately generalizes to open systems!  The weaker Lieb-Robinson bounds discussed in Section \ref{sec:LRbounds} can then also be derived for open systems by just replacing $2\lVert H_e\rVert$ with $\lVert \mathcal{L}_e\rVert_{\infty-\infty}$ as appropriate.  It is also straightforward to extend these results to $k$-local Hamiltonians that couple more than two sites simultaneously.

\section{Bounds on simulatability}\label{sec:compute}
A general Hamiltonian evolution requires exponentiating a large matrix. Nevertheless, exploiting locality, Lieb-Robinson bounds provide a rigorous starting point for studying the complexity of simulating many-body quantum systems on a classical computer and a quantum computer. In this section, we discuss this issue, along with many interesting extensions of Lieb-Robinson bounds that arise from this perspective.

\subsection{Lieb-Robinson bounds and local approximations for dynamics} 
To guide our discussion, we consider classically simulating the value of a local observable $A$ supported in finite set $\mathsf{A}$, given an initial state $\rho$.  In particular, we ask what is the expected value of $A$ after a later time $t$: in the Heisenberg picture, this is equal to $\tr[ \rho A(t)]$. At $t = 0$, this reduces to the marginal $\tr_{\mathsf{A}}[\rho_{\mathsf{A}} A]$, where \begin{equation}
    \rho_{\mathsf{A}} = \tr_{\mathsf{A}^{\mathrm{c}}} \rho 
\end{equation}
is the reduced density matrix of $\rho$ on subset $\mathsf{A}$.  If we know how to calculate $\rho_{\mathsf{A}}$ efficiently (which we often will), then we will classically evaluate $\tr_{\mathsf{A}}[\rho_{\mathsf{A}} A]$ on a computer, which requires far fewer classical bits of memory.  As time evolves ($t > 0$), intuitively,  Lieb-Robinson bounds tell us that the operator $A(t)$ will not have grown too large at small times $t$ -- if it had, then commutators $C_{\mathsf{AB}}(t)$ would be large for sets $\mathsf{B}$ far from $\mathsf{A}$.  Hence for small $t$, the calculation of $\tr[ \rho A(t)]$ should also be tractable.

\subsubsection{Local approximants}

In this section, we discuss how to formally relate the bounds on commutators to the classical simulatability implied above. In many contexts, the local approximation form of Lieb-Robinson bounds are both conceptually and technically more powerful (e.g., when the Hamiltonian has a power-law decaying tail: see Theorem~\ref{thm:lrpowerlaw}).
\begin{prop}[Commutator bounds and local approximant]\label{prop:tildA}
For any operator $A$ and vertex subset $\mathsf{S}\subset \mathsf{V}$,
\begin{align}
    \norm{\overline{\mathbb{P}}_{\mathsf{S}}A-A}\le \sup_{\norm{B_{\mathsf{S}^{\mathrm{c}}}} =1}\norm{[A, B_{\mathsf{S}^{\mathrm{c}}}]}\le 2\norm{\overline{\mathbb{P}}_{\mathsf{S}}A-A}.
\end{align}
Recall the definition of $\overline{\mathbb{P}}_{\mathsf{S}}$ in (\ref{eq:barPS}).
\end{prop}

\begin{proof}
The first inequality uses the Haar integral representation (Proposition~\ref{prop:haar_rep_projector}) where $[\mathrm{d} U]_{\mathsf{S}}$ denotes the Haar measure of unitaries on set $\mathsf{S}$
\begin{equation}
    \norm{A - \overline{\mathbb{P}}_{\mathsf{S}}A} = \left\lVert\int [\mathrm{d} U]_{\mathsf{S}} (A-U^\dagger A U)\right\rVert \le \int [\mathrm{d} U]_{\mathsf{S}} \cdot \norm{[A,U]} \le \sup_{\norm{B_{\mathsf{S}}} =1}\norm{[A, B_{\mathsf{S}}]}
\end{equation}
and applies the commutator bound for each $U$. The second inequality uses that\begin{equation}
    \norm{[A, B_{\mathsf{S}}]} = \norm{[A-\overline{\mathbb{P}}_{\mathsf{S}}A , B_{\mathsf{S}}]}\le 2\norm{A-\overline{\mathbb{P}}_{\mathsf{S}}A} 
\end{equation}
for any operator $B_{\mathsf{S}}$ with $\norm{B_{\mathsf{S}}}=1$. This is the advertised result.
\end{proof}
 For the task for obtaining the local marginal, $\tr_{\mathsf{A}}[\rho_{\mathsf{A}} A]$, we additionally require an explicit, calculable form of the local approximation; this will also help implement the unitary dynamics efficiently on a gate-based quantum computer (Section~\ref{sec:Q_Ham_sim}).  We can already see how to do this by tracing back to Corollary \ref{cor:self-avoiding}, which we illustrate again in the following example:
 
\begin{exam}[Local approximation error on a 1d spin chain]
\label{ex:1d_local_approx}
Consider the Example of Section \ref{sec:warmup_LR}. Decompose the Hamiltonian $H = H_{r-1,r} + (H_{\le r-1} + H_{\ge r})$. Consider the dynamics according to a restrcted Hamiltonian $\tilde{A}(t) = \e^{\CL_{\le r -1} t} A$.  Here by construction, \begin{equation}
    \overline{\mathbb{P}}_{\le r-1} \tilde{A}(t) = \tilde{A}(t),
\end{equation}
although $\tilde A(t) \ne \overline{\mathbb{P}}_{\le r-1} A(t)$ in general.  Still, we see that
\begin{align}
\norm{A(t) - \tilde{A}(t)}
& \le \int_{0}^{\labs{t}} \mathrm{d}s \norm{ \CL_{r-1,r} \e^{\CL_{\le r-1} s} A} \le \int_{0}^{\labs{t}} \mathrm{d}s \; 2h \frac{(2hs)^{r-1}}{(r-1)!}
= \frac{(2h\labs{t})^{r}}{r!}.
\end{align}
The first inequality uses Duhamel's identity (Proposition~\ref{prop:duhamel}) for $\CL = \CL_{r-1,r} + (\CL_{\le r-1} + \CL_{\ge r})$
\begin{align}
    A(t) &= \int_{0}^t \mathrm{d}s\ \e^{\CL (t-s)}\CL_{r-1,r} \e^{\CL_{\le r-1} s}  A + \e^{\CL_{\le r-1} t} A
\end{align}
and the unitary invariance of the operator norm.
The second inequality expands $\CL_{r-1,r} \e^{\CL_{\le r-1} s} A$ into a sum over self-avoiding paths (Theorem~\ref{thm:self-avoiding}).    
\end{exam}

The local approximation holds more generally.  For any Hamiltonian, we define
\begin{subequations}
\begin{align}
    H_{\mathsf{B}}&:= \overline{\mathbb{P}}_{\mathsf{B}} H \quad \text{for each set}\quad \mathsf{B}\label{eq:H_B}\\
    H_{\mathsf{B}:\mathsf{C}} &:= H_{\mathsf{B}\cup \mathsf{C}} - H_{\mathsf{B}} - H_{\mathsf{C}} \quad \text{for any disjoint sets}\quad \mathsf{B} \quad\text{and} \quad \mathsf{C} \label{eq:B:C}.
\end{align}
\end{subequations}
\begin{prop}[Truncated Hamiltonian evolution]\label{prop:local_Heisenberg_evolution}
For a spatially local Hamiltonian in $d$ spatial dimensions, any regions $\sf B$ with complement $\sf C:={\sf B}^{\mathrm{c}}$, and any operator $A$ supported on $\sf A$, there exist constants $\mu,v, c>0$ such that
\begin{align}
        \lnorm{\e^{\ri ( H_{\sf B} + H_{\sf C} + H_{B : C} )t} A \e^{-\ri ( H_{\sf B} + H_{\sf C} + H_{B : C} )t} - \e^{-\ri H_{\sf B} t} A \e^{\ri H_{\sf B} t}} \le 
            c \cdot \labs{\partial \sf A} \norm{A}\mathrm{e}^{-\mu d(A,C)} \left(\mathrm{e}^{\mu v|t|} -1\right).
\end{align}
\end{prop}
Indeed, the error term is reminiscent of the original Lieb-Robinson bound (Theorem~\ref{thm:LRexpTail}).   It is straightforward to generalize this result to any case where a Lieb-Robinson bound (\ref{eq:LR_exp}) holds (but with an additional surface area term due to time integration).
\begin{proof}
By Duhamel's identity,
\begin{align}
    \lnorm{\e^{(\CL_{\mathsf{B}} + \CL_{\mathsf{C}} + \CL_{\mathsf{B} : \mathsf{C}})s} A - \e^{\CL_{\mathsf{B}} s} A}
    &= \lnorm{ \int\limits_0^t \mathrm{d}s\;\e^{(\CL_{\mathsf{B}} + \CL_{\mathsf{C}} + \CL_{\mathsf{B} : \mathsf{C}})(t-s)} \CL_{\mathsf{B} : \mathsf{C}} \e^{ \CL_{\mathsf{B}} s}A  } \notag \\ 
    &\le  \sum_{\ell = 0}^\infty \frac{(2|t|)^{\ell}}{\ell!} \sum_{\substack{\text{self-avoiding paths $\Gamma: \mathsf{A}\rightarrow \mathsf{C}$} \\ \labs{\Gamma} = \ell} } \prod_{j=1}^{\ell} \lVert H_{\Gamma_j}\rVert.
\end{align}
The last inequality is the observation that any path from set $\mathsf{A}$ to set $\mathsf{C}$ must contain a term in $\mathsf{B}:\mathsf{C}$. The conversion from factorial form to exponential form is analogous to Theorem~\ref{thm:LRexpTail}.
\end{proof}

\subsubsection{Classical simulation with controlled error} 
By approximating $A(t)$ by an operator with strictly local support, we can efficiently simulate the expectation of local observables at short times.

\begin{prop}[Classical simulation from local approximation]
Consider a $d$-dimensional lattice and local operators $A, B$ with $\norm{A}=1, \norm{B}=1$ acting on small set $\sf R$ with $\labs{\sf R}= \mathrm{O}(1)$. Consider the correlation function $\langle A(t)B\rangle_{\rho_0}:=\mathrm{tr}[\rho_0A(t)B]$ and suppose the marginals $\tr_{\sf S^c}[ \rho_0] $ can be obtained at cost $\e^{\mathrm{O}(\labs{\sf S})}$. 
Then, for any $\epsilon>0$, there exists a classical algorithm that outputs the local expectation $\langle A(t) B\rangle_{\rho_0}$ up to error $\epsilon$ with 
\begin{align}
\text{(classical memory and runtime)} \le  \exp\left( \mathrm{O}\left(vt+\frac{1}{\mu}\log \frac{c}{\epsilon \mu v}\right)^d\right).
\end{align}
\end{prop}
\begin{proof}
By Proposition~\ref{prop:local_Heisenberg_evolution}, the Heisenberg evolution $A(t)$ can be approximated by a strictly local $\tilde{A}(t)$, that is evolved by the true dynamics restricted to the set of vertices $\sf S$ within distance $L$: 
\begin{align}\label{eq:A=Alocal}
    \left|\langle A(t)B\rangle_{\rho_0}-\langle \tilde{A}(t)B\rangle_{\rho_0} \right|\le \norm{A(t)-\tilde{A}(t)} \norm{B}     \le \frac{c}{\mu v}\mathrm{e}^{\mu (vt-L)}.
\end{align}
Error $\epsilon$ is guaranteed by choosing $L$ to be sufficiently large:\footnote{A better error dependence is possible by using a stronger Lieb-Robinson bound with factorial decay.} \begin{equation}
    L=vt+\frac{1}{\mu}\log \frac{c}{\epsilon \mu v}.
\end{equation}
The expectation $\tr_{\sf S}[\rho_{\sf S}\tilde{A}(t)B]$ can be evaluated by standard linear-algebra manipulation (exact diagonalization and matrix multiplication) at cost $\poly(2^{(L^D)})$, which is the advertised result.
\end{proof}

\subsection{Quantum algorithms for Hamiltonian simulation}
\label{sec:Q_Ham_sim}
Now, suppose we want to simulate the expectation $\tr[\rho O(t)]$ on a quantum computer, given some initial state $\rho$ for a longer time $t$. The task boils down to \textbf{Hamiltonian simulation} \cite{Feynman82sim,lloyd1996universal}, that is, to approximate the true unitary evolution by a product of simple unitaries
\begin{align}
    U = \e^{-\ri H t} \approx V = g_1\cdots g_G.
\end{align}
A Hamiltonian simulation algorithm must achieve the desired \textbf{accuracy} with minimal \textbf{cost}. For simplicity, one may quantify the accuracy by the spectral norm of the difference
\begin{align}
    \norm{U-V} \le \epsilon 
\end{align}
which guarantees accurate simulation of any input state with any observable
$
    \tr[ O (U\rho U^\dagger - V\rho V^\dagger) ] \le 2 \epsilon \norm{O}.
$ The cost is often calculated in terms of the number of gates. Such Hamiltonian simulation algorithms have numerous applications in quantum chemistry \cite{mcardle2020quantum} and materials science~\cite{babbush2018low}.

In this review, we focus on how accurate simulations can be when $H$ itself is a spatially local Hamiltonian, as we can then apply a Lieb-Robinson bound to try and prove that the simulation can be done efficiently. Just as in the classical setting, it makes sense that evolution generated by a spatially local Hamiltonian should be approximated efficiently by only local gates within the Lieb-Robinson light cone. The technical question is how to patch local evolution together to approximate continuous time dynamics well.  After all, in a local circuit, there is an \emph{exact} light cone: information cannot be sent farther than the depth of the circuit (see Section \ref{sec:ruc}).

A first attempt by~\cite{Osborne_2006} is to cut the systems into non-interacting pieces, and then put the interaction back via the interaction picture, which we briefly review: 

\begin{prop}[The interaction picture]
Suppose the Hamiltonian consists of two terms $H = H_0+ V$. Then, 
\begin{align}
    \e^{\ri H t} = \CT\e^{ \ri\int^{t}_0 V(s) ds}\cdot \e^{\ri H_0 t} \quad \text{where}\quad V(s) := \e^{\ri H_0 s} V \e^{-\ri H_0 s}. 
\end{align}
Here $\mathcal{T}$ denotes the time-ordered exponential.
\end{prop}
Now, suppose our system is a spatially local Hamiltonian. We may take $V$ to be the interaction $H_{\sf B:\sf C}$ between two regions $\sf B, \sf C$ and $H_0$ to be $H_{\sf B}+H_{\sf C}$. Then, we expect the unitary $\CT\e^{ \ri\int^{t}_0 V(s) ds}$ to be a quasi-local (by Proposition~\ref{prop:local_Heisenberg_evolution}). One may iterate the above to cut the system into quasi-local patches, giving a quantum algorithm for simulation of spatially local Hamiltonian. The main issue with this approach is implementing the interaction picture, since the time-ordered integrals will be only quasi-local, and therefore expensive to simulate by brute force (such as directly using the Solovay-Kitaev algorithm~\cite{Dawson2005TheSA}, whose costs generally scale with the Hilbert space dimension).

The more recent HHKL algorithm~\cite{haah2020quantum} currently serves as the state-of-the-art method for quantum simulation of spatially local Hamiltonian dynamics. Its core idea is to split the unitary evolution using back-and-forth local evolutions that circumvent the explicit interaction picture, achieved via the following lemma:
\begin{lma}[Patching local evolution operators without the interaction picture]
\label{lem:patching}
For a local Hamiltonian supported on disjoint sets $\sf A, \sf B ,\sf C$, with $\mathsf{d}(\mathsf{A},\mathsf{B}) = \mathsf{d}(\mathsf{C},\mathsf{B}) =1 $ but $\mathsf{d}(\mathsf{A},\mathsf{C})>1$, and for constant time $t = \mathrm{O}(1)$, there is a constant $\mu >0$ such that
\begin{align}
    \norm{ \e^{\ri H_{\sf A\sf B\sf C} t } - \e^{\ri H_{\sf A\sf B} t}\e^{-\ri H_{\sf B} t}\e^{\ri H_{\sf B\sf C} t} } \le \CO\left( \e^{-\mu\mathsf{d}(\mathsf{A},\mathsf{C})}\norm{H_{\sf A:\sf B}} \labs{\text{Supp}(H_{\sf A:\sf B})}\right)
\end{align}
where $H_{\sf A\sf B}:= \overline{\mathbb{P}}_{\sf A \sf B}H$ and $H_{\sf A:\sf B} := H_{\sf A\sf B} - H_{\sf A} - H_B$ are defined  in~\eqref{eq:H_B} and~\eqref{eq:B:C}, and $\mathrm{Supp}(H_{\mathsf{A}:\mathsf{B}})$ is the set on which $H_{\mathsf{A}:\mathsf{B}})$ acts non-trivially.
\end{lma}
Indeed, if the Hamiltonian is commuting, the equality holds; in the non-commuting case, the error is exponentially small in the distance between region $\mathsf{A}$ and $\mathsf{C}$. The terms $\norm{H_{\sf A:\sf B}}$ and $ \labs{\text{Supp}(H_{\sf A:\sf B})}$ scale only with the surface area and merely contribute a polylogarithmic overhead for algorithmic cost.
\begin{proof}
We begin with an elementary identity
\begin{align}
    \e^{\ri H_{\sf A\sf B\sf C} t } - \e^{\ri H_{\sf A\sf B} t}\e^{-\ri H_{\sf B} t}\e^{\ri H_{\sf B\sf C} t} = \left(\e^{\ri H_{\sf A\sf B\sf C} t}\e^{-\ri H_{\sf B\sf C} t} -\e^{\ri H_{\sf A\sf B} t}\e^{-\ri H_{\sf B} t} \right) \cdot \e^{\ri H_{\sf B\sf C} t}.
\end{align}
The first term can be expanded in the interaction picture by isolating the terms cutting $\mathsf{A}$ and $\mathsf{B}$: 
\begin{subequations}\begin{align}
    \e^{\ri H_{\sf A\sf B\sf C} t}\e^{-\ri H_{\sf B\sf C} t} &= V \e^{\ri H_{\sf A} t}\e^{\ri H_{\sf B\sf C} t}\cdot \e^{-\ri H_{\sf B\sf C} t} = V \e^{\ri H_{\sf A} t}\\
    \e^{\ri H_{\sf A\sf B} t}\e^{-\ri H_{\sf B} t} &= V' \e^{\ri H_{\sf A} t}
\end{align}\end{subequations}
where 
\begin{subequations}\begin{align}
 V &:=\CT \exp\left( \int_0^t \mathrm{d} s\e^{\ri (H_{\sf A}+H_{\sf B\sf C}) s}H_{\sf A:\sf B} \e^{-\ri (H_{\sf A}+H_{\sf B\sf C}) s}\right)\\
 V' &:=  \CT\exp\left(\int_0^t \mathrm{d} s\e^{\ri (H_{\sf A}+H_{\sf B}) s}H_{\sf A:\sf B} \e^{-\ri (H_{\sf A}+H_{\sf B}) s}\right).
\end{align}\end{subequations}
Therefore,
\begin{align}
 \lnorm{\e^{\ri H_{\sf A\sf B\sf C} t } - \e^{\ri H_{\sf A\sf B} t}\e^{-\ri H_{\sf B} t}\e^{\ri H_{\sf B\sf C} t} }&\le \norm{ V - V'} \notag \\ 
 &\le \int_0^t \mathrm{d}s \lnorm{ \e^{\ri (H_{\sf A}+H_{\sf B\sf C}) s}H_{\sf A:\sf B} \e^{-\ri (H_{\sf A}+H_{\sf B\sf C}) s} - \e^{\ri (H_{\sf A}+H_{\sf B}) s}H_{\sf A:\sf B} \e^{-\ri( H_{\sf A}+H_{\sf B}) s}} \notag \\ &\le \mathrm{O}( \e^{-\mu \mathsf{d}(\mathsf{A},\mathsf{C})}\norm{H_{\sf A:\sf B}} \labs{\text{Supp}(H_{\sf A:\sf B})}).
\end{align}
The second inequality applies a telescoping sum over the time-ordered exponential. The second inequality uses Proposition~\ref{prop:local_Heisenberg_evolution} and integrates over constant time $t=\mathrm{O}(1)$. This is the advertised result.

\end{proof}
Recursively using the above gives the HHKL algorithm for d-dimensional lattices (Figure~\ref{fig:HHKL}). 
\begin{theor}[The HHKL algorithm~\cite{haah2020quantum}]
A spatially local Hamiltonian on a $L\times\cdots L = L^d$ lattice in $d$ spatial dimensions can be simulated for time $t$ up to $\epsilon$ error in spectral norm using 
\begin{align}
\text{gate complexity} \quad \mathrm{O}(t L^d \mathrm{polylog}(tL^d/\epsilon)) \quad \text{and depth}\quad \mathrm{O}(T \mathrm{polylog}(tL^d/\epsilon)).
\end{align}
\end{theor}
The gate complexity is essentially (up to the polylogarithmic corrections) the space-time volume of the evolution, coinciding with our physical intuition. A matching lower bound (up to polylogarithmic factors) is known by constructing a family of time-dependent circuits~\cite{haah2020quantum}. 

\begin{figure}[t]
\centering
\includegraphics[width=.9\textwidth]{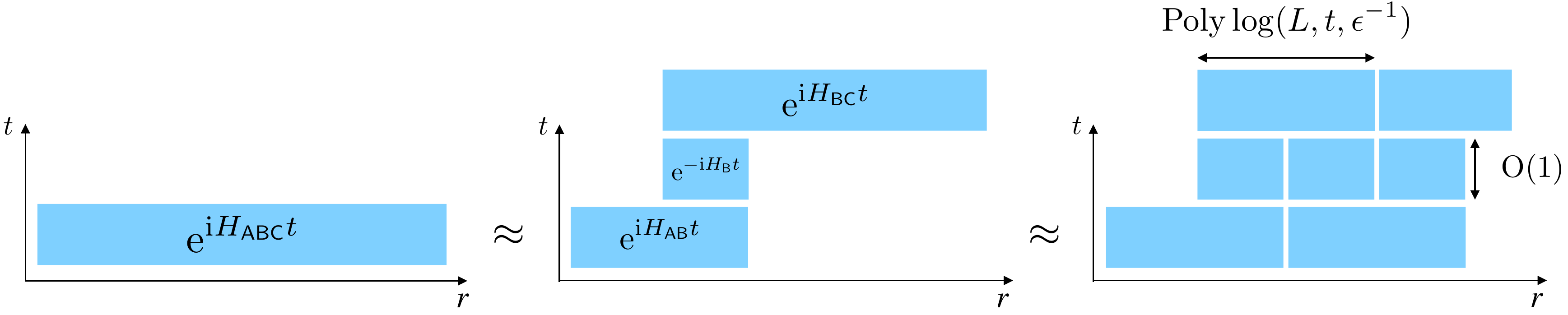}
\caption{
The HHKL decomposition of unitary in one dimension. For each short time $\mathrm{O}(1)$ unitary, the first approximation illustrates one call of Proposition~\ref{prop:local_Heisenberg_evolution}. 
Many iterations (the second approximation) lead to a decomposition in terms of quasi-local unitaries, which can be implemented at exponential precision using standard Hamiltonian simulation techniques (see, e.g.,~\cite{martyn2021grand}.) This strategy naturally extends to higher dimensions.
}\label{fig:HHKL}
\end{figure}

\section{Bounds on entanglement dynamics and correlations}\label{sec:corr}
Now, we use Lieb-Robinson bounds on commutators of local operators to derive (in many cases) optimal bounds on the speed with which various information-theoretic tasks (such as entanglement generation or quantum correlation/entanglement generation) can be performed. 

\subsection{Information signaling and quantum state transfer}
We start from perhaps the most directly relatable task: transmitting a qubit of quantum information across some distance in a many-body system.  A particularly explicit example is to perform state transfer during a single logical qubit initially stored on site $i$ is stored on site $f$ after the protocol: see (\ref{eq:Prop1 States}) below.  The following proposition shows that a quantum state can be transferred no faster than the Lieb-Robinson velocity $v$.

\begin{prop}[State transfer is bounded by the Lieb-Robinson Theorem] \label{prop:transfer norm}

Consider the initial state and final state of the form 
\begin{subequations}\label{eq:Prop1 States}\begin{align}
       \ket{\Psi_i(\alpha,\beta)}  &:=  \left( \alpha\ket{0}_i + \beta\ket{1}_i \right) \otimes \ket{\psi_{-i}}  ,\\
       \ket{ \Psi_f(\alpha,\beta)}  &:=  \ket{\psi_{-f}} \otimes \left( \alpha\ket{0}_f + \beta\ket{1}_f \right),
\end{align}
\end{subequations}
where $\ket{\psi_{-i}}$ and $\ket{\psi_{-f}}$ are both arbitrary states on all qubits except $i$ and $f$, respectively, and $\labs{\alpha}^2+\labs{\beta}^2  =  1$. Then, for any unitary $U$,
\begin{equation}
    U\ket{\Psi_i(\alpha,\beta)}  =  \ket{\Psi_f(\alpha,\beta)} \quad \text{implies}\quad
    \left\lVert  \comm{ U^\dagger X_f U  }{ Z_i} \right\rVert   = 2 .~~
    \label{eq:Prop1 Comm}
\end{equation}
\end{prop}

\begin{proof}
Applying the commutator \eqref{eq:Prop1 Comm} to the initial state \eqref{eq:Prop1 States} involves two parts: 
\begin{align}
\label{eq:Prop1 Comm parts}
    U^\dagger X_f U Z_i   \ket{\Psi_i(\alpha,\beta)}  &=  U^\dagger X_f U  \ket{\Psi_i(\alpha,-\beta)}  =  U^\dagger X_f  \ket{ \Psi_f(\alpha,-\beta)}   =  U^\dagger \ket{ \Psi_f(-\beta,\alpha)}  =  \ket{\Psi_i(-\beta,\alpha)}  \notag\\
    Z_i U^\dagger X_f U  \ket{\Psi_i(\alpha,\beta)}   &=  Z_i U^\dagger X_f  \ket{\Psi_f(\alpha,\beta)}  =  Z_i U^\dagger  \ket{ \Psi_f(\beta,\alpha)}   = Z_i  \ket{\Psi_i(\beta,\alpha)}   =  -\ket{ \Psi_i(-\beta,\alpha)}  ,
\end{align}
and subtracting the second line from the first gives the commutator
\begin{equation}
    \label{eq:Prop1 implication}
    \comm{ U^\dagger X_f U }{   Z_i }  \ket{\Psi_i(\alpha,\beta)}   =   2  \ket{ \Psi_i(-\beta,\alpha)}  
\end{equation}
which implies that $\lnorm{  \comm{ U^\dagger  X_f  U  }{  Z_i  } }  \geq  2$. On the other hand, 
\begin{equation}
    \lnorm{ \comm{ U^\dagger X_f U}{ Z_i} }  \leq  2 \lnorm{ U^\dagger X_f U}  \lnorm{ Z_i}  =  2,
\end{equation}
which follows from
\begin{equation}
    \lnorm{\comm{A}{B}}=\lnorm{AB-BA} \leq \lnorm{AB}+\lnorm{BA} \leq 2 \lnorm{A}  \lnorm{B}  .
\end{equation}
The upper and lower bounds on $ \lnorm{  \comm{ U^\dagger  X_f  U  }{  Z_i  } }$ imply (\ref{eq:Prop1 Comm}).
\end{proof}
This relates directly to Lieb-Robinson bounds since if the protocol $U$ came from continuous time evolution with some local time-dependent Hamiltonian $H(t)$, for time $t\lesssim r/v$ with $r$ the distance between $i$ and $f$, then a Lieb-Robinson bound will forbid (\ref{eq:Prop1 Comm}) from being true. 

By transmitting a qubit, an agent at $i$ can send a bit of classical message to an agent at $f$.  Thus, sending quantum information is no faster than sending classical information. On the other hand, one may wonder if sending classical information could be strictly faster, by some protocol that encodes the classical bit into a quantum state in some more complicated way and then sends it via quantum dynamics. The answer is no since classical information also propagates no faster than the Lieb-Robinson velocity. This is summarized by:
\begin{theor}[Information signaling bounded by Lieb-Robinson (informal version) \cite{Bravyi2006}]\label{prop:holevo}
    Suppose the Lieb-Robinson bound \eqref{eq:LR_exp} holds. Then information (both quantum and classical) travels at a speed upper bounded by the Lieb-Robinson velocity $v$. 
\end{theor}
Ref.~\cite{Bravyi2006} uses the Holevo capacity to quantify the classical information, which is beyond the scope of this review. Here we give an intuitive argument on why Theorem \ref{prop:holevo} should hold, generalizing the idea of Proposition \ref{prop:transfer norm}.

Consider Alice and Bob sitting at space-time points $(x_i,0)$ and $(x_f,t)$ respectively. The system starts in the state $\rho_0$ at time $0$ when Alice accesses the system locally at site $i$.  It then undergoes local dynamics via unitary $U$ until time $t$, when Bob tries to receive the information at site $f$. All correlations Bob can measure are captured by the reduced density matrix $\rho_f$ at site $f$: \begin{equation}
    \rho_f = \tr_{\lbrace f\rbrace^{\mathrm{c}}} \left[ U\rho_0 U^\dagger \right].
\end{equation}
If $\rho_f$ does not depend much on what Alice did at $(x_i,0)$, then Bob effectively cannot retrieve information. For example, suppose Alice has a bit of classical information $0$ or $1$ at time $0$, and she either does nothing to the system if the bit is $0$, or flips the spin (applies unitary $X_i$ at time 0) if the bit is $1$. Then the final state at $t$ is either $U\rho_0 U^\dagger$ or $UX_i\rho_0 X_i U^\dagger$ based on the classical bit of Alice. Although these two states may be drastically different globally, they are indistinguishable locally for Bob at $f$ if this site is far from $i$. Indeed, for any operator $B_f$ on $f$, its expectation value differs between the two states by an amount
\begin{equation}
    \tr\left[B_f\lr{U\rho_0 U^\dagger - UX_i\rho_0 X_i U^\dagger } \right] = \tr\left[\rho_0\lr{B_f(t)-X_i B_f(t) X_i}\right] \le \lnorm{\comm{B_f(t)}{X_i}}.
\end{equation}
Thus if $r\gtrsim vt$, the right hand side is vanishingly small for any $B_f$, meaning that the two states are ``close to each other locally" and Bob cannot distinguish them. In other words, in order to communicate information, the local ``perturbation'' by Alice's gate $X_i$ should be able to reach Bob at time $t$.  The Lieb-Robinson bound tells us how quickly that can happen.

Lieb-Robinson bounds are used to constrain information transfer on general spin networks in \cite{chessa}.

\subsection{Entanglement dynamics}

Consider a system made out of two subsystems $\mathsf{A}$ and $\sf B$ with Hilbert space $\mathcal{H}=\mathcal{H}_{\sf A}\otimes \mathcal{H}_{\sf B}$. For any unitary $U$ acting on $\mathcal{H}$, one can ask how much it can grow \textit{local operators} in one subsystem to the other. According to previous sections, one way to quantify this is to study $\lV \BP_{\sf B} \CU A \rV$, where $A$ is some operator supported in ${\sf A}$, and $\CU$ is the evolution superoperator \begin{equation}
    \CU A := U A U^\dagger.
\end{equation} 
Note that our choice is different than Heisenberg evolution $U^\dagger A U$ for later convenience. In this section, we ask how much \textit{entanglement} $U$ generates between the two parties ${\sf A}$ and ${\sf B}$, and connect this to a Lieb-Robinson bound.

First, to quantify entanglement, consider the \textbf{R\'enyi entropy} for any pure state $\ket{\psi}\in \mathcal{H}$ and any $ 0 \le \alpha \le +\infty$\begin{equation}
    S_\alpha(\ket{\psi}) := \frac{1}{1-\alpha} \ln \tr\rho_{\sf A}^\alpha\quad \text{where}\quad \rho_{\sf A}=\tr_{\sf B} |\psi\rangle\langle\psi|. \label{eq:renyi}
\end{equation}
In particular, we recover the \textbf{von Neumann entropy} at $\alpha\rightarrow 1$ \begin{equation}
    S_1(\ket{\psi}) := -\tr\left(\rho_{\sf A} \ln \rho_{\sf A}\right).
\end{equation}
The entropy $S_\alpha(\ket{\psi})$ is decreasing function of $\alpha$ such that \begin{equation}\label{eq:S1>S2}
     \alpha_1 \le \alpha_2 \quad \text{implies}\quad S_{\alpha_1}(\ket{\psi}) \ge S_{\alpha_2}(\ket{\psi}).
\end{equation}
In the following proposition, we show that the speed of generating the second R\'enyi entropy $S_2$ is bounded by operator growth. 
\begin{prop}[Operator growth bound and bipartite entanglement generation]\label{prop:P>S2}
Consider a unitary $U_{\sf A B}$ acting systems $\sf A B$ and an operator $A$ in the form 
\begin{equation}\label{eq:OA=pp}
    A = |\psi_{\sf A}\rangle \langle \psi_{\sf A}|\otimes I_{\mathsf{B}}.
\end{equation}
Then, we have 
\begin{equation}\label{eq:P>S2}
    \norm{\BP_{\sf B} \CU_{\sf A B} A } \ge 1-\e^{-S_2(\ket{\psi_f})/2} \quad \text{where} \quad \ket{\psi_f} := U_{\sf A B} \ket{\psi_{\sf A}}\otimes\ket{\psi_{\sf B}}
\end{equation}
for arbitrary $\ket{\psi_{\sf B}}$.
\end{prop}

As a result, $U_{\sf AB}$ cannot generate $\mathrm{O}(1)$ entanglement measured by the second R\'enyi entropy starting from any product state $\ket{\psi} = \ket{\psi_A}\otimes \ket{\psi_B}$, unless there is a local operator $A$ that grows sufficiently to the other party by $U_{\sf AB}$: $\norm{\BP_{\sf B} \CU_{\sf AB} A}=\mathrm{O}(1)$. Operator growth bounds then bound the generation of all R\'enyi entropies $S_\alpha$ with $\alpha\ge 2$, according to \eqref{eq:S1>S2}. The bound \eqref{eq:P>S2} is loose in the situation where $U_{\sf AB}$ is the SWAP operation between ${\sf A}$ and ${\sf B}$: No entanglement is generated although operators are moved around. In 
this case, one can bound $\norm{\BP_{\sf A}\BP_{\sf B} \CU_{\sf A B} A }$ instead, which we leave as an exercise. Namely, in order to generate entanglement, a local operator needs to become nonlocal instead of just swapping to the other subsystem.

\begin{proof}
For any bipartite system $\sf A, \sf B$ with dimensions $D_{\sf A},D_{\sf B}$, consider the Schmidt decomposition of the final state $\ket{\psi_f}$ \begin{equation}
    |\psi_f\rangle = \sum_{j=1}^{\min(D_{\sf A},D_{\sf B})} \sqrt{p_j} |j\rangle_{\sf A} \otimes |j\rangle_{\sf B}\quad \text{where}\quad \sum^{\min(D_{\sf A},D_{\sf B})}_{j=1} p_j =1
\end{equation}
for some orthonormal basis $\{|j\rangle\}_{\sf A}, \{|j\rangle\}_{\sf B}$ of ${\sf A}$ and ${\sf B}$. In this basis, we expand the evolved operator as \begin{equation}\label{eq:515}
    \CU_{\sf AB} A = \BP_{\sf B} \CU_{\sf AB} A +  \sum_{ i,j=1 }^{D_{\sf A}} T_{ij} \ket{i}\bra{j}_{\sA} \otimes I_{\sf B} \quad \text{where}\quad T_{ij} = \frac{1}{D_{\sf B}} \bra{i}_{\sf A} \tr_{\sB} (\CU_{\sf AB} A) \ket{j}_{\sf A}.
\end{equation}
Then, rearrange and take the operator norm to obtain
\begin{align}
    \norm{\BP_{\sf B} \CU_{\sf AB} A} \ge \langle \psi_f | \BP_{\sf B}\CU_{\sf AB} A |\psi_f\rangle 
    \ge 1 - \sum_{j=1}^{D_{\sf A}} p_j T_{jj}
     \ge 1 - \sqrt{\sum_{j=1}^{D_{\sf A}} p_j^2}\sqrt{\sum_{j=1}^{D_{\sf A}} \labs{T_{jj}}^2}
    \ge 1-\e^{-S_2(\ket{\psi_f})/2}.
\end{align}
The second inequality uses $\CU_{\sf AB} A \ket{\psi_f} = \ket{\psi_f}$. The third is Cauchy-Schwartz. The last inequality uses the definition of second R\'enyi entropy~\eqref{eq:renyi} and that
\begin{align}
\sum_{j=1}^{D_{\sf A}} \labs{T_{jj}}^2 \le \sum_{ i,j=1 }^{D_{\sf A}} \labs{T_{ij}}^2 \le D_{\sf A}\norm{\sum_{ i,j=1 }^{D_{\sf A}} T_{ij}\ket{i}\bra{j}_{\sA} \otimes I_{\sf B}}_{\mathrm{F}}^2 \le D_{\sf A} \norm{ \CU_{\sf AB} A}_{\rm F}^2=D_{\sf A} \norm{ A}_{\rm F}^2 = 1,
\end{align}
which concludes the proof.

\end{proof}

\subsubsection{von Neumann entanglement outside the Lieb-Robinson light cone}\label{sec:entanglement_dyn}
There are two other potential directions to improve Proposition \ref{prop:P>S2}. Are $\alpha<2$ R\'enyi entropies also bounded by operator growth?  Do tighter measures of operator growth, like Frobenius norm $\norm{\BP_{\sf B} \CU_{\sf AB} A}_{\rm F}$, also bound entanglement generation? We give partial negative answers to these questions using explicit counterexamples. 

\begin{exam}[Generating large von Neumann entanglement with little operator growth]\label{ex:S1>}
Consider two systems $\sf A, \sf B$ with dimension  $\dim {\sf A}=\dim {\sf B}=D$ and local basis $\{|j\rangle\}^{D-1}_0$. Consider the unitary $U_{\sf AB}=U_2\oplus I_{D^2-2}$, where the nontrivial part $U_2$ is a $2\times 2$ matrix \begin{equation}\label{eq:U2=e}
    U_2 = \left(\begin{array}{cc}
        \sqrt{1-\epsilon} &  -\sqrt{\epsilon}\\
        \sqrt{\epsilon} & \sqrt{1-\epsilon}
    \end{array}\right) \quad \text{acting on}\quad \Span{|00\rangle, |\mathrm{diag}\rangle}.
\end{equation} We require that $0<\epsilon<1$ and define the ``diagonal'' state to be \begin{equation}
    |\mathrm{diag}\rangle = \frac{1}{\sqrt{D-1}} \sum_{j=1}^{D-1} |jj\rangle.
\end{equation}
Then, at large $D \gg 1$ and at a constantly small $\epsilon >0$, the unitary $U_{\sf AB}$ generates arbitary large von Neumann entropy \begin{equation}\label{eq:S>lnN}
     S_1(U_{\sf AB}|00\rangle)=\mathrm{\Omega}(\epsilon \ln D)
\end{equation}
yet for any local operator $A$, \begin{equation}\label{eq:P<roote}
    \norm{\BP_{\sf B} \CU_{\sf AB} A} \le \norm{A} \mathrm{O}(\sqrt{\epsilon}).
\end{equation}
\end{exam}
Simply put, the diagonal state $\ket{\text{diag}}$ has lots of von Neumann entanglement despite being one-dimensional.
\begin{proof}
To prove \eqref{eq:P<roote}, observe that $U_{\sf AB}$ is close to identity $I$: $\norm{U_{\sf AB}-I} \le \mathrm{O}(\sqrt{\epsilon})$. Therefore, \begin{align}
    \norm{\BP_{\sf B} \CU_{\sf AB} A} = \norm{\BP_{\sf B} (U_{\sf AB} AU_{\sf AB}^\dagger-A)} &\le \norm{U_{\sf AB}  A U_{\sf AB}^\dagger- A } \nonumber\\ &\le \norm{(U_{\sf AB}-I) A  U_{\sf AB}^\dagger} + \norm{ A  (U_{\sf AB}^\dagger-I)}  \le \norm{ A } \mathrm{O}(\sqrt{\epsilon}).
\end{align}
But \eqref{eq:U2=e} leads to \begin{align}
    U_{\sf AB}|00\rangle &= \sqrt{1-\epsilon}|00\rangle + \sqrt{\frac{\epsilon}{D-1}} \sum_{j=1}^{D-1} |jj\rangle, \nonumber\\ S_1(U_{\sf AB}|00\rangle) &= -(1-\epsilon)\ln(1-\epsilon) + \epsilon \ln \frac{D-1}{\epsilon}=\mathrm{\Omega}(\epsilon \ln D)
\end{align}
which verifies \eqref{eq:S>lnN}.
\end{proof}

Given the above example, as well as Example \ref{exam:frobentangle}, any improvement of Proposition \ref{prop:P>S2} should only involve $S_\alpha$ with $1<\alpha\le 2$, and operator $p$-norms with $p>2$.  Such a generalization seems to be an open problem.  See also \cite{Bentsen:2018uph,harrow2021} for further remarks on discrepancies between operator growth and entanglement generation.

\subsubsection{Entanglement generation bounds from interaction strength}
\label{sec:entanglement_resource}
In this section, we briefly discuss entanglement generation bounds \textit{independently} of Lieb-Robinson bounds, as these are of great importance on their own. Suppose systems $\mathsf{A}$ and $\mathsf{B}$ have limited bipartite interaction. Can we bound the rate of bipartite entanglement growth? Here, the spatial locality is not the pronounced structure but rather the bipartition of the system.

More precisely, suppose we begin with a pure state $\ket{\psi}_{\mathsf{AB}}$, and we are free to operate any unitary operator on either set $\mathsf{A}$ or $\mathsf{B}$. The goal is to bound the growth rate of entanglement entropy given bipartite interactions $H_{\mathsf{AB}}$. The answer to this question is thoroughly addressed in a series of works \cite{Bravyi2006,vanacoleyen,Vershynina2018EntanglementRF}; here, we provide an elementary argument and state the sharpest result. A technical difficulty is that the sets \textsf{A} and \textsf{B} can have arbitrarily large dimensions and arbitrary entanglement structure; a resource theoretical approach appears natural for addressing this.
\begin{prop}[Entanglement cost of bipartite Hamiltonian evolution {\cite{Cirac_2001_2qubit}}]
    Any 2-qubit Hamiltonian evolution $\e^{\ri H }$ can be implemented using $\CO(\norm{H})$- bits of bipartite entanglement entropy.
\end{prop}
Further, since bipartite entanglement is non-increasing under local operations and classical communication (LOCC) (see Section \ref{sec:measurement}), internal dynamics in \textsf{A} and \textsf{B} cannot increase the entanglement further:
\begin{corol}[Limited interaction means limited entanglement]
Any 2-qubit Hamiltonian evolution $\e^{\ri H_{ij} }$ acting on a bipartite pure state $\ket{\psi}_{\sf A \sf B}$ for $i \in \mathsf{A}, j\in \mathsf{B}$ can at most generate $\CO(\norm{H_{ij}})$ bits of bipartite entanglement entropy between systems $\mathsf{A},\mathsf{B}$.
\end{corol}
\begin{prop}[Entanglement rate of bipartite Hamiltonians {\cite{Audenaert2013QuantumSD,Shrimali:2022bvt}}]\label{prop:entanglement_rate_general}
For any systems $\mathsf{A},\mathsf{B}$ and for all initial pure states, suppose the global Hamiltonian takes the form $H_{\mathsf{ab}}+ H_{\mathsf{A}} +H_{\mathsf{B}}$ where $\mathsf{a} \subset \mathsf{A}, \mathsf{b} \subset \mathsf{B}$. Then, the bipartite von Neumann entanglement rate is bounded
by
\begin{align}
    \frac{\mathrm{d}S_1}{\mathrm{d}t}  \le 8 \norm{H_{\mathsf{ab}}} \log [\min (\dim(\mathsf{A}),\dim(\mathsf{B}))].
\end{align}
\end{prop}
See Section~\ref{sec:power-law} for an application to the entanglement rate in power-law interacting systems.
\subsection{Connected correlation functions}
\label{sec:connected}
Suppose we start from a ``short-range correlated'' state $|\psi\rangle$ such as a product state on all qudits. How long does it take for two remote regions ${\sf A}$ and ${\sf B}$ to become correlated? Unlike Proposition \ref{prop:P>S2}, the correlation between the two parties ${\sf A}$ and ${\sf B}$ is not bounded by the commutator quantity $C_{{\sf A}{\sf B}}$: operators in one party do not need to grow to the other in order to build up correlation. In fact, ${\sf A}$ and ${\sf B}$ can be maximally entangled (by sharing $|{\sf A}|=|{\sf B}|$ pairs of Bell states) even if $C_{{\sf A}{\sf B}}=0$! As one example, in a 1d spin chain with ${\sf A}=\{1\}$ and ${\sf B}=\{2L\}$, consider a protocol that first locally prepares a Bell pair on sites $\{L,L+1\}$ out of a product state, and then transfers the two qubits left and right to $1$ and $2L$ respectively (using SWAP gates for example). In the final state, $1$ and $2L$ share a Bell pair and are maximally entangled. However, a local operator $A_1$, after the Heisenberg evolution that is backward in time, only grows ``halfway'' to site $L+1$. Similarly, $B_{2L}$ only extends to $L$.

In the above example, two regions of distance $2L$ can be correlated after time $t\approx L/v$. This suggests that the ``correlation speed'' is bounded by $2v$ instead of $v$. Indeed, this will be proven in Theorem \ref{thm:cor<2v}. Prior to that, we need to first quantify a useful notion of correlation. 

\begin{prop}[Connected correlations are bounded]\label{prop:cor}
Define the connected correlation function for a state $\ket{\psi}$ \begin{equation}\label{eq:cor}
    \cor_\psi({\sf A},{\sf B}) := \max_{\norm{A}, \norm{B}\le 1} \alr{A B}_\psi - \alr{ A }_\psi \alr{B}_\psi, 
\end{equation}
where $A,B$ are Hermitian operators acting on systems ${\sf A}$ and ${\sf B}$.
Then, the connected correlation satisfies \begin{equation}
    0\le \cor_\psi({\sf A},{\sf B}) \le 1.
\end{equation}
\end{prop}
Indeed, if $\psi$ is a product state between ${\sf A}$ and ${\sf B}$ in the sense that $\tr_{({\sf A}\cup {\sf B})^c} \ket{\psi}\bra{\psi} = \rho_{\sf A}\otimes \tilde{\rho}_{\sf B}$, then $\cor_\psi({\sf A},{\sf B})=0$; otherwise, if ${\sf A}$ and ${\sf B}$ are maximally correlated by a Bell state $\psi$, then $\cor_\psi({\sf A},{\sf B})=1$. More generally, $\cor_\psi({\sf A},{\sf B})$ measures how correlated ${\sf A}$ and ${\sf B}$ are, including both classical and quantum-mechanical correlation. 

\begin{proof}
The first inequality $\cor_\psi({\sf A},{\sf B})\ge 0$ is trivial by choosing $A=B=0$ and that the maximum must be larger. To prove $\cor_\psi({\sf A},{\sf B}) \le 1$, it suffices to show \begin{equation}\label{eq:AB-AB<1}
    \alr{A B}_\psi - \alr{ A }_\psi \alr{B}_\psi \le 1
\end{equation}
for any Hermitian operators $A,B$ obeying $\norm{A}, \norm{B}\le 1$. To see this, consider Hermitian operators \begin{equation}
    \tilde{A} = A - \alr{ A }_\psi, \quad \tilde{B} = B- \alr{ B }_\psi
\end{equation}

Then, we get \begin{align}
    \lr{\alr{A B}_\psi - \alr{ A }_\psi \alr{B}_\psi}^2 = \lr{\alr{\tilde{A}\tilde{B}}_\psi}^2 &\le \alr{\tilde{A}^2}_\psi \alr{\tilde{B}^2}_\psi \nonumber\\ &= \lr{\alr{A^2}_\psi-\alr{A}_\psi^2} \lr{ \alr{B^2}_\psi-\alr{B}_\psi^2 } \le \alr{A^2}_\psi\alr{B^2}_\psi \le 1.
\end{align}
The first inequality is Cauchy-Schwatz $\labs{\braket{\psi_2|\psi_1}}^2 \le \labs{\braket{\psi_2|\psi_2}}^2\labs{\braket{\psi_1|\psi_1}}^2$. This implies \eqref{eq:AB-AB<1}, which concludes the proof.
\end{proof}

\begin{theor}[Bounds on correlation generation \cite{Bravyi2006}]
\label{thm:cor<2v}
Suppose the initial state $\ket{\psi}$ has a finite correlation length $\xi$ defined by \begin{equation}\label{eq:cor<xi}
    \cor_\psi({\sf A},{\sf B}) \le \tilde{c}\lr{\labs{\partial \sf A}+\labs{\partial \sf B}} \e^{-d({\sf A},{\sf B})/\xi}, \quad \forall {\sf A},{\sf B}.
\end{equation}
If the Lieb-Robinson bound \eqref{eq:LR_exp} holds with constant $\mu, v$, then after time $t$, the final state $\ket{\psi_f}=U\ket{\psi}$ roughly has correlation length $2vt+\xi$. More precisely, for any two subsets ${\sf A},{\sf B}$ with \begin{equation}\label{eq:d>2vt}
    \mathsf{d}({\sf A},{\sf B})>2vt,
\end{equation}
the connected correlation is exponentially suppressed \begin{equation}\label{eq:cor_f<}
    \cor_{\psi_f}({\sf A},{\sf B}) \le (c+\tilde{c}) \lr{\labs{\partial \sf A}+\labs{\partial \sf B}} \exp\lr{-\frac{\mathsf{d}({\sf A},{\sf B})-2vt}{\xi+2\mu^{-1}} }.
\end{equation}
\end{theor}

\begin{proof}
In the Heisenberg picture, correlation in the final state is equivalent to correlation of evolved operators in the initial state, namely \begin{equation}\label{eq:cor_f=cor_i}
    \cor_{\psi_f}({\sf A},{\sf B}) = \max_{\norm{A}, \norm{B}\le 1} \alr{A B}_{\psi_f} - \alr{ A }_{\psi_f} \alr{B}_{\psi_f} = \max_{\norm{A}, \norm{B}\le 1} \alr{A(t) B(t)}_{\psi} - \alr{ A(t) }_{\psi} \alr{B(t)}_{\psi}.
\end{equation}
Define $\mathsf{R}({\sf A},r)=\{i\in\mathsf{V}: \mathsf{d}(i,{\sf A})\le r\}$ (and similarly for $\mathsf{R}({\sf B},r)$) with a tunable parameter $r<d({\sf A},{\sf B})/2$. According to \eqref{eq:LR_exp} and Proposition \ref{prop:tildA}, there exists an operator $\tilde{A}$ supported in $\mathsf{R}({\sf A},r)$ (and $\tilde{B}$ in $\mathsf{R}({\sf B},r)$ such that \begin{equation}\label{eq:At-A<c}
    \lnorm{A(t)-\tilde{A}} \le c\labs{\partial \sf A} \e^{-\mu(r-vt)}, \quad \lnorm{B(t)-\tilde{B}} \le c\labs{\partial \sf B} \e^{-\mu(r-vt)}.
\end{equation}
Then, the correlation function in \eqref{eq:cor_f=cor_i} is \begin{align}\label{eq:cor<mu+xi}
    \alr{A(t) B(t)}_{\psi} - \alr{ A(t) }_{\psi} \alr{B(t)}_{\psi} &= \alr{A(t) [B(t)-\tilde{B}]}_{\psi} - \alr{ A(t) }_{\psi} \alr{B(t)-\tilde{B}}_{\psi} \nonumber\\ &\qquad +  \alr{[A(t)-\tilde{A}]\tilde{B}}_{\psi} - \alr{ A(t)-\tilde{A} }_{\psi} \alr{\tilde{B}}_{\psi} + \alr{\tilde{A}\tilde{B}}_{\psi} - \alr{ \tilde{A} }_{\psi} \alr{\tilde{B}}_{\psi} \nonumber\\
    &\le \lnorm{B(t)-\tilde{B}} + \lnorm{A(t)-\tilde{A}} \lnorm{\tilde{B} } + \cor_\psi(\mathsf{R}({\sf A},r),\mathsf{R}({\sf B},r)) \nonumber\\
    &\le c \lr{\labs{\partial \sf A}+\labs{\partial \sf B}} \e^{-\mu(r-vt)} + \tilde{c}\lr{\labs{\partial \sf A}+\labs{\partial \sf B}} \e^{-\frac{\mathsf{d}({\sf A},{\sf B})-2r}{\xi}}.
\end{align}
The third line uses~\eqref{eq:AB-AB<1} and the fact that correlation is proportional to the norms of each operator, due to the linearity of Cor and Proposition \ref{prop:cor}. The last line uses \begin{equation}
    \lnorm{\tilde{B}}=\lnorm{ B(t)-\mathbb{P}_{\mathsf{R}(\mathsf{B},r)}B(t)}\le \lnorm{B(t)}=1
\end{equation}
together with \eqref{eq:At-A<c} and \eqref{eq:cor<xi}. We choose $r$ such that the two terms in \eqref{eq:cor<mu+xi} are comparable: \begin{equation}
    r= \frac{\mathsf{d}({\sf A},{\sf B})+\xi\mu vt}{2+\xi\mu}\in \lr{vt, \frac{1}{2}\mathsf{d}({\sf A},{\sf B}) },
\end{equation}
which implies that \begin{equation}
    \mu\cdot (r-vt) = \frac{\mathsf{d}({\sf A},{\sf B})-2r}{\xi}
\end{equation}
using \eqref{eq:d>2vt}. We plug $r$ into \eqref{eq:cor<mu+xi} and take maximum over $\mathsf{A},\mathsf{B}$ in \eqref{eq:cor_f=cor_i} to conclude the proof.
\end{proof}

\subsection{Measurement-enhanced protocols}\label{sec:measurement}
Theorem \ref{thm:cor<2v} bounds the total correlation generated between two faraway regions due to time evolution. This correlation can be either quantum or classical; indeed quantum correlation is also generated by the Bell pair preparation protocol discussed (Section~\ref{sec:connected}). 

The distinction between classical and quantum correlation becomes extremely important, however, when one considers quantum dynamics with local measurements and active feedback.  Here, \emph{classical} information can propagate through the ``experimentalist'' who performs projective measurements and applies local unitaries based on those measurements.  Since the experimentalist may only be limited by Einstein's speed of light $c\rightarrow \infty$, we should effectively consider this communication of classical information to be instantaneous.  Does this classical communication, combined with projective measurement, allow us to beat the Lieb-Robinson bound?

It may seem that the answer is obviously yes.  For example, after a measurement done at site $1$, one can immediately apply a gate to site $L\gg 1$ that depends on the previous measurement outcome. Moreover, a local measurement is able to collapse the global quantum state, and one can apply gates adaptively based on outcomes of arbitrarily faraway measurements. 

However, the situation is exactly like the EPR paradox discussed in the Introduction. If no measurement outcomes are used for feedback, local measurements are just local completely positive trace-preserving (CPTP) maps that do not propagate information. As discussed in Theorem \ref{thm:open_LRB}, Lieb-Robinson bounds also hold in open quantum systems \cite{open_LRB}; one can interpret an open quantum system as one that is measured, while the measurement outcome is discarded and averaged over.

What if measurement outcomes are used to adjust the unitary dynamics adaptively? For concreteness, consider the task of teleporting a quantum state from one end (site $1$) to the other (site $N=L+1$) in a 1d spin chain. First, one should avoid directly measuring site $1$ because the state would collapse and destroy the quantum information. Instead, site $1$ needs to interact with its neighbors before some measurement is done. Thus the task depends on what the initial state is on sites $2,\cdots,N$.  If the sites $2$ and $N$ share a Bell pair, one can perform the following standard teleportation protocol \cite{teleport93} that takes ${\rm O}(1)$ operations.

\begin{exam}[Standard teleportation protocol]\label{exam:teleport}

Suppose the initial state on the three relevant qubits $1,2,N$ is \begin{align}
    \ket{\psi} &= (\alpha\ket{0} + \beta \ket{1})_1\otimes \frac{1}{\sqrt{2}} (\ket{00}+\ket{11})_{2N} \nonumber\\ &= \frac{1}{2} (\ket{00}+\ket{11})_{12} \otimes (\alpha\ket{0} + \beta \ket{1})_N + \frac{1}{2} (\ket{01}+\ket{10})_{12} \otimes (\alpha\ket{1} + \beta \ket{0})_N \nonumber\\
    &+ \frac{1}{2} (\ket{01}-\ket{10})_{12} \otimes (\alpha\ket{1} - \beta \ket{0})_N + \frac{1}{2} (\ket{00}-\ket{11})_{12} \otimes (\alpha\ket{0} - \beta \ket{1})_N,
\end{align}
where we have expanded in the Bell-basis on $1$ and $2$. The first step of this protocol is to measure $1$ and $2$ in the Bell basis, so that the state collapses to one of the four terms above.  This can be done by measuring the commuting operators $X_1X_2$ and $Z_1Z_2$. If the measurement outcome is $\frac{1}{\sqrt{2}} (\ket{00}+\ket{11})_{12}$, and this classical information is transferred to an agent Bob at $N$, then Bob knows the state is already the honest qubit Alice at site $1$ wants to transfer. If Bob knows that the measurement outcome is $\frac{1}{\sqrt{2}} (\ket{01}+\ket{10})_{12}$ instead, he can use an $X$ gate on $N$ to ``correct'' the state because $X(\alpha\ket{1} + \beta \ket{0})_N = (\alpha\ket{0} + \beta \ket{1})_N$. This correctability holds for the other two outcomes as well, so based on the transferred classical information about the measurement outcome, the state is honestly transferred from $1$ to $N$ deterministically, \emph{after error correction has been applied}.

A long-range Bell pair is consumed in the above process, which is itself hard to generate if the initial state is a product state. Starting from short-range entangled states and/or product states, it is then expected that teleportation requires large resources that scale with $L$.
\end{exam}

As it turns out, however, such Bell pairs can be efficiently prepared using measurement-enhanced teleportation protocols that operate in constant time. This is, for example, behind the theory of measurement-based quantum computation \cite{raussendorf,jozsa,DominicMBQC_SPT,raussendorf19}. Applying the identity gate on qubit on $1$, in MBQC, amounts to teleporting qubit 1 to $L$ by pure measurements.  MBQC is, in its simplest avatar, based on the cluster state which can be prepared in constant depth.  An easier version of this idea to understand is the 
quantum repeater \cite{repeater98}, or entanglement-swapping teleportation
protocol (ESTP) \cite{Friedman:2022vqb}, which teleports a qubit to distance \begin{equation}\label{eq:L=2MT}
    L\approx (2M+1)T,
\end{equation}  
using $T$ layers of unitary gates and $M$ local measurements. Fig.~\ref{fig:SWAP circuit}(a) gives an example of $L=15,T=5,M=2$, where the spin chain is divided into $M+1=3$ parts, each of length roughly $L/3$. In the leftmost part, the quantum state $\ket{\psi}$ is simply transported by SWAP gates. In each of the other parts, a Bell pair in the middle is generated and then transported to the two ends of the part by SWAP gates. Then a Bell-basis measurement is performed in each shaded area that connects adjacent parts, and all outcomes are collected to decide the error-correction unitary $\mathcal{R}$ that recovers $\ket{\psi}$ at the rightmost site. Conceptually, one can think of this as a cascade of standard teleportation protocols in Example \ref{exam:teleport}, which first transfers $\ket{\psi}$ from site $A_1$ to $A_2$ by measuring $A_1,B_1$, and then from site $A_2$ to the final site by measuring $A_2,B_2$.

\begin{figure}[t]
\centering
\includegraphics[width=.98\textwidth]{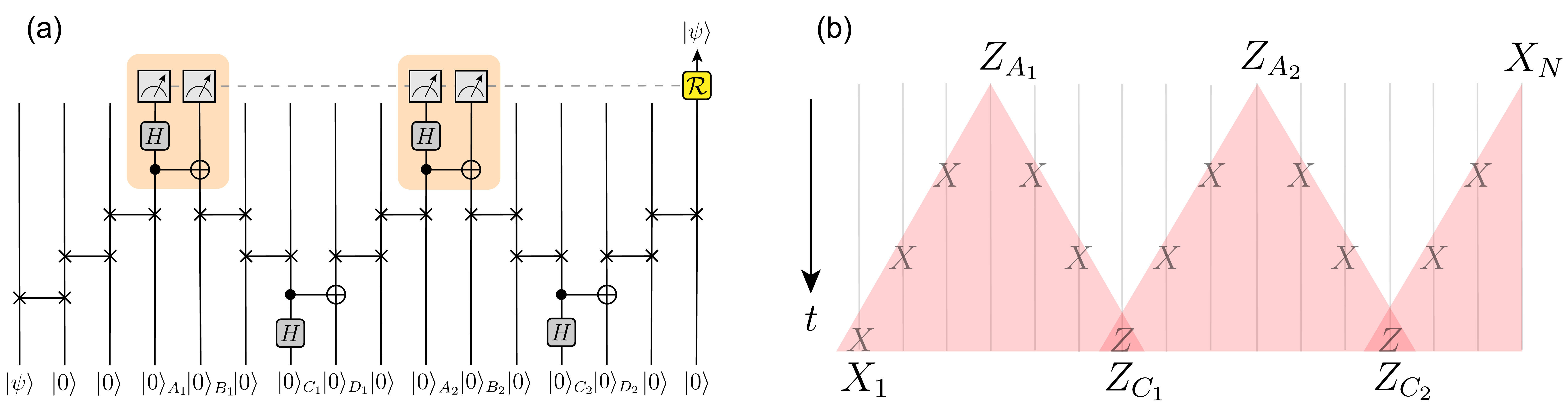}
\caption{(a) Circuit design for the entanglement-swapping teleportation protocol (ESTP), illustrated for $L=15$ using a two-local Clifford circuit depth $T=5$ and $M=2$ two-qubit Bell measurements. Bell pairs are generated on $C$ and $D$ qubits via a Hadamard--CNOT sequence, and transported to $A$ and $B$ qubits via SWAP gates.  The shaded areas indicate the standard teleportation protocol and include $Z$ measurements; the dashed line denotes classical communication. The logical qubit $\ket{\psi}$ starts at $j=1$ and teleports to the rightmost site after applying the error-correction gate $\mathcal{R}$, which is determined by the measurement outcomes. (b) Heisenberg evolution of the 
final logical operator $X_N$ for the ESTP depicted in (a). The local $Z_{A_1},Z_{A_2}$ operators are obtained by the measurement and error-correction procedure. Each of them grows to a product of two $X$s with Lieb-Robinson velocity $1$, as depicted by the edges of the shaded cones. When the light cones overlap, the whole operator becomes $X_1Z_{C_1}Z_{C_2}$, which is an initial logical operator. Figure taken with permission from \cite{Friedman:2022vqb}.
}
\label{fig:SWAP circuit}
\end{figure}

One can use \eqref{eq:L=2MT} to make a tradeoff between unitary dynamics and measurements in a quantum teleportation protocol. Is \eqref{eq:L=2MT} the best one can achieve in all possible protocols with $M$ measurements and unitary dynamics of time $T$? \cite{Friedman:2022vqb} gives a positive answer by extending Lieb-Robinson bounds to this setting of quantum dynamics with measurements. The idea comes from examining the ESTP in the operator language. Fig.~\ref{fig:SWAP circuit}(b) shows how the final logical operator $X_N$ is evolved in the Heisenberg picture (backward in time). Since the protocol can be schematically written as $\mathcal{W}=\CR\CM U $, which does measurements $\CM$ and applies an adaptive gate $\CR$ in the end, $X_N$ is first evolved to $\CM^\dagger \CR^\dagger X_N \CR\CM=X_N Z_{A_1}Z_{A_2}$ acting on the measurement sites $A_1,A_2$. Then these two ``seeds'' together with $X_N$, will grow into light cones due to the circuit dynamics $U$. In order to teleport the quantum information, the evolved operators $\mathcal{W}^\dagger X_N \mathcal{W}=U^\dagger \CM^\dagger \CR^\dagger X_N \CR\CM U$ and $\mathcal{W}^\dagger Z_N \mathcal{W}$ need to commute on all sites except $1$. This turns out to require that the light cones not only need to touch site $1$, but also need to overlap with their neighboring light cones, which makes \eqref{eq:L=2MT} optimal. This is formalized by the following Theorem.

\begin{theor}[Speed limit in quantum dynamics with measurements (informal version)]\label{thm:L<MT}
Consider a teleportation protocol that starts from a product state of all qubits, and teleports a qubit to distance $L$ using measurements in $M$ local regions and unitary dynamics of time $T$. The unitary dynamics is generated by a time-dependent Hamiltonian $H(t)$ that may depend on previous measurement outcomes. The measurement regions are also allowed to be adaptive. If the pure unitary dynamics generated by $H(t)$ has Lieb-Robinson velocity $v$, then there exist constants $M_0,T_0$ that do not depend on $L,M,T$, such that \begin{equation}\label{eq:L<2MT}
    L\le v(2M+M_0)(T+T_0).
\end{equation}
\end{theor}

Although we refer to \cite{Friedman:2022vqb} for the detailed proof, we would like to mention one key idea called Stinespring dilation \cite{Stinespring,ChoisThm}. To be specific, any quantum channel, like the dynamics with measurement and feedback, is equivalent to a unitary channel on a ``dilated'' Hilbert space $\mathcal{H}_{\mathrm{dil}}=\mathcal{H}_{\mathrm{phys}}\otimes \mathcal{H}_{\mathrm{ss}}$. Here $\mathcal{H}_{\mathrm{phys}}$ is the physical Hilbert space, while $\mathcal{H}_{\mathrm{ss}}$ is composed of ancilla Stinespring qubits that record the measurement outcomes. This justifies our previous notations like $\mathcal{W}=\CR\CM U$. The above Theorem then comes from applying the Lieb-Robinson methods to this dilated unitary dynamics.   The theorem makes clear that not only the act of measurement, \emph{but also subsequent error correction}, is \emph{required} to beat a Lieb-Robinson bound; see also \cite{AaronMIPT}.

Theorem \ref{thm:L<MT} also has various generalizations and applications \cite{Friedman:2022vqb}. To name a few, the initial state can be generalized to all states with short-range entanglement. If one wants to teleport $Q\gg 1$ qubits to distance $L$ with Hamiltonian dynamics of time $T$, although $M$ measurement \emph{regions} satisfying \eqref{eq:L<2MT} suffice, the number of measurements \emph{done} $M'$ should scale with $Q$ such that \begin{equation}
    L\lesssim 2vT\lr{\frac{M'}{Q}+1},
\end{equation} 
(with a technical caveat that measurement locations do not depend on measurement outcomes).  
Beyond quantum state transfer, \eqref{eq:L<2MT} also bounds the resources needed for preparing long-range entangled states, including long-range Bell pairs, error correcting code states  \cite{Gottesman:1997zz,terhalrmp}, the GHZ state (\ref{eq:GHZ}), W state (\ref{eq:W}), and spin-squeezed states \cite{Ma_2011,sensing_rmp}.  This theorem also has strong implications on the ease with which many tasks, such as preparing highly entangled quantum states, can be achieved using hybrid protocols involving both unitary dynamics and measurement.  This has been a subject of intense recent interest \cite{eldredge20,cirac21,Verresen:2021wdv,Devulapalli:2022xty,Lu:2022jax,Tantivasadakarn:2022hgp,Iqbal:2023shx,Foss-Feig:2023uew}.

\section{Ground states of gapped systems}\label{sec:gap}
So far, we have only applied the Lieb-Robinson bounds to dynamics. Remarkably, we will see that this \textit{temporal} bound also implies \textit{spatial} bounds for properties of gapped systems. Historically, analyzing the ground states of gapped systems using Lieb-Robinson bounds \cite{Hastings_koma,nachtergaele06,LSM04,Hastings_markov04} was what popularized the Lieb-Robinson Theorem in the broader physics community.

This section assumes that the lattice Hamiltonian $H$ is time-independent, with a Lieb-Robinson bound given by (\ref{eq:LR_exp}).
Without loss of generality, we also assume the Hamiltonian is non-negative $H\ge 0$ with ground energy zero;
the ground subspace projector $P_0$ satisfies $HP_0 = 0$. As a crucial assumption, we impose the existence of a spectral gap $\Delta>0$ above the ground states.

\subsection{Exponential clustering of gapped ground states}
We define the connected correlation for any ground state $\ket{\psi}$ such that $P_0\ket{\psi} = \ket{\psi}$ as follows
\begin{align}\label{eq:cor=P0}
    \widetilde{\mathrm{Cor}}_\psi(\sf A,B) :=& \max_{\norm{A}, \norm{B} \le 1}\ \left\langle AB \right\rangle_\psi - \frac{1}{2}\left[\left\langle AP_0B \right\rangle_\psi + \left\langle B P_0 A \right\rangle_\psi\right], \nonumber\\ =& \max_{\norm{A}, \norm{B}\le 1}\ \sum_{\phi:E_\phi\ge \Delta} \left\langle A |\phi\rangle\langle\phi| B\right\rangle_\psi + \frac{1}{2}\left\langle AP_0B \right\rangle_\psi -\frac{1}{2}\left\langle B P_0 A \right\rangle_\psi,
\end{align} 
where $\alr{\cdot}_\psi :=\langle\psi|\cdot|\psi\rangle$, $A,B$ are operators acting in subsystem $\sf A,B$ respectively. The second line exposes the excited states $\ket{\phi}$ with energy $E_\phi \ge \Delta$. Technically, the above definition is not equivalent to~\eqref{eq:cor}: They coincide if the ground state is unique $P_0=\ket{\psi}\bra{\psi}$. However, if the ground subspace is degenerate, a ground state may have long-range correlations, i.e., $\cor(\sf A,B)$ does not decay with the distance between $\sf A$ and $\sf B$. For example, consider the Greenberger–Horne–Zeilinger (GHZ) state \cite{GHZ89} \begin{equation}\label{eq:GHZ}
    \ket{\psi}=\ket{\rm GHZ}:=\frac{1}{\sqrt{2}}\lr{\ket{\bm 0} + \ket{\bm 1}},
\end{equation}  
where $\ket{\bm 0}$ is the product state where all sites are in state $0$ (analogously for $\ket{\bm 1}$). $\ket{\psi}$ is a ground state of the Ising Hamiltonian $H=\sum_{\braket{ij}}Z_i Z_j$, and is highly entangled. However, $\widetilde{\mathrm{Cor}}_\psi(\sf A,B)$ defined above vanishes as long as $\sf A$ and $\sf B$ do not overlap. 

The following theorem generalizes the above example to arbitrary gapped ground states: the connected correlations $\widetilde{\mathrm{Cor}}_\psi(\sf A,B)$ decay exponentially in the distance. 

\begin{theor}[Exponential clustering of gapped ground states \cite{Hastings_koma,nachtergaele06}]\label{thm:clustering}
Suppose the Lieb-Robinson bound (\ref{eq:LR_exp}) holds for the Hamiltonian $H$ with gap $\Delta$, and ground state energy 0. Then, for any ground state $\ket{\psi}$ and any $\epsilon \in (0,1)$, \begin{equation}
    \widetilde{\mathrm{Cor}}_\psi({\sf A, B}) \le \left(1+ \sqrt{\frac{2}{\pi\mathrm{e}\epsilon}} \frac{c\mu v}{\Delta} \right) \mathrm{e}^{-(1-\epsilon)\mu \mathsf{d}(\sf A,B)}.
\end{equation}
\end{theor}
\begin{proof}
To bound the connected correlations with Lieb-Robinson bounds, one begin with expanding the commutator in the eigenbasis of $H$
\begin{align}\label{eq:comm=cor}
    \left\langle [A(t), B]\right\rangle_\psi &= \left\langle A(t) B\right\rangle_\psi - \left\langle  BA(t)\right\rangle_\psi \nonumber\\ &= \left\langle AP_0B \right\rangle_\psi - \left\langle B P_0 A \right\rangle_\psi+ \sum_{\phi:E_\phi\ge \Delta} \mathrm{e}^{-\mathrm{i}tE_\phi} \left\langle A |\phi\rangle\langle\phi| B\right\rangle_\psi - \mathrm{e}^{\mathrm{i}tE_\phi} \left\langle  B|\phi\rangle\langle\phi|A\right\rangle_\psi.
\end{align}
The second line uses the assumption that the ground state has zero energy $P_0A(t)P_0=P_0AP_0$. 

The insight that converts temporal bounds (Lieb-Robinson) to spatial bounds (decay of correlation) is spectral filtering \cite{LSM04}: consider a kernel function $K(t)$ whose Fourier transform is denoted by 
\begin{align}
    \hat{K}(E) = \int\limits^\infty_{-\infty} \rd t \; K(t) \e^{-\ri tE},
\end{align}
with normalization \begin{equation}\label{eq:K0}
    \hat{K}(0)=\int\limits^\infty_{-\infty}\mathrm{d}t \; K(t) = \frac{1}{2}.
\end{equation}
Then, we may approximate the connected correlation by the weighted time-integral of~\eqref{eq:comm=cor}\begin{equation}\label{eq:cor=K}
    \widetilde{\mathrm{Cor}}_\psi({\sf A, B}) = \max_{\norm{A}, \norm{B}\le 1}\left[ \left\langle  B\mathcal{K}_-A\right\rangle_\psi - \left\langle A \mathcal{K}_+ B\right\rangle_\psi + \int^\infty_{-\infty}K(t)\left\langle [A(t), B]\right\rangle_\psi \mathrm{d}t\right]
\end{equation}
with ``error'' operators \begin{subequations}\begin{align}
    \mathcal{K}_+ &:=\sum_{\phi:E_\phi\ge \Delta} \left(\hat{K}(E_\phi)-1\right)|\phi\rangle\langle\phi|,\\ \mathcal{K}_-&:=\sum_{\phi:E_\phi\ge \Delta} \hat{K}(-E_\phi)|\phi\rangle\langle\phi|.
\end{align}\end{subequations}

To proceed, we impose two requirements for the kernel $K(t)$.  (\emph{1}) The Fourier transform approximates a step function:
    \begin{align}
        \hat{K}(E) \approx \begin{cases}
        1\quad \text{if}\quad E\ge \Delta\\
        0\quad \text{if}\quad E\le -\Delta
        \end{cases}.
    \end{align}
    This ensures both error operators $\mathcal{K}_+$ and $\mathcal{K}_-$ are small.
  (\emph{2}) The kernel $K(t)$ decays sufficiently fast at large $|t|$, so that the last term in \eqref{eq:cor=K} is tightly bounded.
Our choice is the Gaussian filter with tunable variance \begin{equation}
    K(t) =\frac{\mathrm{i}}{2\pi} \lim_{\epsilon^\prime\rightarrow 0+} \frac{\mathrm{e}^{-\alpha t^2} }{t+\mathrm{i}\epsilon^\prime}.
\end{equation}
 Indeed, it decays exponentially, and satisfies (\ref{eq:K0}). Furthermore, one can calculate its Fourier transform: 
\begin{lma}[Fourier transform of the error function]
If $\mathrm{Erf}(x):=\frac{2}{\sqrt{\pi}}\int^x_0\mathrm{d}\xi
\;\mathrm{e}^{-\xi^2}$ is the error function,
\begin{align}
    \hat{K}(E) = \frac{1}{2}\left(1+\mathrm{Erf}\left(\frac{E}{2\sqrt{\alpha}}\right)\right)= \left\{\begin{aligned}
        1+ \mathrm{O}\left(\mathrm{e}^{-E^2/4\alpha}\right),& \quad (E>0), \nonumber\\
    \mathrm{O}\left(\mathrm{e}^{-E^2/4\alpha}\right),& \quad (E<0). 
    \end{aligned} \right.
\end{align}
\end{lma}
Consequently, the norm of operators $\labs{\mathcal{K}_+}$ and $\labs{\mathcal{K}_-}$ are both bounded by $\mathrm{e}^{-\Delta^2/4\alpha}/2$. Together with  (\ref{eq:LR_exp}), we can bound the connected correlation \begin{align}
    \widetilde{\mathrm{Cor}}_\psi(\sf A,B) &\le \mathrm{e}^{-\Delta^2/4\alpha} +c\mathrm{e}^{-\mu \mathsf{d}(\sf A,B)} \int^\infty_{-\infty}|K(t)| \left(\mathrm{e}^{\mu v|t|} -1\right)\mathrm{d}t \nonumber\\ &=\mathrm{e}^{-\Delta^2/4\alpha} +\frac{c\mu v}{\pi}\mathrm{e}^{-\mu \mathsf{d}(\sf A,B)} \int^\infty_0 \mathrm{e}^{\mu vt}\mathrm{e}^{-\alpha t^2} \mathrm{d}t \le \mathrm{e}^{-\Delta^2/4\alpha} +\frac{c\mu v}{\pi}\mathrm{e}^{-\mu \mathsf{d}(\sf A,B)} \sqrt{\frac{\pi}{\alpha}} \nonumber\\ &= \left(1+ \frac{2c\mu v}{\sqrt{\pi}\Delta}\sqrt{\mu \mathsf{d}(\sf A,B)} \right) \mathrm{e}^{-\mu \mathsf{d}(\sf A,B)}  \le \left(1+ \sqrt{\frac{2}{\pi\mathrm{e}\epsilon}} \frac{c\mu v}{\Delta} \right) \mathrm{e}^{-(1-\epsilon)\mu \mathsf{d}(\sf A,B)}.
\end{align}
The second line's inequality uses the elementary bound $\mathrm{e}^x-1\le x\mathrm{e}^x$ for $x\ge 0$. The last line optimizes the tunable parameter \begin{equation}
    \alpha = \frac{\Delta^2}{4\mu \mathsf{d}(\sf A,B)}
\end{equation}
and used the estimate $\sqrt{2\mathrm{e}\epsilon x}\mathrm{e}^{-x}\le  \mathrm{e}^{-(1-\epsilon)x}$ for $x\ge 0$ and $\epsilon>0$. This is the advertised result.
\end{proof}

\subsection{Local properties of gapped systems}

Theorem \ref{thm:clustering} on exponential clustering implies that in a (unique) gapped ground state, a region $\sf A$ does not have much correlation with faraway vertices. Then $\sf A$ basically correlates with $\sf{A}^{\rm c}$ only via the vertices near the boundary $\partial \sf A$, so one naturally conjectures that the entanglement entropy of $\sf A$ is bounded by an \textbf{area law} \begin{equation}\label{eq:areaLaw}
    S_{\sf A} = {\rm O}\lr{ \labs{\partial \sf A}}.
\end{equation}
This area law is proven \cite{area07} for general 1D gapped systems originally by Hastings, using Lieb-Robinson techniques. With refinements \cite{area_FF12,area1d13,area_Renyi14} afterward, the result is summarized as follows. 
\begin{theor}[Area law of 1D gapped ground states \cite{area1d13}] \label{thm:areaLaw}
Consider a chain of $q$-dimensional qudits with a unique ground state that has a gap $\Delta$. Then the entanglement entropy across any cut is bounded by ${\rm O}\lr{\frac{\log^3 q}{\Delta}}$.
\end{theor}
The gap condition is explicitly used in the proof, which we omit here.\footnote{Gapless states generically do not have area laws: e.g. a lattice regularization of a conformal field theory \cite{Calabrese:2009qy}.} On the other hand, any state in 1D satisfying the exponential clustering condition is proven to obey an area law \cite{area_clusterNP,area_cluster12,area_cluster18}, so an area law for 1D gapped ground states directly follows Theorem \ref{thm:clustering}, if one does not care about the scaling of the entanglement with the gap, etc. As an application, Theorem \ref{thm:areaLaw} guarantees that 1D gapped ground states can be faithfully approximated by matrix product states (MPSs) with bound dimension sublinear in $N$ \cite{area1d13}, in contrast to ${\rm exp}(N)$ for general states. Based on the MPS representation, efficient classical algorithms have been developed for calculating the ground state and its properties. For example, the heuristic algorithm called density matrix renormalization group (DMRG) had been widely used \cite{DMRG_RMP05}; a closely related algorithm with provable guarantees was later proposed~\cite{algorithm1d15}.

In higher dimensions, the area law \eqref{eq:areaLaw} is an important open problem. We refer to the literature \cite{area_spin10,area_RMP10,area_sufficient14,area_tree19,Kuwahara1d,area_2dFF22} on recent progress, for example, a proof for 2D frustration-free systems \cite{area_2dFF22}. For general graphs, see \cite{area_counterexam14} for a counterexample that violates \eqref{eq:areaLaw} explicitly.

Instead of assuming a gap for a single Hamiltonian $H$, one can consider a family of Hamiltonians $\{H(s):0\le s\le 1\}$    
with the local terms in $H(s)$ depending on $s$ continuously. We assume the ground states have a gap lower bounded by $\Delta$ for the entire family of Hamiltonian $H(s)$. This assumption is common in condensed matter settings, where one says that two $H$ belong to the same phase of matter if and only if this interpolation exists. Intuitively, ground state properties of two Hamiltonians are qualitatively the same if they are in the same phase, so one can study solvable points in the phase space, and generalize results to the whole phase. 

Remarkably, the above intuition can be made rigorous. As proven in \cite{automorphic12}, the ground state subspace (defined by projectors $P(0)$ and $P(1)$) of $H(0)$ and $H(1)$ are connected by a quasi-local unitary $U$ \begin{equation}
    P(1) = U P(0) U^\dagger.
\end{equation}
Here $U$ is quasi-local in the sense that it is generated by a finite-time evolution of some sufficiently local (time-dependent) Hamiltonian $\tilde{H}$ (where local terms decay exponentially with the support size). Local operators are mapped to local operators by $U$ due to Lieb-Robinson bound, so the local properties of $P(0)$ and $P(1)$ are smoothly connected. Moreover, $\tilde{H}$ acts nontrivially only near places where $H(s)$ changes with $s$. This implies that the ground states locally only depend on terms in $H$ that lie in the neighborhood of that local region, dubbed ``local perturbations perturb locally''. The quasi-local unitary $U$ is called ``quasi-adiabatic continuation (evolution)'', and has many other applications \cite{LSM04,quasiadia05,Bravyi2006,LSM07,simu_adiab07}, including a rigorous proof of quantization of Hall conductance \cite{Hall_quant15,thermalhall}.

Above, we have seen the consequences of locality \emph{assuming} the gap condition. Perhaps surprisingly, Lieb-Robinson bounds also play an important role in \emph{deciding} whether a Hamiltonian is gapped or not (more precisely, how large the gap is), despite the fact that this spectral gap problem is undecidable in general \cite{undecidability15,undecidability20}. In fact, the revisiting of such bounds originates from such a problem. In 1961, Lieb, Schultz, and Mattis (LSM) proved that a particular spin-$1/2$ Hamiltonian (which is translation invariant and $\mathrm{U}(1)$-symmetric) cannot have a ${\rm \Omega}(1)$ gap above its unique ground state \cite{LSM61}. Hastings generalized this LSM theorem to higher dimensions using the quasi-adiabatic continuation method:
\begin{theor}[Lieb-Schultz-Mattis theorem in higher dimensions \cite{LSM04,LSM07}] 
Consider a local Hamiltonian $H$ on a finite-dimensional lattice. Suppose there is one direction of length $L$, along which $H$ is translational invariant with periodic boundary conditions. Suppose the total number of vertices $N={\rm poly}(L)$. Let $H$ have a conserved charge $Q$ and ground state $\ket{\psi}$. If the ground state filling factor $\bra{\psi}Q\ket{\psi}/L$ is not an integer, then the gap between $\ket{\psi}$ and the first excited state is bounded by \begin{equation}
    \Delta = {\rm O}\lr{\frac{\log L}{L}}.
\end{equation}
\end{theor}

On the other hand, Lieb-Robinson techniques are also used to prove an ${\rm \Omega}(1)$ gap for some $H$ that is close to a gapped $H_0$. To be specific, consider $H=H_0+V$ where $V$ is an extensive sum of local terms each of order $\epsilon \ll \Delta$. The naive perturbation theory typically diverges for this many-body setting, and the gap is not stable in general. For example, the gap closes already at $\epsilon\sim 1/N$ when perturbing the Ising Hamiltonian $H_0=-\sum_{\braket{ij}}Z_i Z_j$ by a magnetic field $V=\epsilon \sum_i Z_i$. It is then a remarkable fact that for frustration-free $H_0$ with local topological order, the gap is provably stable \cite{stab_FF_topo13}. 

\section{Bounds on thermalization}\label{sec:thermalization}
In this section, we return to bounds on dynamics.  In particular, we will describe the extent to which Lieb-Robinson bounds can give meaningful constraints on the thermalization time scales \cite{bound_noneq_rev22} observed in local correlation functions. Results in this area are more limited, and overall this is an interesting area for further exploration; we will highlight what is known.

\subsection{Relaxation times of local observables}
Consider a quantum system in initial mixed state $\rho$ evolved under Hamiltonian $H$. On physical grounds, we expect equilibration of local observable $A$ after some time $t_{\mathrm{thermal}}$ such that 
\begin{align}
 \tr[\rho A(t_{\mathrm{thermal}})] \approx \tr[\rho_{\mathrm{thermal}} A].  
\end{align}

Can we bound this thermalization time $t_{\mathrm{thermal}}$ based on Lieb-Robinson bounds? Assuming a Hamiltonian of the form (\ref{eq:2localH}), we have \begin{align}\label{eq:thermal_simple}
    \labs{\tr[\rho A(t)]-\tr[\rho A(0)]} &\le \lnorm{A(t)-A(0)} = \lnorm{\int^t_0 \mathrm{d} s \e^{\ri sH}\ri [H, A]\e^{-\ri sH} }\nonumber\\ &\le t\lnorm{[H,A]}\le t\norm{A}\sum_{e\in \mathsf{E}:e\cap \mathsf{A}\neq \emptyset} \lnorm{H_e}.
\end{align}
Allowing $\rho$ to be locally perturbed away from equilibrium, so that $| \tr[\rho A] - \tr[\rho_{\mathrm{thermal}}A]| \sim \norm{A}$, we obtain \begin{equation}
    t_{\mathrm{thermal}} \gtrsim \lr{\sum_{e\in \mathsf{E}:e\cap \mathsf{A}\neq \emptyset} \lnorm{H_e}}^{-1}.
\end{equation}
This also holds for relaxation of time-ordered correlation functions $\tr[\rho A(t)B(0)]$, because one can insert the $B(0)$ in the left-hand side of \eqref{eq:thermal_simple}, with the result unchanged assuming $\norm{B}=1$.  See \cite{Nussinov:2021fgc} for extensions of this idea to thermal states.

As a simple generalization, suppose $H=H_0+V$ where $H_0$ does not yield thermal dynamics. We bound thermalization by the local norm of $V$ alone, which is useful when $V$ is a weak perturbation \cite{Kukuljan:2017xag}.
\begin{prop}[Bound on the local thermalization time]\label{prop:thermal}
Consider a perturbed local Hamiltonian $H=H_0+V$ on a $d$-dimensional lattice. If $H_0$ obeys the Lieb-Robinson bounds\eqref{eq:LR_exp}, then \begin{align}
    \labs{\tr[\rho A(t)]-\tr[\rho\e^{t \CL_0} A]} \le c_{\rm thermal}\norm{A}\norm{V}_{\rm local} |\mathsf{A}|t\lr{vt+c'_{\rm thermal} }^D,
\end{align}
where $\CL_0:=\ri [H_0,\cdot]$, and \begin{equation}\label{eq:local_norm}
    \norm{V}_{\rm local}=\max_{j \in \mathsf{V}} \sum_{e\in \mathsf{E}: e\ni j}\lnorm{V_e},
\end{equation} 
is the local norm of $V$. The quantities $c_{\rm thermal},c'_{\rm thermal}$ are ${\rm O}(1)$ constants determined by $c,\mu$ and the lattice geometry.
\end{prop}
\begin{proof}
We use the Duhamel identity, similarly to \eqref{eq:thermal_simple}:
\begin{align}\label{eq:thermal1}
    \labs{\tr[\rho A(t)]-\tr[\rho \e^{t \CL_0} A]} &\le \lnorm{\int^t_0 \mathrm{d} s \e^{\ri sH}\ri [V, \e^{(t-s) \CL_0}A]\e^{-\ri sH} }\le \int^t_0 \mathrm{d} s \lnorm{[V, \e^{s \CL_0}A]}.
\end{align}
If $\e^{s \CL_0}A$ is supported in a set $\mathsf{S}$, then \begin{equation}
    \lnorm{[V, \e^{s \CL_0}A]}\le \norm{A} \norm{V}_{\rm local}|\mathsf{S}|.
\end{equation} 
From \eqref{eq:LR_exp} and its relation to the support of an evolved operator in Proposition \ref{prop:tildA}, we know that the dominant part $\tilde{A}(s)$ of $\e^{s \CL_0}A$ is supported in a set $\mathsf{S}$ containing vertices no farther than $vs+c'_{\rm thermal}$ to the original set $\mathsf{A}$: \begin{equation}
    \e^{s \CL_0}A = \tilde{A}(s)+ \sum_{r=\lfloor vs+c'_{\rm thermal}\rfloor+1} \tilde{A}_r,
\end{equation}
where $\tilde{A}_r$ has support no farther than $r$ to the original set $\mathsf{A}$. Here the constant $c'_{\rm thermal}$ is chosen such that $\lnorm{[V, \sum_r \tilde{A}_r]}\le \sum_r \lnorm{[V, \tilde{A}_r]}$ is smaller than the main contribution $\lnorm{[V, \tilde{A}(s)]}$, which is possible because the volume increases polynomially with $r$, while $\lnorm{\tilde{A}_r}$ decays exponentially. Then \eqref{eq:thermal1} becomes \begin{align}
    \labs{\tr[\rho A(t)]-\tr[\rho \e^{t \CL_0} A]} &\le \int^t_0 \mathrm{d} s\, 2\lnorm{[V, \tilde{A}(s)]} \le \int^t_0 \mathrm{d} s\, \norm{A} \norm{V}_{\rm local} |\mathsf{A}|\tilde{c}_{\rm thermal} (vs+c'_{\rm thermal})^D \nonumber\\ &\le c_{\rm thermal}\norm{A}\norm{V}_{\rm local} |\mathsf{A}|t\lr{vt+c'_{\rm thermal} }^D.
\end{align}
Here we have used the fact that $\mathsf{S}$ (support of $\tilde{A}(s)$) has the largest volume if $\mathsf{A}$ is a set of faraway vertices, where each vertex grows to a ball of radius $vs+c'_{\rm thermal}$.
\end{proof}

Proposition \ref{prop:thermal} implies a bound $t_{\rm thermal}\gtrsim \lr{\norm{V}_{\rm local}}^{\frac{1}{D+1}}$ if the unperturbed dynamics $\tr[\rho\e^{t \CL_0} A]$ is far from the equilibrium of $H$.   We do not expect this bound to be tight; in somewhat more specialized settings, stronger bounds (Corollary \ref{cor:decaydensity}) can be found.  In general, it is important to find stronger bounds on thermalization.

Up to now we only bounded the \emph{local} thermalization process. If the initial state is inhomogeneous with respect to some conserved charge, thermalization is slower. Transport is, firstly, bounded by the Lieb-Robinson velocity $v$, so if the length scale of the initial inhomogeneity is $L$, we have \begin{equation}
    t_{\rm global\ thermal} \gtrsim L/v.
\end{equation}
In usual systems however, this is a serious underestimate of the thermalization time scale, even accessible via local correlation functions.  If in particular there is a single conserved quantity (usually energy, e.g. whenever $H$ is time-independent), then that quantity will relax diffusively to equilibrium, meaning that in practice $t_{\rm global\ thermal} \gtrsim L^2$ \cite{thoulesstime} for finite-range interactions.

The bounds we have described thus far are rather ``simple" in that we simply pointed out that thermalization of local correlation functions is constrained by operator growth! A difficult open is to prove bounds on the nature of decay more generally.  For example, it seems reasonable that if one takes the thermodynamic limit of $N\rightarrow\infty$ interacting degrees of freedom, we might expect that for any local operator $A$,\begin{equation}
    \sup_{t\ge t_0} \lim_{N\rightarrow\infty}\langle A(t)A(0)\rangle \ge  C \mathrm{e}^{-\gamma t_0}
\end{equation}
for some non-negative constants $C$ and $\gamma$.  This would rule out in particular that correlation functions can decay as $\mathrm{e}^{-t^2}$.   The physical intuition for this is that in typical many-body quantum systems, there are poles in the lower half of the complex plane in the Green's functions for operator $A$ (see e.g. \cite{lucasbook}).   Any such pole would lead to at best a finite $\gamma$ (with hydrodynamic poles leading to algebraic decay).   However, a proof of the finiteness of $\gamma$ (in the thermodynamic limit $N\rightarrow\infty$) and a constraint on $\gamma$ in terms of local coupling constants remains an open problem, as far as we know.

\subsection{Prethermalization}\label{sec:preth}
In some cases, we can achieve far stronger results than Proposition \ref{prop:thermal}. To see how this is possible, we describe a classic example of a slowly thermalizing system: the Fermi-Hubbard model, with Hamiltonian $H=H_0+V$, where \begin{equation}\label{eq:fermi_hub}
    H_0 = \Delta \sum_{i\in\mathsf{V}} n_{i,\uparrow}n_{i,\downarrow}, \quad V = \epsilon \sum_{\lbrace i,j\rbrace \in \mathsf{E}}\sum_{\sigma=\uparrow, \downarrow} c^\dagger_{j,\sigma} c_{i,\sigma}.
\end{equation}
Here $c_{i,\sigma}$ is the annihilation operator for a fermion of spin $\sigma\in\lbrace \uparrow,\downarrow\rbrace$ at site $i\in\mathsf{V}$.   Consider the limit where $\epsilon \rightarrow 0$ while $\Delta$ stays finite.  In this limit, the ground state $H=0$ becomes highly degenerate: the states of finite energy ($\Delta,2\Delta,\ldots$) correspond to those where two fermions of opposite spin sit on the same site.  We call such an excitation a doublon, and a singly-occupied site a singleton.  

Suppose that we create one doublon excitation in a sea of singletons: how long will it take to decay, if $\epsilon$ is very small (but not exactly 0)?
If $\epsilon\ll \Delta$, Fermi's golden rule suggests we cannot simply split it into two singletons.  Fermions hopping on a lattice have a bandwidth: their maximal kinetic energy is $\sim \epsilon \ll \Delta$. There can therefore be no way, via Fermi's golden rule, to spontaneously decay into just two energetic singletons.  We must instead look for a much higher order process, where the doublon splits and then \emph{virtually} transfers its energy $\Delta$ into the increased kinetic energy of $k_* \sim \Delta/\epsilon$ singletons.  In perturbation theory, this will require at least $k_*$ powers of the perturbation $V$, meaning that the doublon decay rate is expected to be \cite{preth_exp10,preth_Hub10,preth_Hub12}
\begin{equation}
    \frac{1}{t_{\text{doublon decay}}} \sim \epsilon \left(\frac{\epsilon}{\Delta}\right)^{\epsilon/\Delta}.
\end{equation}

An exponential bound $t_{\text{doublon decay}}\sim \exp[\Delta/\epsilon]$ was rigorously proved in \cite{preth_rigor17}. The proof uses the fact that $H_0$ is trivially diagonalizable, and extends to Floquet systems and to other settings where $H_0$ (or its Floquet generalization) is solvable \cite{prethermal1,prethermal2,prethermal3,prethermal4}.

The intuitive argument we gave above seems to only rely on the existence of a gap in $H_0$. Thus one expects a general robustness result for all $H_0$ with a many-body gap.  Indeed, there is a long history of the study of false vacuum decay \cite{falseVac_Coleman}, wherein local correlation functions appear consistent with ground states of a degenerate $H_0$, even when the true Hamiltonian $H=H_0+V$ consists of a perturbation that has closed the gap.  One often takes, e.g., $H_0$ to have a ferromagnetic ground state, while $V$ is a symmetry-breaking field that selects one of the degenerate vacua as the true ground state.

However, as we have already explained at the beginning of Section \ref{sec:LRbounds}, any robustness of a false vacuum cannot arise from the local robustness of eigenstates of $H$.  A notion of spatial locality and Lieb-Robinson bound will play an important role in any proof.  The strongest known bound on prethermalization is:

\begin{theor}[Prethermalization of gapped systems (slightly informal) \cite{preth22}]\label{thm:preth}
Let $H_0$ be a spatially local Hamiltonian in $D$ spatial dimensions, with a gap $\Delta$ in the many-body spectrum.  Let $H=H_0+V$, with $\lVert V\rVert_{\mathrm{local}} = \epsilon$, as defined in (\ref{eq:local_norm}). Then, there exists a quasilocal unitary $U$, a Hamiltonian $H_*$ and an operator $V_*$ such that 
\begin{equation}
  U^\dagger (H_0+V)U=  H_* + V_*\label{eq:rotation}
\end{equation}
and for any single-site operator $A$ \begin{equation}
    \lVert U^\dagger A U-A\rVert = \mathrm{O}(\epsilon/\Delta).
\end{equation}
In particular, the Hamiltonian $H_*$ is block diagonal between the eigenstates of $H_0$ above and below the gap. The operator satisfies
\begin{equation}\label{eq:t*<exp}
    \lVert V_*\rVert_{\mathrm{local}} \le \frac{1}{t_*} \quad \text{where}\quad t_* \sim \frac{1}{\epsilon} \exp \left[ c\left(\frac{\Delta}{\epsilon}\right)^a \right]
\end{equation} 
for any $a<1/(2D-1)$ and some $0<c<\infty$. 
The unitary can be written as \begin{equation}
 U = \mathcal{T} \exp\left[\int\limits_0^T   \mathrm{d}t \hat{H}(t)\right]
\end{equation}
for $T=\mathrm{O}(\epsilon)$ and $\hat{H}(t)$ quasilocal in the sense that \begin{equation}
\sum_{\mathsf{S}\subset\mathsf{V}} \lVert \hat H_{\mathsf{S}}(t) \rVert \mathrm{e}^{\mathrm{diam}(\mathsf{S})^\alpha} = \mathrm{O}(1) \quad \text{for any}\quad 0<\alpha<1.
\end{equation}

\end{theor}

In other words, there exists some effective Hamiltonian $H_*$ that effectively describes the dynamics of \emph{local correlation functions} for times $t\ll t_*$.  For example, if we start in one of the degenerate ground states of $H_0$, $|0\rangle$, \begin{equation}
    \left| \langle 0| (U^\dagger A U)(t) |0\rangle - \langle 0|A(t)|0\rangle \right| \le \lVert U^\dagger A U - A \rVert \lesssim \epsilon,
\end{equation}
meaning that
\begin{equation}
  \langle 0|A(t)|0\rangle \approx \langle 0| U^\dagger \mathrm{e}^{\mathrm{i}H_*t}A\mathrm{e}^{-\mathrm{i}H_*t}U|0\rangle + \lVert V_*\rVert_{\mathrm{local}}t. 
\end{equation}
If $A$ acts in a trivial way on the ground states, but non-trivially on typical low-energy states, then we see that up to error $\mathrm{O}(\epsilon)$, $\langle 0|A(t)|0\rangle$ will evolve very slowly away from its ground state value. For a non-perturbatively long time, it will appear from $\langle A\rangle$ as if the system is in one of its ground states, even if the perturbation has added a finite energy density to the state!  This proves, therefore, that false vacuum decay is non-perturbatively slow.
\begin{proof}[Proof idea]
To formalize the intuition about a mismatch of energy scales between $H_0$ and $V$ causing prethermalization, one needs to do perturbation theory in a rigorous way. A convenient technique, used in \cite{preth_rigor17}, is the Schrieffer-Wolff transformation \cite{SWolf66,SW_many11}.  Note that this proof strategy is similar in spirit to the proof of the Kolmogorov-Arnold-Moser Theorem \cite{arnold}.

Starting from the original $H=H_0+V_1$ with $V_1=V$, we first look for a quasi-local unitary $U_1=\e^{A_1}$ that block-diagonalizes the system among the gapped subspaces of $H_0$ at order $\epsilon^1$. Note that we do not want a unitary that block-diagonalizes the system \emph{completely}, as it cannot be quasilocal due to the orthogonality catastrophe \cite{ortho_cata67}. More precisely, we look for anti-Hermitian $A_1 = \mathrm{O}(\epsilon)$ obeying \begin{equation}\label{eq:UHU=}
    \e^{-A_1}(H_0+V_1)\e^{A_1} = H_0 + V_1 + [H_0, A_1] + \mathrm{O}(\epsilon^2).
\end{equation}
We demand that for some $D_2$, block-diagonal betewen the high/low energy subspaces of $H_0$: \begin{equation}\label{eq:V1=D2}
    V_1 + [H_0, A_1] = D_2.
\end{equation}
$A_1$ is not uniquely determined by \eqref{eq:V1=D2}, but one solution suffices. 

For commuting $H_0$ like \eqref{eq:fermi_hub}, $A_1$ and $D_2$ can be found easily in the eigenstate representation of $H_0$. For example, in (\ref{eq:fermi_hub}), $H_0$ has integer spectrum $0,\Delta,2\Delta,\ldots$, and one can choose \cite{preth_rigor17} \begin{subequations}\label{eq:D2A1}\begin{align}
    D_2&=\frac{\Delta}{2\pi}\int\limits_0^{2\pi/\Delta}\mathrm{d}t\ \mathrm{e}^{\mathrm{i}tH_0}V_1\mathrm{e}^{-\mathrm{i}tH_0}, \\
    A_1 &= -\mathrm{i} \int\limits_0^{2\pi/\Delta}\mathrm{d}t\ \lr{1-\frac{\Delta t}{2\pi}} \mathrm{e}^{\mathrm{i}tH_0}V_1\mathrm{e}^{-\mathrm{i}tH_0}.
\end{align}\end{subequations}
By the evolution of $H_0$ in \eqref{eq:D2A1}, each local term of $V_1$ grows larger in support, but remains strictly local because $H_0$ is commuting. 

For general $H_0$, a solution similar to \eqref{eq:D2A1} still holds.  One finds a filter function $w(t)$ with compact Fourier transform \cite{quasi-adia10,automorphic12}, which decays reasonably quickly as  $w(t) \sim \mathrm{e}^{-|t|/\log^2 |t|}$ at large $|t|$, and chooses \cite{preth22} \begin{equation}
    D_2 \sim \int\limits_{-\infty}^\infty \mathrm{d}t \; w(t) \mathrm{e}^{\mathrm{i}tH_0}V_1\mathrm{e}^{-\mathrm{i}H_0}.
\end{equation} 
The matrix elements $\langle E| D_2 |E^\prime\rangle$ between eigenstates of $H_0$ are proportional to the Fourier transform of $w(t)$:  $\widehat{w}(E-E^\prime)$; hence our filter function with compact Fourier transform enforces $D_2$ being block-diagonal across the gap.
However, terms in $D_2$ (and $A_1$) are no longer strictly local, and one needs to invoke a special Lieb-Robinson bound \cite{preth22} to control how large they can become. With $A_1$ chosen to satisfy \eqref{eq:V1=D2}, we have rotated the Hamiltonian by $U_1$ to $H_0 + D_2+V_2$, where $V_2\sim \epsilon^2$ is the last term in \eqref{eq:UHU=}. At this second step, we wish to further block-diagonalize the Hamiltonian, using a quasi-local unitary $U_2=\e^{A_2}$ that is determined similar to \eqref{eq:V1=D2}.  One again needs to invoke locality and Lieb-Robinson bounds to show that $U_2$ is still quasi-local. This process can be iterated up to some optimal order $k_*\sim (\Delta/\epsilon)^a$, where the range of operators has become so large that further Schrieffer-Wolff transformations do not decrease the local strength of $V_k$ anymore. At this optimal order, \eqref{eq:rotation} is achieved with $H_* = H_0+D_{k_*}$ and an exponentially small $V_* = V_{k_*}$.
\end{proof}

 Based on \eqref{eq:rotation}, there is a hierarchy of how strongly $t_*$ depends on $\epsilon$. First, \eqref{eq:rotation} always holds trivially with $U=I$ and $t_*=1/\epsilon$, so the first nontrivial result would be $t_*\sim \epsilon^{-k_*}$ for some finite $k_*>1$ using finite-order perturbation theory \cite{hamazaki, finite_order_math}.  Usually Fermi's golden rule implies that $t_*\sim \epsilon^{-2}$.  
 
 In integrable models, one can define a decay time in an alternative fashion to \eqref{eq:rotation}.  While perturbing away from such integrability typically gives rise to $t_*\sim \epsilon^{-2}$ governed by Fermi golden rule \cite{VANHOVE55,Mori_rev18}, there are exceptions with $t_*\sim \epsilon^{-2k}$ for $k>1$ \cite{te4_08,Surace:2023wqq}. 

There are even cases with infinite-time stability ($t_*=\infty$), such as frustration-free ground states with local topological order \cite{stab_FF_topo13}. The toric code \cite{toric_code03} is the classic example of such a state. Intuitively, topological order guarantees a \emph{macroscopic} code distance in the language of quantum error correction: any operator of size smaller than system length neither couples nor distinguishes different low-energy sectors that encode quantum information. Thus perturbation theory converges up to $d_{\rm code}$-th order, and the remaining $V_*$ vanishes in the thermodynamic limit. In addition to preserving the gap, the energy splitting in the ground subspace is also then exponentially small, making such models robust quantum memories at zero temperature. Although we believe the frustration-free condition is a technical issue rather than being physical, Theorem \ref{thm:preth} is currently the best bound for frustrated systems.

On general grounds, one might have expected the prethermalization time to scale as \begin{equation}
    t_{\mathrm{thermal}}\sim \exp \left[\left(\frac{\Delta}{\epsilon}\right)^d \right]
\end{equation}
in $d$ spatial dimensions, since this is the energy barrier one needs to overcome to tunnel out of the false vacuum \cite{falseVac_Coleman}.  A rigorous proof of this result would likely require something beyond a Lieb-Robinson bound to control the validity of the Schrieffer-Wolff transformation, and this is an interesting open problem.
\section{Quantum walk bounds and the Frobenius light cone}\label{sec:frobenius}

So far, we have discussed a notion of light cone inspired by Lieb-Robinson bounds on operator norms of commutators, such as $\lVert [A_0(t),B_r]\rVert = C(r,t)$.
These bounds have been popular because they hold for all matrix elements of the commutator: therefore, if a Lieb-Robinson bound exists, it serves as a versatile subroutine whenever a notion of locality is needed.  

However, there are many physical settings where one does not \emph{want} a Lieb-Robinson-like bound, but rather something different.   For example, suppose we wish to calculate a retarded Green's function in some finite temperature many-body system: \begin{equation}
    G^{\mathrm{R}}_{AB}(r,t) := \frac{\mathrm{i}}{Z(\beta)} \mathrm{tr}\left( \mathrm{e}^{-\beta H} [A_r(t),B_0]\right),
\end{equation}
where $Z(\beta) := \mathrm{tr}(\mathrm{e}^{-\beta H})$ is the thermal partition function. 

Especially if temperature $T$ is very small (or inverse temperature $\beta=1/T$ large compared to couplings in $H$), quite often $G^{\mathrm{R}}_{AB}$ will vanish outside of a light cone with an apparent \textit{temperature-dependent} velocity.  By definition, this temperature dependence cannot be captured by a Lieb-Robinson bound, as the single Lieb-Robinson bounds must accommodate all states.

Therefore, it is desirable to incorporate the initial state dependence into a Lieb-Robinsonb bound. For the most part, this is an open problem in mathematical physics, with preliminary progress just beginning.  The simplest context where many rigorous results -- which are notably stronger than optimal Lieb-Robinson bounds -- appear is when studying the \emph{Frobenius norm} of a commutator: (\ref{eq:frobeniusnorm}). This should not be a surprise: the Lieb-Robinson bounds must hold for all matrix elements of an operator, while the Frobenius norm simply bounds the \emph{average} magnitude of a matrix element between any states in Hilbert space.  What is more non-trivial is the qualitatively new methods for bounding the Frobenius norm, which can both be applied to physically relevant problems and give us new and helpful insight into the bottlenecks of the underlying quantum dynamics.  

In this section, we will review this Frobenius light cone in our usual context: systems with local interactions on a lattice.  In later sections, we will show that it is the Frobenius approach to bounding commutators that can have elegant generalizations to more challenging problems, including systems with power-law interactions (Section~\ref{sec:power-law}) or bosons (Section~\ref{sec:bosons}).

\subsection{Quantum walk of a single particle on the line}\label{sec:babyQW}
We begin by revisiting the toy problem of Section~\ref{sec:warmup_oneparticle}; as before, we wish to bound $C(r,t) = \langle r|\mathrm{e}^{-\mathrm{i}Ht}|0\rangle$.  This time, we will not try to think of this problem combinatorially.  Instead, we observe that $|C(r,t)|^2$ represents the classical probability of measuring the particle on site $r$ at time $t$.  Our strategy -- and more generally, the strategy of existing ``quantum walk bounds" -- will be to bound this probability distribution directly, using methods of classical probability theory, as an indirect way of saying something useful about the underlying quantum dynamics.

Before we explain \emph{how} such a quantum walk bound could be found, let us emphasize \emph{why} this shift in perspective should be quite useful.  In the combinatorial approach that underlies the simplest Lieb-Robinson bounds, we found in (\ref{eq:1d_sum_paths}) that $C(r,t)$ was bounded by a sum over all paths. 
Somewhat annoyingly in this formula, $C(0,t) > 1$ once $t>0$.  Now of course this is merely an artifact of us trying to express $C(r,t)$ in some elegant way -- physically $|C(0,t)|^2\le 1$, as it is the probability of measuring the particle on site 0.  Still, when we derived a  bound on $C(r,t)$ at time $t$, it depended on our bound at time $t-\mathrm{\Delta}t$.  The overcounting that we are doing is not innocuous -- it is ``corrupting" our bound for all later time, artificially inflating the values of all $C(r,t)$.  Might it be leading to an overestimate of the velocity of the ``light cone" outside of which $C(r,t)$ is exponentially small?

To understand whether this concern is justified, we seek a formalism in which (perhaps indirectly) we are assured that $C(r,t)$ remains the coefficients of some normalized wave function.  Yet this is somewhat awkward since it is $|C(r,t)|^2$ which represents the probabilities.  The strategy which has been used so far is to solve this problem by not bounding $C(r,t)$ directly, but rather by bounding the expectation values of operators on the Hilbert space.  Suppose we define \begin{equation}
    F := \sum_{r\in\mathbb{Z}} F_r |r\rangle\langle r|;
\end{equation}
then \begin{equation}
    \langle F(t)\rangle = \langle \psi(t)|F|\psi(t)\rangle = \sum_{r\in\mathbb{Z}} F_r |C(r,t)|^2.
\end{equation}
Since $F$ is diagonal in the position basis of interest, $\langle F\rangle$ can be interpreted using classical probability theory as simply the average value of the random variable $F_r$.  But now we can efficiently bound \begin{equation}
    \frac{\mathrm{d}}{\mathrm{d}t} \langle F(t)\rangle = \langle \mathrm{i}[H,F]\rangle. \label{eq:ddtft}
\end{equation}
While at this point there are a variety of strategies that one could use, two common ones are to seek functions $F$ where it can be proved that for some constant $c>0$ \begin{equation}
    |\langle \mathrm{i}[H,F]\rangle |\le c \quad \text{or}\quad |\langle \mathrm{i}[H,F]\rangle |\le c \langle F\rangle. \label{eq:QWexpF}
\end{equation}
In this section, we will focus on this latter possibility, which leads to tighter bounds -- the strategies for dealing with the former are quite similar.

Observe that one choice\footnote{We might really wish to use $b|r|$, not $br$, in the exponent, but this choice will simplify a few equations and the approach's merits are more easily revealed.} satisfying (\ref{eq:QWexpF}) is \begin{equation}
    F_r = \mathrm{e}^{br} \quad \text{for constant}\quad b>0.
\end{equation}
Indeed, \begin{align}
    |\langle \mathrm{i} [H,F]\rangle| &\le \left|h\sum_{r} \left(\mathrm{i}\bar\psi_r \psi_{r+1} - \mathrm{i}\bar\psi_{r+1} \psi_{r}\right)\left(\mathrm{e}^{b(r+1)} - \mathrm{e}^{br}\right) \right| \notag \\
    &\le h\left(\mathrm{e}^{b/2}-\mathrm{e}^{-b/2}\right) \sum_r \left(|\psi_r|^2 \mathrm{e}^{br}  + |\psi_{r+1}|^2\mathrm{e}^{b(r+1)}\right) \notag\\
    &= 4h \sinh\frac{b}{2} \sum_r |\psi_r|^2 \mathrm{e}^{br}
    = 4h\sinh\frac{b}{2} \langle F\rangle. \label{eq:qwsinhb}
\end{align}
The second inequality uses Cauchy-Schwartz (e.g.) \begin{equation}
    \labs{\mathrm{i} \bar \psi_r \psi_{r+1} \mathrm{e}^{b(r+1)}} = \labs{\mathrm{i}\bar \psi_r \mathrm{e}^{br/2}}\labs{\psi_{r+1}\mathrm{e}^{b(r+1)/2}}\mathrm{e}^{b/2} \le \mathrm{e}^{b/2}\cdot \frac{1}{2}\left[ \left|\mathrm{i}\bar \psi_r \mathrm{e}^{br/2}\right|^2 + \left|\psi_{r+1}\mathrm{e}^{b(r+1)}\right|^2 \right].
\end{equation}
Now, since $F_r\ge 0$ is a non-negative operator, and $|\psi_r|^2$ is a well-posed classical probability distribution, we can invoke Markov's inequality: \begin{equation}
    \mathbb{P}[\text{particle is at }x\ge x_0\text{ at time }t] = \sum_{r=x_0}^\infty |\psi_r(t)|^2 \le \frac{\langle F(t)\rangle}{F_{x_0}}. \label{eq:markov}
\end{equation}
Here and below $\mathbb{P}[\cdots]$ is used to denote the probability of an event arising, and we use this notation (rather than the expected value of a quantum observable) to highlight the close mathematical connections to probability theory.  Combining (\ref{eq:ddtft}), (\ref{eq:qwsinhb}), and (\ref{eq:markov}), we can now choose the optimal value for the parameter $b$ to get the tightest possible bound on the velocity of the particle: \begin{equation}
    \mathbb{P}[\text{particle is at }x\ge x_0\text{ at time }t] 
    \le \exp\left[-bx_0\left(1-\frac{4ht}{x_0}\frac{\sinh (b/2)}{b}\right) \right]. \label{eq:4bopt}
\end{equation} 
Since $\sinh(b)/b \ge 1$, we conclude that we should take $b\rightarrow 0$ to get the strongest possible light cone bound.  Our bound is exponentially small when $2ht \le x_0$, implying that the velocity of the particle is \begin{equation}
    v \le 2h.
\end{equation}
This is tighter than what we could find using our Lieb-Robinson combinatorics in (\ref{eq:vlrsec31}).   This velocity admits a natural physical interpretation: it is the largest possible group velocity of a particle in the system: see (\ref{eq:vgsec31}).

We call this approach the ``quantum walk" approach to bounding dynamics since we aim to use (as much as possible) the unitarity of the quantum dynamics to constrain the quantum walk of the wave function.  In the many-body setting, this problem can of course become much more complicated, but we will describe a few examples ($F$ ansatzes) where this method has been used to tightly bound quantum dynamics.  This approach has also been employed in the literature on continuous-space Lieb-Robinson bounds (Section \ref{sec:continuous}).

\subsection{Operator growth and operator size}
We now turn to the many-body problem. As explained in Section \ref{sec:LRbounds}, a critical difference between our single-particle warm-up and the many-body problem is operator dynamics appear more natural than state dynamics. Indeed, we have already seen that local operators evolve ``slowly" (at least for short times) under local Hamiltonian dynamics.  This was the key insight behind our derivation of a Lieb-Robinson bound in Section \ref{sec:LRbounds}.   And in recent years, capturing the growth of local operators -- beyond the Lieb-Robinson bound -- has become a question of particular interest among physicists.  There are a few (related) reasons why.  Firstly, the Lieb-Robinson bounds capture the ``worst case" speed of information, but it may be the case that in typical states signals propagate much more slowly.  We will see that this is strikingly the case in systems with power-law interactions in Section \ref{sec:power-law}.   Secondly, the growth of operators has been conjectured to be related to the emergence of geometry and gravity via the AdS/CFT correspondence in string theory: we will discuss such theories in Section \ref{sec:all-to-all}.  Thirdly, typical experiments probe \emph{thermal} correlation functions and so often times we are only interested in the behavior of the commutator $[A(t),B]$ averaged over (exponentially) many states.

If the Lieb-Robinson bounds are too specialized, we can try alternatively to study the opposite limit where we only ask about the typical behavior of a commutator.  Indeed, suppose we want to know: how large should we expect $ [A(t),B] | \varphi^\prime\rangle$ to be for randomly chosen state $|\varphi\rangle$? As usual in probability theory, it is easier to study the square of this object, where the averaging becomes simple: denoting $C=[A(t),B]$ for simplicity, along with a finite-dimensional Hilbert space, 
 \begin{equation}
    \mathbb{E}_{\varphi}\left[ \norm{C\ket{\varphi}}^2\right] = \mathbb{E}_{\varphi}\mathrm{tr}\left( C^\dagger C \ket{\varphi}\bra{\varphi}\right)=\frac{\mathrm{tr}\left(C^\dagger C\right)}{\mathrm{tr}(I)}. \label{eq:gettingfrobenius}
\end{equation}
Here we have used $\mathbb{E}_{\varphi}$ to denote expectation over the Haar measure\footnote{Concenptually we often take the ``random" average over the Haar measure. However, one only needs a much weaker randomness $\mathbb{E}_{\varphi}\ket{\varphi}\bra{\varphi} = I/\tr[I]$ for~\eqref{eq:gettingfrobenius}.} .   Note that the final object in (\ref{eq:gettingfrobenius}) is the Frobenius norm, introduced in (\ref{eq:frobeniusnorm}).  Hence, the Frobenius norm of $C$ will tell us about the size of \emph{typical} matrix elements of $C$.

In a system with a many-body Hilbert space, there is a particularly valuable way to think about this Frobenius norm. Using the elementary properties in Section \ref{sec:qubits}, we notice two key facts.   Firstly, if we expand an operator into the $|a_1\cdots a_N)$ basis, the coefficients $(a_1\cdots a_N|A)$ can be thought of as the elements of an ``operator wave function", which (by Proposition \ref{prop:pauliorthogonal}) will be normalized just as a usual quantum wave function. In particular, observe that \begin{prop}[Dynamics preserves the operator ``wave function" normalization]
\label{prop:Aconslength}
Let $|A(t)) = \mathrm{e}^{\mathcal{L}t}|A)$ for some Hamiltonian $\CL = \ri [\cdot ,H]$. Then \begin{equation}
    (A(t)|A(t)) = (A|A).
\end{equation}
\end{prop}
\begin{proof}
Take the time derivative \begin{equation}
    \frac{\mathrm{d}}{\mathrm{d}t} (A(t)|A(t)) = (A(t)|\mathcal{L}^\dagger + \mathcal{L}|A(t)) = 0,
\end{equation}
using $\mathcal{L}^\dagger = -\mathcal{L}$, which follows from the cyclic trace identity.
\end{proof}

Secondly, we will make heavy of the super-projector $\BP_A$ (Definition~\ref{defn:super_projector}) instead of commutators; recall 
\begin{align}
[A,B]= [A,\mathbb{P}_{\mathsf{A}}B].    \end{align}
This statement holds irrespective of a Frobenius light cone, and indeed such notation was used in \cite{alpha_3_chenlucas}. However, this notation of projection is particularly nice when working with the Frobenius inner product, as $\mathbb{P}$ can be thought of as an explicit projection matrix! Therefore, we have the intuitive picture of the ``light cone" for operator growth as associated with a small operator exploring the intersection of increasingly many $\mathbb{P}_j$ hyperplanes for sites $j$: see Figure \ref{fig:Fnorm_1d}.

\begin{figure}[t]
\centering
\includegraphics[width=.8\textwidth]{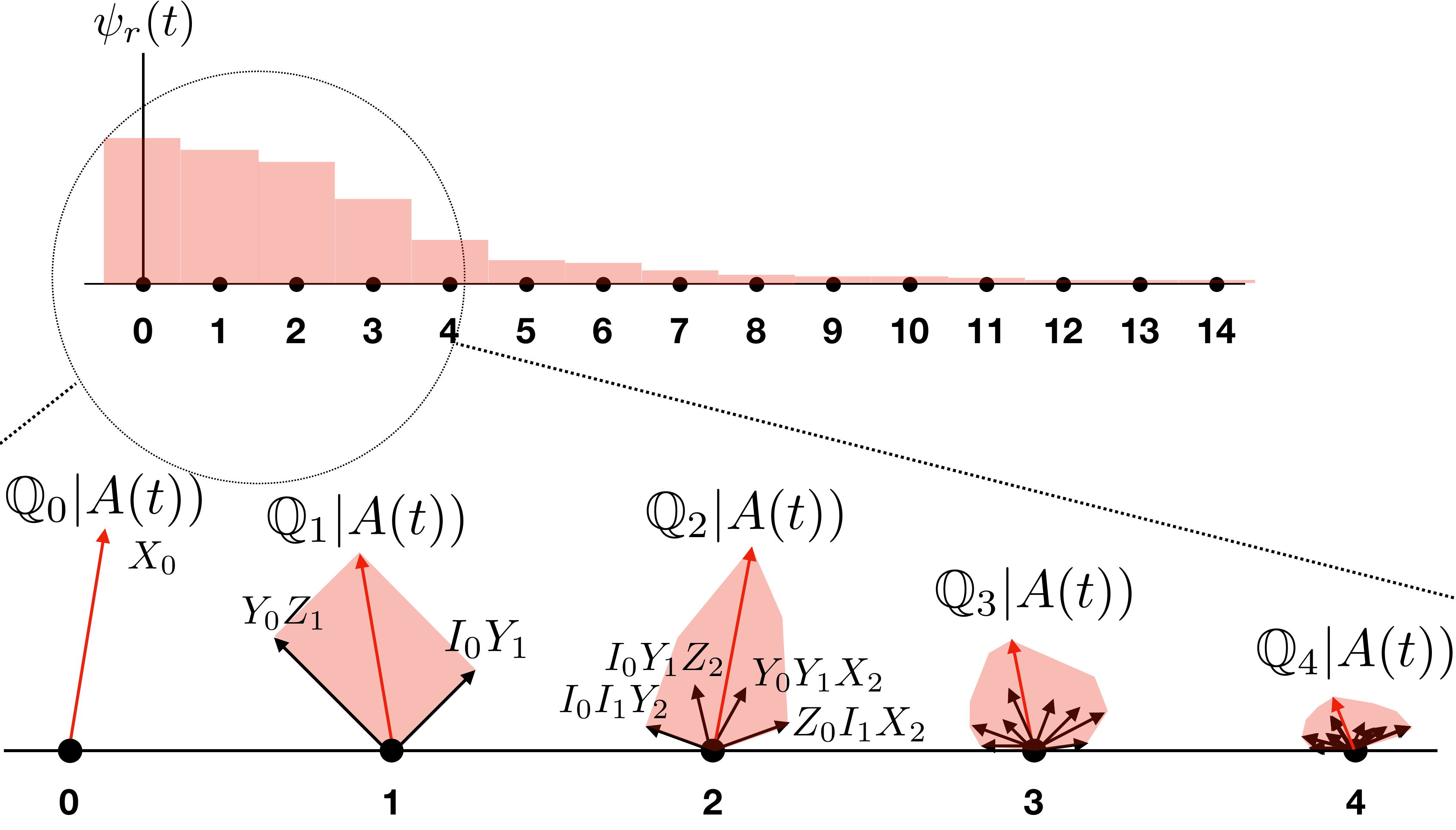}
\caption{
The many-body wave function can be captured in the Frobenius norm. Each projector $\BQ_r$~\eqref{eq:Q_r} selects operators whose right most non-trivial site is $r$, whose precise decomposition is irrelevant to our quantum walk bounds.}
\label{fig:Fnorm_1d}
\end{figure}

A crude way to capture how much an operator has grown (which has become an object of some study in its own right) is the \textbf{operator size}.  Define the superoperator \begin{equation}
    \mathcal{S}|a_1\cdots a_N) = \left[\sum_{j=1}^N \mathbb{I}(a_j\ne 0) \right]|a_1\cdots a_N).
\end{equation}
Thus, $\mathcal{S}$ counts the number of sites on which a given Pauli string is not the identity.  As a superoperator, size can act non-trivially on a complicated operator (just as not all wave functions are eigenstates of a general Hermitian operator).   But we often care about the \emph{average size} of an operator: $(A|\mathcal{S}|A)$.  The average size of an operator is, fortunately, related to the typical size of commutator matrix elements: \begin{prop}[Average operator size]
\label{prop:operatorsize}
For a quantum system with $N$ qubits,
\begin{equation}
    (A|\mathcal{S}|A) = \frac{1}{8} \sum_{j=1}^N \sum_{a=1}^3 ([A,X_j^a]|[A,X_j^a]).
\end{equation}
\end{prop}
Hence, the average size of operator $A$ measures the Frobenius norm of an average commutator of $A$ with a single-site Pauli. While this proposition is for qubit degrees of freedom, for appropriate changes in constant factors it also holds for qudits.  Notice the absence of spatial locality in the problem; often notions of operator size are most relevant in problems \textit{without} a notion of spatial locality.  We will discuss much more about operator size and why it was studied in Section \ref{sec:lyapunov}.

\subsection{Many-body quantum walk bounds}

In this section, we describe a strategy to efficiently bound Frobenius light cones dubbed the ``many-body quantum walk" in \cite{Lucas_2020}.  This technique is very similar to one which has been developed independently in the mathematics literature (see e.g. \cite{Faupin:2021nhk,Faupin:2021qpp}), though the notation is rather different.  The key idea is that since the Frobenius norm is simply the length of the ``operator" in the ``operator Hilbert space", we can follow our technique from Section \ref{sec:babyQW}. 

In particular, suppose we want to bound $\lVert [A_0(t),B_r]\rVert_{\mathrm{F}}$ -- this can only happen if the operator string $|A(t))$ contains terms that act non-trivially on site $r$.
\begin{prop}[From the ``operator wave function" to Frobenius light cone]
Let $\mathbb{P}_r:=\mathbb{P}_{\{r\}}$. Then, the following inequality holds: \begin{equation}\label{eq:APA>AB_F}
    4(A_0(t)|\mathbb{P}_r|A_0(t)) \cdot \norm{B_r}^2 \ge \lVert [A_0(t),B_r]\rVert_{\mathrm{F}}^2.
\end{equation}
\end{prop}
\begin{proof}
This follows from H\"older's inequality (Proposition~\ref{prop:holder}):
\begin{equation}
    \lVert [A_0(t),B_r]\rVert_{\mathrm{F}} = \lVert [\BP_r A_0(t),B_r]\rVert_{\mathrm{F}} \le 2 \norm{B_r} \lVert \BP_r A_0(t) \rVert_{\mathrm{F}},
\end{equation}
which squares to \eqref{eq:APA>AB_F}. 
\end{proof}

\subsubsection{Frobenius bounds on general graphs}

Now that we have related the Frobenius light cone to the expectation value of $\mathbb{P}_r$, it remains to bound $(A_0(t)|\mathbb{P}_r|A_0(t))$.  This is very similar to what we did in Section \ref{sec:babyQW}, except that $\mathbb{P}_r$ is not the projection onto a single state. 

\begin{theor}[Frobenius bound on a graph of bounded degree] \label{thm:qw}
In the model studied in Example \ref{exam:bound_degree}, 
\begin{equation}
    (A_0(t)|\mathbb{P}_v|A_0(t)) \le \mathrm{e}^{v_{\mathrm{B}}t-\mathsf{d}(v,0)}
\end{equation}
where the butterfly velocity $v_{\mathrm{B}}$, which characterizes the Frobenius light cone, is defined as \begin{equation}\label{eq:vB=min_b}
    v_{\mathrm{B}}=2h \min_{b>0}b^{-1}[d+\e^{-b}+(d-1)\e^b].
\end{equation}
\end{theor} 
\begin{proof}
We follow the strategy of Section \ref{sec:babyQW}.  Define \begin{equation}
    \mathcal{F} = \sum_{v\in \mathsf{V}} \mathrm{e}^{b\mathsf{d}(v,0)}\mathbb{P}_v. 
\end{equation}
Now consider the time derivative
\begin{equation}\label{eq:dFdt=sumLP}
    \frac{\mathrm{d}}{\mathrm{d}t} (A|\mathrm{e}^{-\mathcal{L}t}\mathcal{F}\mathrm{e}^{\mathcal{L}t}|A) = -(A(t)| [\mathcal{L},\mathcal{F}]|A(t)) = -\sum_{v\in\mathsf{V}}\mathrm{e}^{b\mathsf{d}(v,0)}  (A(t)|[\mathcal{L},\mathbb{P}_v]|A(t)).
\end{equation}
The only terms that survive the commutator $[\mathcal{L},\mathbb{P}_v]$ are those where either $\mathcal{L}$ annihilates the non-trivial operator on site $v$, \emph{or} creates one:  \begin{equation}\label{eq:LP=LeP}
    [\mathcal{L},\mathbb{P}_v] = (I-\mathbb{P}_v)\mathcal{L}\mathbb{P}_v - \mathbb{P}_v\mathcal{L}(I-\mathbb{P}_v) = \sum_{e\in\partial v} [\mathcal{L}_e,\mathbb{P}_v].
\end{equation}
So \eqref{eq:dFdt=sumLP} becomes \begin{align}
    \frac{\mathrm{d}}{\mathrm{d}t} (A|\mathrm{e}^{-\mathcal{L}t}\mathcal{F}\mathrm{e}^{\mathcal{L}t}|A) &=-\sum_{e=\{u,v\}\in \mathsf{E}}  (A(t)|[\mathcal{L}_e,\mathrm{e}^{b\mathsf{d}(v,0)} \mathbb{P}_v+ \mathrm{e}^{b\mathsf{d}(u,0)} \mathbb{P}_u]|A(t)) \nonumber\\
    &= \sum_{e=\{u,v\}\in \mathsf{E}} \mathrm{e}^{b\mathsf{d}(v,0)} (A(t)| (\BP_v \CL_e (1-\BP_v) \BP_u - \BP_u (1-\BP_v)\CL_e \BP_v)|A(t)) + (u\leftrightarrow v)
\end{align}
Here in the second line, we have used \eqref{eq:LP=LeP} where $(I-\mathbb{P}_v)\mathcal{L}_e\mathbb{P}_v$, for example, gets an extra factor $\mathbb{P}_u$ in front because otherwise $\mathcal{L}_e$ just annihilates the identity operator.

If we define \begin{equation}
    a_v(t) := \sqrt{(A(t)|\mathbb{P}_v|A(t))},
\end{equation}
then observe that for $e=\lbrace u,v\rbrace$:\begin{equation}
    \labs{(A(t)| \BP_v \CL_e (1-\BP_v) \BP_u|A(t))} \le a_ua_v \cdot 2\lVert H_e\rVert \le \left(a_u^2+a_v^2\right) \cdot \lVert H_e\rVert.
\end{equation}
Then \begin{align}\label{eq:dFdf<dF}
    \frac{\mathrm{d}}{\mathrm{d}t} (A|\mathrm{e}^{-\mathcal{L}t}\mathcal{F}\mathrm{e}^{\mathcal{L}t}|A) &\le 2h \sum_{e=\{u,v\}\in \mathsf{E}} \lr{\e^{b \mathsf{d}_u}+\e^{b \mathsf{d}_v} }\left(a_u^2+a_v^2\right) 
    = 2h \sum_{v\in \mathsf{V}} a_v^2 \sum_{u:\mathsf{d}(u,v)=1}  \lr{\e^{b \mathsf{d}_u}+\e^{b \mathsf{d}_v} } \notag \\
    &\le 2h \sum_{v\in \mathsf{V}} a_v^2 \e^{b \mathsf{d}_v} (d+\e^{-b}+(d-1)\e^b) 
    = 2h(d+\e^{-b}+(d-1)\e^b)  (A|\mathrm{e}^{-\mathcal{L}t}\mathcal{F}\mathrm{e}^{\mathcal{L}t}|A).
\end{align}
In the second line's inequality, we have assumed the worst case where $v$ has only one neighbor $u$ that is closer to the initial vertex $0$; all other $d-1$ neighbors are farther. Exponentiating \eqref{eq:dFdf<dF} gives \eqref{eq:vB=min_b}.
\end{proof}

\subsubsection{One dimension}

In one dimension, we can get much stronger bounds by thinking only about the right most site on which an operator acts.  Let us define the projector \begin{equation}
    \mathbb{Q}_r|\lbrace a_i\rbrace) = \mathbb{I}\left[a_j=0 \text { if }j>r, a_r \ne 0 \right] |\lbrace a_i\rbrace). \label{eq:Q_r}
\end{equation}
Hence $\mathbb{Q}_r$ projects onto the rightmost site on which an operator acts (Figure~\ref{fig:Fnorm_1d}). Observe that if we define for any operator $|A(t))$, \begin{equation}
    |\psi_r(t)|^2 = (A(t)|\mathbb{Q}_r|A(t)),
\end{equation}
then a simple modification of the proof of Theorem \ref{thm:qw}, upon defining \begin{equation}
    \mathcal{F} = \sum_r \mathrm{e}^{br}\mathbb{Q}_r,
\end{equation}
exactly reproduces the calculation in Section \ref{sec:babyQW}.  In particular, we find that in one-dimensional models with nearest neighbor hopping (as in Example \ref{exam:1d_tight}), \begin{equation}
    v_{\mathrm{B}} \le 4 \max_{e\in\mathsf{E}}\lVert H_e \rVert.
\end{equation}
This suggests that there may indeed be a qualitative difference between Frobenius and Lieb-Robinson light cones.

We don't know whether there is any one-dimensional model for which the Lieb-Robinson velocity $2\mathrm{e}h \ge v_{\mathrm{LR}}>v_{\mathrm{B}}$.  In fact, upon careful inspection, we can find the velocity $2\mathrm{e}h$ in our Frobenius light cone bound!  Suppose that $r \gg 4ht$, and choose the value of $b$ which minimizes (\ref{eq:4bopt}), where $h$ should multiply by $2$ for the many-body case.  One finds that \begin{equation}
    (A(t)|\mathbb{Q}_r|A(t)) \lesssim \exp\left[ -2r \log \frac{r}{2\mathrm{e}ht}\right]. \label{eq:algebraicfrobenius}
\end{equation}
The only difference between this bound and our earlier Lieb-Robinson bound in Example \ref{exam:1d_tight} is the additional 2 in the exponential.  This arises from the inner product in $(A(t)|\mathbb{Q}_r|A(t)) $ -- the coefficient of $t^r/r!$ is \emph{squared} when evaluating the Frobenius light cone.

\label{sec:matrix_concentration}

\subsection{Entanglement and the Frobenius light cone}
We have seen in Section \ref{sec:entanglement_dyn} that von Neumann entanglement generation could be large even when operator growth is small (outside the Lieb-Robinson light cone).  One might, however, presume that the Frobenius light cone might be closely related to the second R\'enyi entropy (\ref{eq:renyi}), due to the fact that \begin{equation}
    (\rho|\mathbb{P}_{\mathsf{A}}|\rho) = \mathrm{e}^{-S_2(\mathsf{A})}.
\end{equation}
Indeed, \cite{Bentsen:2018uph,harrow2021} discuss how entanglement generation generally appears slower than operator growth; this is also the case in random circuit models (Section \ref{sec:ruc}). However, the following counterexample shows that operator growth may not always precede  entanglement generation.
\begin{exam}[Generating a maximally entangled state with little Frobenius operator growth]\label{exam:frobentangle}
Consider the same setting as Example \ref{ex:S1>} with 
$\epsilon=1-1/D.$ Then, \begin{equation}\label{eq:U0=max}
    U|00\rangle = \frac{1}{\sqrt{D}} \sum_{j=0}^{D-1} |jj\rangle,
\end{equation}
is maximally entangled. Yet for any local operator $ A $, its evolution $\CU A=U A U^\dagger$ satisfies \begin{equation}\label{eq:PF<1/D}
    \norm{\BP_{\sf B} \CU  A }_{\rm F} \le \norm{ A }_{\rm F} \mathrm{O}(D^{-1}).
\end{equation}
\end{exam}
Although this is a somewhat tedious calculation, we produce it in full since it has not appeared in the literature before (to our knowledge). 
\begin{proof}
 Since \eqref{eq:U0=max} comes from the direct calculation, we focus on proving \eqref{eq:PF<1/D}. Because $(1-\BP_{\sf B}) \CU$ is much easier to compute than $\BP_{\sf B} \CU$, we use the equivalence between \eqref{eq:PF<1/D} and \begin{equation}
    \norm{(1-\BP_{\sf B}) \CU  A }_{\rm F} \ge \norm{ A }_{\rm F} [1-\mathrm{O}(D^{-1})].
\end{equation}
Let $\{T_J\}_{J=0}^{D^2-1}$ be a normalized operator basis of $A$, so that the superoperator $(1-\BP_{\sf B}) \CU$ can be represented as a matrix $\CM_{J_1 J}$: \begin{equation}
    (1-\BP_{\sf B}) \CU T_{J_1} = \sum_{J} \CM_{J_1J} T_{J}.
\end{equation}
\eqref{eq:PF<1/D} is then further equivalent to the statement that the eigenvalues of $\CM$ are all of the form $1-\mathrm{O}(D^{-1})$.

We choose $T_J$ to be the matrix-element operators $|j\rangle\langle j'|$, and denote each $J$ by a pair $(j,j')$. We prove that the matrix $\CM$ has the following structure: (\emph{1}) if $j\neq j'$, \begin{equation}\label{eq:MJjj'}
        \CM_{J,jj'} = [1-\mathrm{O}(D^{-2})] \mathbb{I}(J=jj').
    \end{equation}
(\emph{2}) If $j\neq 0$, \begin{equation}\label{eq:Mjjjj}
        \CM_{j_1 j'_1,jj} = \mathbb{I}(j_1=j_1^\prime)\left[ [1-\mathrm{O}(D^{-1})] \mathbb{I}(j_1=j) + \mathrm{O}(D^{-2})\right].
    \end{equation}
For now, assume these properties are true; we now show that $\mathrm{eig}\CM=1-\mathrm{O}(D^{-1})$, and thus \eqref{eq:PF<1/D}. As a result of \eqref{eq:MJjj'}, $\CM$ is diagonal outside the subspace $\Span{(j,j):j=0,\cdots,D-1}$, with the diagonal entries being $1-\mathrm{O}(D^{-2})$. Thus, it remains to verify whether the sub-matrix $\CM^\prime_{j_1j} := \CM_{j_1j_1,jj}$  has $1-\mathrm{O}(D^{-1})$ eigenvalues. According to \eqref{eq:Mjjjj}, the diagonals of $\CM^\prime_{j_1j}$ are also $1-\mathrm{O}(D^{-1})$, while the off-diagonals are $\mathrm{O}(D^{-2})$. Then any eigenvalue of $\CM'_{j_1j}$ is indeed $1-\mathrm{O}(D^{-1})$, from the Gershgorin circle theorem \cite{hornbook}.

Lastly, we verify \eqref{eq:MJjj'} and \eqref{eq:Mjjjj} explicitly. $U$ acts on states by \eqref{eq:U0=max} and \begin{align}\label{eq:Ujj=}
    U\ket{jj} &= (I-\ket{\rm diag} \bra{\rm diag}) \ket{jj} + \lr{-\sqrt{1-\frac{1}{D}}\ket{00} + \sqrt{\frac{1}{D}} \ket{\rm diag} } \braket{{\rm diag} | jj } \nonumber\\
    &= \ket{jj} + {\rm O}(D^{-1/2}) \ket{00} + {\rm O}(D^{-1})\sum_{j'\neq 0} \ket{j'j'},
\end{align}
for $j\neq 0$. Then the action on $T_{jj'}$ ($j\neq j'$) is \begin{align}\label{eq:Ujj'ii}
    \CU \lr{|j\rangle\langle j'|\otimes \sum_{i}|i\rangle\langle i| } &= |j\rangle\langle j'|\otimes \sum_{i}|i\rangle\langle i| + (\mathcal{U}-\mathcal{I}) \left[\ket{jj}\bra{j'j} + \ket{jj'}\bra{j'j'}\right] \nonumber\\
    &= |j\rangle\langle j'|\otimes \sum_{i}|i\rangle\langle i| +  \left[(U-I)\ket{jj}\bra{j'j} + \ket{jj'}\bra{j'j'}(U^\dagger -I)\right],
\end{align}
where we have used $\bra{j'j}U^\dagger = \bra{j'j}$ for example in the second line.
When acting further with $1-\mathbb{P}_{\sf B}$, the first term is unchanged, while only the ${\rm O}(D^{-1}) \ket{jj}$ term of $(U-I)\ket{jj}$ survives ($j'$ similarly), so the result is proportional to $T_{jj'}$: 
\begin{align}\label{eq:1-PB=1/D}
    (1-\mathbb{P}_{\sf B})\CU \lr{|j\rangle\langle j'| } = |j\rangle\langle j'| \mlr{ 1+ {\rm O}(D^{-1}) \frac{1}{D}\tr_{\sf B}(\ket{j}\bra{j}) } = |j\rangle\langle j'| \mlr{ 1+ {\rm O}(D^{-2})}.
\end{align}

Therefore \eqref{eq:MJjj'} holds. Similar arguments also show that the diagonal elements of the sub-matrix $\CM'_{j_1j}$ is $1-{\rm O} (D^{-1})$, so what remains is to prove its off-diagonals are ${\rm O}(D^{-2})$. Similar to \eqref{eq:Ujj'ii}, \begin{equation}
    \CU \lr{T_{jj} } = \cdots + U \ket{jj}\bra{jj} U^\dagger = \cdots + \lr{\alpha_j \ket{jj} + {\rm O}(D^{-1/2}) \sum_{i\neq j}\ket{ii} } \lr{\alpha_j \bra{jj} + {\rm O}(D^{-1/2}) \sum_{i\neq j}\bra{ii} },
\end{equation}
where $\cdots$ only contributes to diagonals, and we have combined \eqref{eq:U0=max} and \eqref{eq:Ujj=}. After expanding the product and project by $1-\mathbb{P}_{\sf B}$, the cross terms $\ket{ii}\bra{i'i'}$ ($i\neq i'$) are eliminated. Furthermore, $\ket{jj}\bra{jj}$ only contributes to diagonals, so off-diagonals only come from ${\rm O}(D^{-1})\sum_i \ket{ii}\bra{ii}$, which gain an extra ${\rm O}(D^{-1})$ similar to \eqref{eq:1-PB=1/D}. Thus the off-diagonals are indeed ${\rm O}(D^{-2})$.
\end{proof}

\subsection{Hamiltonians with random coefficients}\label{sec:random_sum_path}
In Section~\ref{sec:all-to-all}, we will encounter Hamiltonians with random coefficients, such as the Sachdev-Ye-Kitaev model. Applying the deterministic bounds for random Hamiltonian often yields unphysical results:  in practice, operator growth is ``incoherent'', but the Lieb-Robinson bounds use the triangle inequality throughout, which adds terms ``coherently''. Capturing the effects of classical (external) randomness in the Hamiltonian in an operator growth bound has recently become possible.  And, remarkably, it will turn out that these methods are often valuable even when there is no \textit{intrinsic} randomness in the problem!

\subsubsection{Matrix concentration bounds}
We present an instructive example that captures the essential problems with Lieb-Robinson bounds when we are interested in typical state behavior \cite{Chen2021ConcentrationFT}.  Consider
\begin{align}
    H = Z_1 + \cdots + Z_N,
\end{align}
where each Pauli $Z_i$ is supported on qubit $i$. The ``size'' of this matrix depends on the question of interest. The spectral (infinity) norm gives the largest eigenvalue in magnitude:
\begin{align}
    \lnorm{Z_1 + \cdots +Z_N} = N.
\end{align}
The Frobenius norm gives the average magnitude of eigenvalues:
\begin{align}
    \lnormp{Z_1 + \cdots + Z_N}{\mathrm{F}} = \sqrt{N}.
\end{align}
In other words, the worst case is qualitatively different from the average case. In fact, in this problem, the eigenvalue distribution is equivalent to the probability distribution of a sum of independent random variables $ S_N:= x_1+\cdots+x_N$ each drawn from the Rademacher distribution $\mathbb{P}(x_i=1)=\mathbb{P}(x_i=-1)=1/2$. 
Now, we may call a \textit{concentration inequality} to describe how rarely the random variable deviates from its expectation
\begin{align}
\mathbb{P}( \lambda_i \ge \epsilon) \equiv   \mathbb{P}( S_N \ge \epsilon) \le \e^{-\epsilon^2/2N}  \label{eq:hoeff}.
\end{align}
Therefore, the \textit{typical} magnitude of eigenvalues $ \labs{\lambda}= \mathrm{O}(\sqrt{N}) \ll N$ is much smaller than the extreme eigenvalues. 
This simple example illustrates that the ``size'' of high-dimensional objects could behave quite differently depending on the norm; this distinction could lead to drastically different implications (e.g., in power-law interacting systems in Section~\ref{sec:power-law}).

To derive concentration for more complicated matrix functions, we highlight a family of recursive inequalities for their Schatten $p$-norms, which proved extremely versatile.
\begin{prop}[Uniform smoothness for subsystems {\cite{Chen2021ConcentrationFT,ricardXu16}}]\label{prop:unif_subsystem_intro}
Consider matrices $A, B \in \CB(\CH_i\otimes\CH_j)$ that satisfy the so-called non-commutative martingale condition \begin{align}
\tr_i(B) = 0\quad \text{and}\quad A= A_j\otimes I_i.    
\end{align}
Then, for $p \ge 2$, 
\begin{equation}
\lVert A + B\rVert_{p}^2\le \lVert A \rVert_{p}^2  + (p-1)\lVert B\rVert_{p}^2.
\end{equation}
\end{prop}
Remarkably, the martingale condition is compatible with a wide range of matrices beyond independent sums. At the same time, uniform smoothness delivers \textit{sum-of-squares} behavior (analogous to independent sums) that contrasts with the triangle inequality, which is \textit{linear}
\begin{align}
    \norm{A+B} \le \norm{A}+\norm{B}.
\end{align}
This difference underpins the essential distinction between the worst and typical cases. 

For random Hamiltonians, the flavor of the problem changes slightly; we can think of adding independent Gaussian coefficients in our guiding example
\begin{align}
     H = g_1 Z_1 + \cdots + g_N Z_N.
\end{align}
The Gaussian coefficient (i.e., external randomness) requires the following version of uniform smoothness regarding the expected $p$-norm $\vertiii{A}_{p}:= (\BE[ \norm{A}_p^p ] )^{1/p}$ that will allow us to control the spectral norm by setting $p \gtrsim \log(\dim H)\sim N$. 

\begin{prop}[{Uniform smoothness for expected p-norm~\cite[Proposition~4.3]{HNTR20:Matrix-Product}}] \label{prop:uniform_smoothness_expected_p}
Consider random matrices $A, B$ of the same size that satisfy
$\BE[A|B] = 0$. When $2 \le p$,
\begin{equation}
\vertiii{A+B}_{p}^2 \le \vertiii{A}_{p}^2  + (p-1)\vertiii{B}_{p}^2 .
\end{equation}
\end{prop}
Compared with Proposition~\ref{prop:unif_subsystem_intro}, the above contains both classical randomness (the expectation) and quantum randomness (the trace), which is especially suitable for matrices with random coefficients. Historically, uniform smoothness (both Proposition~\ref{prop:unif_subsystem_intro} and Proposition~\ref{prop:uniform_smoothness_expected_p}) is a descendant of the scalar \textit{two-point inequality} or \textit{Bonami's inequality}~\cite{garling_2007}, which features in, e.g., Boolean analysis~\cite{ODonnell2014AnalysisOB}. The matrix version was first derived~\cite{Tomczak1974} and later rewritten in the above form, leading to simple derivations of matrix concentration for martingales~\cite{HNTR20:Matrix-Product,naor_2012}. We expect these robust inequalities to find applications in numerous quantum information settings, by exploiting the tensor product structure of the Hilbert space or by the random coefficients: e.g. when studying power-law interacting systems in Section~\ref{sec:power-law}, Trotter error~\cite{Chen2021ConcentrationFT}, dynamics with random Hamiltonians in Section~\ref{sec:random_sum_path}, or randomized quantum simulation~\cite{chen2020quantum}.

\subsubsection{Bounds based on self-avoiding paths}
Combining matrix concentration inequalities (Section~\ref{sec:matrix_concentration}) with the self-avoiding path (Section \ref{sec:self_avoid}) yields operator growth bounds for random Hamiltonians in the Frobenius norm.

\begin{theor}[Operator growth bounds for random time-independent Hamiltonian~\cite{chen21}]
\label{thm:informal}
Consider a random time-independent 2-local Hamiltonian where the terms $H_e$ are independent for $e\in\mathsf{E}$,  zero mean: $\BE [H_e] =0$, and bounded almost surely: $\lV H_e\rV\le b_e$.
Then, for any normalized operators $\norm{A}=\norm{B} = 1$ supported on subsets $\mathsf{A}$ and $\mathsf{B}$ respectively, the Frobenius norm of the commutator can be bounded by a weighted \textit{incoherent} sum over self-avoiding paths of interactions $\Gamma=\{\Gamma_\ell,\ldots,\Gamma_1\}$, as in \eqref{eq:Self-avoiding}:
\begin{align}
 \BE \left[\lnormp{ [A(t), B ]}{\rm F}^2 \right]\le 4\sum_{\text{self-avoiding paths }\Gamma}\ \int\limits_{\substack{t<t_{\ell}<\cdots <t_1\\ \ell = \labs{\Gamma}}} \mathrm{d}t_\ell\cdots \mathrm{d}t_1 \prod_{1\le k\le \ell} \left(e^{\beta_k(t_{k+1}-t_k)} \frac{8 b_{\Gamma_k}^2}{\beta_k} \right)
\end{align}
where $\beta_{k} = \beta (\Gamma_{\ell},\cdots,\Gamma_{k})$ are tunable parameters that can depend on the path and $t_{\ell+1}:=t$. 
\end{theor}
Compared with the deterministic bounds (Theorem~\ref{thm:self-avoiding}), the expression above has interaction strengths appearing in squares $b_X^2$, entailing the incoherence across different terms. The tunable parameters $\beta_k$ are slightly distracting, but a convenient choice often suffices. The above bound naturally extends to the case with $p$-norms which leads to sharper concentration (Section~\ref{sec:all-to-all}) and Brownian circuits where the randomness is both spatial and temporal~\cite{chen21}.

\section{Systems with all-to-all interactions and holographic quantum gravity}\label{sec:all-to-all}
In this section, we will turn to the study of Hamiltonians that can include few-body interactions between a small number of the $N$ total degrees of freedom at a time.   In particular, we will often focus on $k$-local models of the form \begin{equation}
    H = \sum_{j=1}^k \sum_{i_1<i_2<\cdots < i_k} h^{a_1\cdots a_k}_{i_1\cdots i_k}(t) X_{i_1}^{a_1}\cdots X_{i_k}^{a_k} \label{eq:klocalH}
\end{equation}
that couple together $N$ qubits.  Here we assume that each $h^{a_1\cdots a_k}_{i_1\cdots i_k}(t) = \mathrm{O}(1)$.   Such models are the extreme opposite of the spatially local models we have studied thus far, and their study will lead to the ultimate limits on what quantum systems can achieve using few-body interactions.  There is hope that a future quantum computer might be able to implement such generic kinds of models (at the cost of many local qubits used for teleportation, as in Section \ref{sec:measurement}), although one may also wish to use photons \cite{leroux} or trapped ions \cite{Britton2012} to realize non-local couplings.  However (\ref{eq:klocalH}) is achieved, we come across interesting theoretical questions both in quantum information and, interestingly enough, in high energy physics.  In this section, we will describe how Lieb-Robinson and Frobenius bounds can be generalized to such systems, despite the lack of spatial locality.

\subsection{Fast state preparation}

We first discuss preparing globally entanglement many-body states from product states. We consider the GHZ state $\ket{\rm GHZ}$ \eqref{eq:GHZ} and the W state \cite{W_state} \begin{equation}
    \ket{\rm W} := \frac{1}{\sqrt{N}}\lr{\ket{10\cdots 0} + \ket{010\cdots 0} + \cdots + \ket{0\cdots 01}}, \label{eq:W}
\end{equation}
as two examples. The fastest protocols, to our knowledge, are the following:

\begin{exam}[A ${\rm O}(1)$-time GHZ state preparation protocol]\label{exam:GHZ_all2all}
The Hamiltonian \begin{equation}\label{eq:H_GHZ}
    H = (I-Z_1) \sum_{j=2}^N (I-X_j),
\end{equation}
prepares GHZ state in time $t=\pi/4$, starting from $\ket{+}_1\otimes \ket{\bm 0}_{2\cdots N}$ where $\ket{+}=(\ket{0}+\ket{1})/\sqrt{2}$.
\end{exam}
\begin{proof}
Observe that \begin{equation}
    \ket{\rm GHZ} = U_{1N}\cdots U_{13}U_{12} \ket{+}_1\otimes \ket{\bm 0}_{2\cdots N},
\end{equation}
where \begin{equation}
    U_{1j} = \e^{-\ri \frac{\pi}{4}(I-Z_1)(I-X_j)},
\end{equation}
is the controlled-NOT (CNOT) gate on qubit $1$ and $j$. Since all $U_{1j}$s commute, they combine to $U_{1N}\cdots U_{13}U_{12}=\e^{-\ri \pi H/4}$ with $H$ given in \eqref{eq:H_GHZ}.
\end{proof}

\begin{exam}[A ${\rm O}(N^{-1/2})$-time W state preparation protocol \cite{strongly_gorshkov}]\label{exam:GHZ}
Consider Hamiltonian \begin{equation}
    H = \ri X^-_1 \sum_{j=2}^N X^+_j + (\rm H.c.),
\end{equation}
where $X^\pm := (X\pm \ri Y)/2$.
Starting from $\ket{1}_1\otimes \ket{\bm 0}_{2\cdots N}$, we obtain the W state after time $t=\frac{1}{\sqrt{N-1}}\arccos(1/\sqrt{N})$.
\end{exam}

\begin{proof}
By direct calculation, \begin{align}
    H\ket{1}_1\otimes \ket{\bm 0}_{2\cdots N} &= \ri \ket{0}_1\otimes \sum_{j=2}^N X^+_j \ket{\bm 0}_{2\cdots N} = \ri \sqrt{N-1} \ket{0}_1\otimes \ket{\rm W}_{2\cdots N}, \\
    \quad H\ket{0}_1\otimes \ket{\rm W}_{2\cdots N} &= -\ri\ket{1}_1\otimes \sum_{j=2}^N X^-_j \ket{\rm W}_{2\cdots N} = -\ri\sqrt{N-1} \ket{1}_1\otimes \ket{\bm 0}_{2\cdots N}.
\end{align}
Thus, $H/\sqrt{N-1}$ acts as the Pauli Y matrix in $\Span{\ket{1}_1\otimes \ket{\bm 0}_{2\cdots N}, \ket{0}_1\otimes \ket{\rm W}_{2\cdots N}}$. Since for a 2-level system \begin{equation}
    \e^{-\ri \theta Y}\ket{0} = \cos\theta \ket{0} + \sin \theta \ket{1},
\end{equation}
we choose $\sqrt{N-1}t = \theta = \arccos(1/\sqrt{N})$ so that \begin{equation}
    \e^{-\ri t H}\ket{1}_1\otimes \ket{\bm 0}_{2\cdots N} = \frac{1}{\sqrt{N}} \ket{1}_1\otimes \ket{\bm 0}_{2\cdots N} + \sqrt{\frac{N-1}{N}} \ket{0}_1\otimes \ket{\rm W}_{2\cdots N} = \ket{\mathrm{W}},
\end{equation}
which indeed prepares the W-state.
\end{proof}

See \cite{GHZ_all2all19} for another GHZ protocol with ${\rm O}(1)$ time. 
It is an open question whether the above protocols are asymptotically optimal in terms of the $N$ scaling of time $t$. There is a separation between them and the best-known lower bound. For example, the operator growth bound on entanglement generation in Proposition \ref{prop:P>S2} yields $t={\rm \Omega}(1/N)$ for preparing GHZ. The following Proposition achieves a $\log N$ factor improvement by assuming a local qubit of information (instead of a single product state) is encoded in the $N$-qubit repetition code:
\begin{equation}\label{eq:U_encode_GHZ}
    U \lr{\alpha \ket{0} +\beta \ket{1}}_1 \otimes \ket{\bm 0}_{2,\cdots,N} = \alpha\ket{\bm 0} + \beta \ket{\bm 1},
\end{equation}
Note that Example \ref{exam:GHZ_all2all} fulfills this condition.

\begin{prop}[A lower bound on preparing GHZ using all-to-all interaction]
\label{prop:ghzprep}
For any $H$ of the form (\ref{eq:klocalH}) with $k=2$, if $U=\e^{-\ri t H}$ prepares GHZ in the sense of \eqref{eq:U_encode_GHZ}, then \begin{equation} \label{eq:t>logN/N}
    t={\rm \Omega}\lr{\frac{\log N}{N}}.
\end{equation}
\end{prop}

\begin{proof}
We prove for 2-local Hamiltonians.  Similar to Proposition \ref{prop:transfer norm}, we have \begin{equation}\label{eq:GHZ_XZ=2}
    \left\lVert  \comm{ U X_1 U^\dagger  }{ Z_2} \right\rVert   = 2,
\end{equation}
because $U X_1 U^\dagger$ and $Z_2$ are logical operators for the repetition code. Note that \eqref{eq:GHZ_XZ=2} holds for any $Z_j$, while only one $Z_{j\neq 1}$ is sufficient to yield our bound. 

We want to utilize Theorem \ref{thm:sum_of_paths} since \eqref{eq:GHZ_XZ=2} implies $C_{12}(t)=1$. The matrix $h$ defined by \eqref{eq:huv=Huv} takes the form $h=a J + b N I$, where $a,b$ are ${\rm O}(1)$ constants, and $J$ is the $N$-by-$N$ matrix with all elements equal to $1$. Note that we do not need $b=-1/N$ from the self-avoiding path techniques in Section \ref{sec:self_avoid}, since it does not change the final scaling. 
The matrix $h$ can be easily diagonalized: it has eigenvalue $aN+bN$ for the state $\ket{\psi}=\lr{1,\cdots,1}^{\rm T}/\sqrt{N}$, and eigenvalue $bN$ for all other orthogonal states. Then Theorem \ref{thm:sum_of_paths} yields \begin{align}
    C_{12}(t) \le \lr{\e^{2th}}_{12} = \e^{2t(a+b)N} \braket{2|\psi} \braket{\psi|1} + \e^{2tbN} \braket{2|(I-\ket{\psi}\bra{\psi})|1} = \frac{1}{N} \lr{\e^{2t(a+b)N}-\e^{2tbN} },
\end{align}
which implies that \eqref{eq:t>logN/N} holds whenever $C_{12}(t)=1$, as it must be to prepare the GHZ state. Here we have used the smallness of $\braket{2|\psi}= \braket{\psi|1}=1/\sqrt{N}$ to get the extra $\log N$ factor.
\end{proof}

A stronger bound $t={\rm \Omega}(N^{-1/2})$ holds for Frobenius norm growth \cite{yin20}, in comparison to the operator norm growth \eqref{eq:GHZ_XZ=2}. However, it is not obvious such operator growth should be relevant for state preparation.  While naively, the logical $X_1$ must grow to $X_1\cdots X_N$ after time evolution, it only must have such a long Pauli string on states stabilized by $Z_1Z_2$, etc., meaning that we could have $X_1(t) = X_1\cdots X_N \frac{1+Z_1Z_2}{2}\cdots + \cdots $ which has exponentially small Frobenius weight.  Making progress on this question is an important open problem.

Finally, one can also consider circuit models (instead of Hamiltonian models) with all-to-all connectivity. Assuming each qubit is acted on by at most one local gate at each time step (i.e., do not parallelize overlapping yet commuting gates as in Example \ref{exam:GHZ}), the depth of the circuit needs to be ${\rm \Omega}(\log N)$ for preparing any state that is globally entangled. The reason is simply that each qubit needs to build up correlation with all other qubits, yet the strict light cone is restricted in qubits of number exponential in depth. ${\rm \Theta}(\log N)$-depth circuits are known for GHZ and W states. GHZ is prepared simply by inductively applying CNOT gates to double the qubits sharing the GHZ. The W state can be prepared in a similar fashion \cite{W_all2all_logN18}. We refer to the literature \cite{Dicke_all2all22,sparse_all2all,Stephen:2022dmt} for further discussions. 

\subsection{Lyapunov exponents, quantum chaos, and operator growth}\label{sec:lyapunov}
In the previous section, we discussed models with all-to-all interactions as valuable for the fast preparation of interesting entangled states; however, the Hamiltonians involved are not thermodynamically extensive: $\lVert H\rVert \sim N^k$ for a $k$-local Hamiltonian. In this section, we will describe random models with all-to-all Hamiltonians, but which are thermodynamically extensive: $\lVert H\rVert = \mathrm{O}(N)$.  As we will see, this does not simply mean dividing by $N^{k-1}$ in (\ref{eq:klocalH}). 

A paradigmatic model to study is the Sachdev-Ye-Kitaev (SYK) model \cite{sachdevye,maldacena16,kitaev17} of $N$-interacting Majorana fermions: $N$ operators chosen to obey the \textit{anti-commutation} relation\footnote{There is a simple way, the Jordan Wigner transform, to present $2N$ Majorana fermion operators in terms of Pauli matrices: $X_1 = \psi_1$, $Y_1 = \psi_2$, $Z_1 X_2 = \psi_3$, $\ldots$, $Z_1\cdots Z_{N-1}Y_N = \psi_{2N}$.  Note however that a Hamiltonian which is $k$-local in terms of Majorana fermions may be $N$-local written in terms of Pauli matrices.  We will not spell out in this review, but it is straightforward to show that all of the notions of locality, operator size etc., continue to make sense in a Hamiltonian written in terms of Majorana fermions.} \begin{equation}
    \lbrace \psi_i, \psi_j\rbrace = 2\mathbb{I}(i=j).
\end{equation}
We then consider the random Hamiltonian \begin{equation}
    H = \frac{\ri^{q/2}}{N^{(q-1)/2}} \sum_{i_1<i_2<\cdots <i_q} J_{i_1\cdots i_q} \psi_{i_1}\cdots \psi_{i_q},\label{eq:SYK_H}
\end{equation}
where $J_{i_1\cdots i_q}$ are independent and identically distributed (i.i.d.) random variables with variance 
\begin{equation}
    \mathbb{E}\left[J_{i_1\cdots i_q}^2 \right] = \frac{1}{2q} \left(\begin{array}{cc} N-1 \\ q-1 \end{array}\right)^{-1}.
\end{equation}
Note that the model is chosen so that the maximal and minimal eigenvalues of $H$ scale linearly with $N$ (thermodynamic extensivity) \cite{Hastings2021OptimizingSI,Herasymenko2022OptimizingSF}.

The fast scrambling conjecture~\cite{Sekino:2008he} asserts that the time it takes for an operator to ``grow large" scales should as $\sim\log N$. A cartoon model for this conjecture comes from considering random circuit dynamics (Section \ref{sec:ruc}), in which we apply one interaction term in each discrete time step.  In the first interaction, a single Majorana fermion grows as \begin{equation}
    \psi_1 \rightarrow [\psi_1\cdots \psi_q, \psi_1] = \psi_2 \cdots \psi_q.
\end{equation}
The operator has grown from size 1 to size $q-1$, because the Hamiltonian involves $q$ Majorana fermions.  In the next time step, each of these $q-1$ seeds can also grow into $q-1$ new fermions, and so on; after discrete time $t\in \mathbb{Z}^+$, we then estimate that $\text{size} \sim (q-1)^t$, and the time it takes for the size to scale as $N$ is then $t\sim \log N$ \cite{Lashkari:2011yi}.

There is a remarkable analogy between this cartoon of operator growth, and classical chaos, where one studies some number of degrees of freedom $x_j(t)$ governed by deterministic equations, one often finds that \begin{equation}
    \frac{\partial x_i(t)}{\partial x_j(0)} \sim \mathrm{e}^{\lambda_{\mathrm{L}} t},
\end{equation}
where the exponent $\lambda_{\mathrm{L}}$ is called the Lyapunov exponent.  While we caution that this exponential growth can also arise from saddle point instabilities \cite{Xu:2019lhc}, it is often associated with the onset of chaotic and irregular behavior (often known as the butterfly effect).  In quantum mechanics, it was first noted a long time ago \cite{larkin} that there is a natural analogy between $\partial x(t)/\partial x(0)$ and a commutator $[x(t),p(0)]$ (here $p$ denotes the momentum operator).   In recent years, this analogy has been extended to discrete systems via the study of  \textbf{out-of-time-ordered correlators} (OTOC) at infinite temperature between two small operators $A_i$ and $B_j$, acting on single qubits $i$ and $j$: \begin{equation}
    \norm{[A_i(t),B_j]}_{\mathrm{F}}^2=\frac{\tr\left([A_i(t),B_j]^\dagger[A_i(t),B_j]\right) }{\tr(I)} \le \frac{1}{N} \mathrm{e}^{\lambda_{\mathrm{L}}t}. \label{eq:otocexplicit}
\end{equation}
This correlation function is called out-of-time-ordered because, in ordinary many-body quantum physics, one usually studies either time-ordered or anti-time-ordered correlation functions (e.g. when studying linear response in thermal systems).  To actually evaluate the correlation function in (\ref{eq:otocexplicit}) in an ``experiment", one needs to evolve both forwards and backward in time when evaluating $\langle \psi | A_i(t)B_jA_i(t)B_j |\psi\rangle$ in any state $|\psi\rangle$.  As such, these are not physically accessible correlation functions in most experiments, although nuclear magnetic resonance has been able to achieve such a task for a long time (usually in the study of fairly simple Hamiltonians) \cite{baum,Garttner:2016mqj}: see also \cite{swingle16,Landsman:2018jpm} for some recent experimental proposals and progress.  Nevertheless, they have been of some interest in the past decade, as we will discuss more in Section \ref{sec:qg}.

Indeed, we have already seen OTOCs in Section \ref{sec:frobenius} as a natural probe of the Frobenius light cone.  Is it possible that there is a sharper analogy between classical chaos, and exponential growth of OTOCs? We have already seen in (\ref{eq:algebraicfrobenius}) that in local quantum systems, OTOC growth may be algebraic:  $[A_0(t),B_r] \sim t^{2r}$, rather than exponential.  Therefore, genuine exponential growth in OTOCs should only be expected in special circumstances.  One such setting is in systems with either a semiclassical degree of freedom, such as an infinite-dimensional boson, or a large-$S$ spin model \cite{Rozenbaum:2016mmv,Lewis-Swan:2018sdr,monikass,alavirad}: see \cite{yin20,yin21} for some rigorous results on engineered models of fast scramblers \cite{Li:2020zuj,Belyansky:2020bia}.  Alternatively, one can study systems with all-to-all interactions among $N$ finite-dimensional systems, such as Majorana fermions or qubits, e.g. SYK. To the extent that there is an honest period of exponential growth in out-of-time-ordered correlators in local models, it is likely only the case when there are perturbatively weak interactions in quantum field theory \cite{Stanford:2015owe,Grozdanov:2018atb}; see also the cartoon circuit model of \cite{keselman}.

In systems involving all-to-all interactions between $N$ local degrees of freedom, we will define the Lyapunov exponent in quantum mechanics as the growth rate of operator size:\footnote{Note that relative to the classical definition, there is a factor of 2 mismatch.  In the quantum chaos literature, however, the normalization here is fairly standard.} 
\begin{equation}
   (A(t)|\mathcal{S}|A(t)) \approx \mathrm{e}^{\lambda_{\mathrm{L}}t}, \;\;\; (1\ll \lambda_{\mathrm{L}}t \ll \log N ).
\end{equation}
Using Proposition \ref{prop:operatorsize}, we can relate this to the \emph{typical} value of an OTOC between $A$ and a randomly chosen single-site operator.

Let us now return to the SYK model~(\ref{eq:SYK_H}). From the proof of Proposition \ref{prop:ghzprep}, we find \begin{equation}
    \tr\left(\lbrace \psi_1(t),\psi_2\rbrace^2\right) \le \frac{1}{N}\left(\exp\left[c N^{(q-1)/2} t \right] -1 \right).
\end{equation}
Because the exponent scales with $N$, we cannot prove the fast scrambling conjecture for the SYK model. This is due to the requirement that $H$ be extensive for \emph{typical} states, in contrast to the Lieb-Robinson bound (which assumes worst-case scaling of $\lVert H\rVert$, and is saturated only by commuting Hamiltonians). Microscopic calculations, in contrast, do find a finite Lyapunov exponent \cite{Roberts:2018mnp}. New techniques have been developed to describe operator growth in such systems by incorporating concentration bounds from probability theory \cite{Lucas_2020,chen21}.  
\begin{theor}[(Heuristic statement) Fast Scrambling bound]\label{theor:fastscrambling}
For fixed $q$ and sufficiently large $N$, there exists a constant $c=\mathrm{O}(1)$ such that (\ref{eq:otocexplicit}) holds for times \begin{equation}
    |t| < c \log N.
\end{equation}
Therefore, the Lyapunov exponent is constant $\lambda_{\mathrm{L}}=\mathrm{O}(1)$.
\end{theor}
\begin{proof}[Proof idea]
The first proof was done by brute force combinatorics first in \cite{Lucas_2020}, and later by more elegant methods for general random Hamiltonians in \cite{chen21}.  Both proofs rely heavily on the notion of concentration bounds for random systems, while \cite{chen21} also introduced matrix concentration methods (Section~\ref{sec:matrix_concentration}) for many-body Hamiltonians. The latter was stated in Section \ref{sec:random_sum_path}.

To give some brief conceptual intuition, one can think of the Lyapunov exponent as being finite because the growth of operator size is at best exponential (as in our cartoon above).  A way to make this rigorous is to show that \cite{Lucas_2020} \begin{equation}
    \lVert \mathbb{Q}_s \mathcal{L} \mathbb{Q}_{s^\prime} \rVert \le \frac{\lambda_{\mathrm{L}}\max(s,s^\prime)}{q-2}\mathbb{I}(|s-s^\prime|<q-2). \label{eq:cq2}
\end{equation}
where $\mathbb{Q}_s = \mathbb{I}(\mathcal{S}=s)$ is the projection onto operators of size $s$, and $\lambda_{\mathrm{L}}=\mathrm{O}(1)$.  One shows that (\ref{eq:cq2}) holds with extremely high probability in the SYK model, demonstrating non-perturbatively the scaling obtained diagrammatically in \cite{Roberts:2018mnp}.  If (\ref{eq:cq2}) holds, then we find that \begin{equation}
    \frac{\mathrm{d}}{\mathrm{d}t}(A(t)|\mathcal{S}|A(t)) = (A(t)|[\mathcal{S},\mathcal{L}]|A(t)) \le \lambda_{\mathrm{L}} (A(t)|  \mathcal{S} | A(t)).
\end{equation}
Solving this differential inequality leads to (\ref{eq:otocexplicit}).
\end{proof}

In general, the proof of \cite{Lucas_2020} relies on the randomness of the coupling constants in order to get strong bounds on $\lambda_{\mathrm{L}}$.   In the less rigorous physics literature, the randomness of coupling constants can be thought of as helping to organize an expansion of calculations in terms of Feynman diagrams, with dangerous loop diagrams suppressed at large $N$.   There are also non-random quantum systems where this same diagrammatic suppression can occur.  A particularly relevant example is in matrix quantum mechanics, which is relevant to holography since the matrix degrees of freedom represent the start and end points of open strings on brane stacks \cite{Banks:1996vh}.   A recent ``cartoon matrix model" \cite{Lucas:2020pgj} has also confirmed that such models exhibit analogs of Theorem \ref{theor:fastscrambling}.

There is a notion of ``Krylov complexity" which has been introduced \cite{altman}, that also appears to compute the Lyapunov exponent $\lambda_{\mathrm{L}}$, yet appears distinct from operator size.

\subsection{Random quantum dynamics}\label{sec:ruc}
In this section, we briefly review a particularly simple limit of operator dynamics where the problem has random \emph{spacetime} evolution.   In this case, the problem of operator growth reduces to a completely classical problem.

The minimal model for operator growth is the random unitary circuit \cite{Nahum:2016muy,Nahum:2017yvy,vonKeyserlingk:2017dyr,Fisher:2022qey}.   We will focus on the dynamics in a model with nearest-neighbor interactions in one spatial dimension, but the discussion straightforwardly generalizes to arbitrary graphs.  As shown in Figure \ref{fig:ruc}, we consider time evolution in discrete time steps: $U=\cdots U(2)U(1)$, where \begin{subequations}\label{eq:ruc}
\begin{align}
    U(1) &= U_{12}(1)U_{34}(1)\cdots , \\
    U(2) &= U_{23}(2)U_{45}(2)\cdots,
\end{align}
\end{subequations}
and so on.  Here $U_{ij}$ denotes a 2-local randomly chosen unitary matrix acting on sites $i$ and $j$. We choose -- at every time step -- the $U_{ij}$ from the Haar distribution, which means that we choose a unitary matrix uniformly from all possible choices.  The key is that since every unitary shows up exactly once, it is easy to perform time averages.

\begin{figure}[t]
    \centering
    \includegraphics[width=3in]{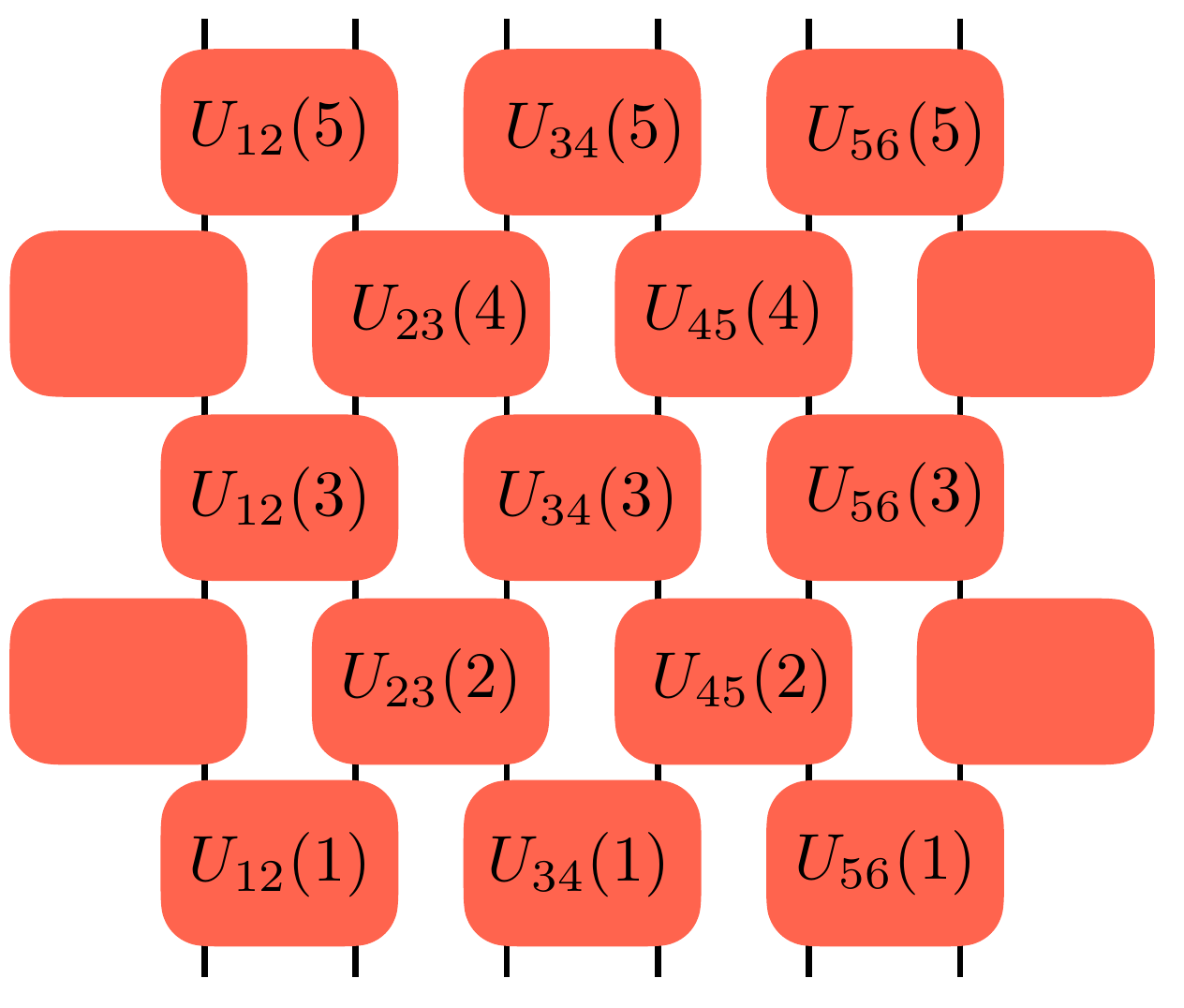}
    \caption{A spacetime cartoon of the random unitary circuit (\ref{eq:ruc}): each gate is usually taken to be Haar random and is independently chosen at each discrete time step.  Six sites of a longer one-dimensional chain are depicted. 
    }
    \label{fig:ruc}
\end{figure}

To see how, let us further restrict to the case where each site has a single qubit on it. Let us consider how a Pauli matrix $X_1$ grows under the first time evolution $U(1)$: \begin{align}
    U(1)^\dagger X_1 U(1) &=  U_{12}^\dagger U_{34}^\dagger \cdots X_1 \cdots U_{34}U_{12}= U_{12}^\dagger X_1 U_{12} \notag \\
    &= \sum_{a_1=1}^3 c_{a_1}X_1^{a_1} + \sum_{a_2=1}^3 c_{a_2}X_2^{a_2} + \sum_{a_1,a_2=1}^3 c_{a_1a_2}X_1^{a_1}X_2^{a_2}.
\end{align} 
Since $X_1$ is local, the only unitary that evolves it non-trivially is $U_{12}$.  If $U_{12}$ is Haar random, all operators acting on the set $\lbrace 1,2\rbrace$ that square to $I$ are equally likely: in particular, this means that \begin{equation}
    \mathbb{E}\left[c_{a_1}^2\right] = \mathbb{E}\left[c_{a_1}^2\right] = \mathbb{E}\left[c_{a_1a_2}^2\right] = \frac{1}{15},
\end{equation}
while any two coefficients are uncorrelated.

It is helpful to adopt a ``super-superoperator" picture in which \begin{equation}
    \mathcal{U}(0)|X_1] = \sum_{a_1,a_1^\prime=1}^3 c_{a_1}c_{a_1^\prime} |X_1^{a_1})(X_1^{a_1^\prime}| + \cdots,
\end{equation}
since in this picture it is easy to average over $U_{12}$: \begin{equation}
    \mathbb{E}\left[\ |X_1(1)]\ \right] = \frac{1}{15}\left(\sum_{a=1}^3 \left(|X_1^a)(X_1^a| + |X_2^a)(X_2^a| \right) + \sum_{a,b=1}^3 |X_1^aX_2^b)(X_1^aX_2^b|  \right). \label{eq:circuitoperatorgrowth}
\end{equation}
The only terms that show up in this average are ``diagonal" in a Pauli basis, meaning that the resulting dynamics can be mapped to a classical stochastic process, essentially corresponding to what sites have a non-trivial Pauli on them.  This allows for large-scale numerical simulations of the resulting dynamics.  Note that when one considers non-Haar random dynamics, such as dynamics constrained by a conservation law, there is no longer a simple picture for operator dynamics in general \cite{Khemani:2017nda,Rakovszky:2017qit,xiaorahul}, as the coefficients $c$ become correlated and off-diagonal terms in (\ref{eq:circuitoperatorgrowth}) cannot be ignored.

Random unitary circuits are useful as toy models for a broad range of problems, as they often illustrate more general phenomena.  One such phenomenon that we have already encountered is a discrepancy between a Lieb-Robinson light cone and a Frobenius light cone.  In units of lattice spacing, the Lieb-Robinson velocity in the circuit of Figure \ref{fig:ruc} is $v_{\mathrm{LR}}=1$, since in principle there exists a circuit where $X_1$ evolves to $X_1\cdots X_L$ at discrete time $t=L-1$.  However, in a typical circuit, one finds the operator is supported with overwhelming probability inside a smaller domain \cite{Nahum:2017yvy,vonKeyserlingk:2017dyr}: this corresponds to a butterfly velocity (see Section \ref{sec:frobenius}) $v_{\mathrm{B}}<1$.  In these models, $v_{\mathrm{B}}$ can (in 1d) be analytically computed.

It is straightforward to extend this discussion to continuous time dynamics.  In this case, as one example, one considers a \textbf{Brownian Hamiltonian} of the form \begin{equation}
    H(t) = \sum_{e\in\mathsf{E}}h_e(t)A_e(t),
\end{equation}
where $A_e$ is an operator acting on edge $e$, while the couplings $h_e(t)$ are taken to be (usually) Gaussian white noise obeying: \begin{equation}
    \mathbb{E}\left[h_e(t)h_{e^\prime}(t^\prime)\right] = \delta_{ee^\prime}\delta(t-t^\prime).
\end{equation}
In this case, again one finds that the operator growth problem reduces to a classical continuous-time Markov process: see e.g. \cite{Lashkari:2011yi,Lucas:2019aif}. 

\subsection{Possible connections to quantum gravity and black holes}\label{sec:qg}
We now briefly discuss the historical origin of the interest in quantum operator growth and many-body chaos.   While the observation that OTOCs should probe chaos was first noticed many decades ago \cite{larkin}, the OTOC became a much more intensely studied object when it was noticed to possibly relate to quantum theories of gravity.

A bit of background is in order.  About 25 years ago, it was first noticed \cite{Maldacena:1997re} that certain quantum field theories (QFT) which arise in string theory appear to admit a description in terms of gravity in one higher dimension.  It is widely believed (albeit an open conjecture) that this \textbf{holographic duality} provides our first explicit model of quantum gravity, and quantum black holes: see \cite{lucasbook} for a review on the subject.  Here, we will focus on a peculiar aspect of the duality:  a black hole in the gravity theory is related to a thermal state in the QFT.  Suppose, from the gravity side, we toss in a small particle into the black hole.  In the QFT, this is interpreted as applying a spatially local operator such as $A_i$.  Under time evolution in the QFT, we expect that $A_i(t)$ becomes a large size operator as discussed above, while in the gravity picture, the particle is falling towards the black hole horizon \cite{Shenker:2013pqa}.  It was therefore proposed in \cite{Susskind:2018tei,Brown:2018kvn} that there must be a relation between the particle motion in gravity and operator size; this conjecture has some evidence for it in SYK \cite{Qi:2018bje} but remains open more generally.   See \cite{Brown:2015bva,Haehl:2021emt} for related ideas about the black hole interior and quantum complexity.

For the purposes of this review, the key feature of this story is that black hole states are \emph{finite temperature states} -- in fact, the temperature $T \ll J$ must actually be very low compared to any microscopic energy scale $J$ appearing in the Hamiltonian!  If $T\sim J$, the black hole becomes large enough that it is not believed any semiclassical description of gravity exists.  This is parametrically the opposite regime of the one studied in Section \ref{sec:lyapunov}, where rigorous results on Lyapunov exponents are known.  At finite temperature, it is believed that OTOCs of the form \begin{equation}
    \tr \left(\sqrt{\rho} [A_i(t),B_j]^\dagger \sqrt{\rho} [A_i(t),B_j]\right) \lesssim \frac{1}{N} \mathrm{e}^{2\mpi T t}
\end{equation}
have a Lyapunov exponent $\lambda_{\mathrm{L}} \le 2\mpi T$ \cite{Maldacena:2015waa}.  While it seems many physicists believe that this conjecture is proven, the ``proof" relies on crucial (but reasonable) assumptions about typical thermal/chaotic systems which are not proven.   In our view, it would be extremely valuable to prove a bound on the Lyapunov exponent at finite temperature, assuming only $k$-locality of the Hamiltonian.   Such an achievement would have key implications for quantum gravity since it is known that the holographic models of semiclassical gravity actually saturate the conjectured bound on the Lyapunov exponent \cite{Shenker:2013pqa}.  A rigorous proof of this bound would imply at least one sense in which black holes are the ``fastest scramblers" -- more explicitly, they would be the quantum systems that can most effectively \textbf{scramble} information contained in a local operator $A_i$ into a non-local operator $A_i(t)$.  We expect a rigorous proof of any such bound to be an extreme challenge; Section \ref{sec:finiteT} describes the limited results known to us in finite temperature models.

While we have focused on a particular class of random $k$-local Hamiltonians in this review, we emphasize that there has been a large literature on OTOCs and operator growth \cite{Aleiner:2016eni,Patel:2017vfp,kivelson,Chowdhury:2017jzb} in spatially local Hamiltonians as well, including via holographic duality \cite{Roberts:2014isa,Blake:2016wvh,Roberts:2016wdl} and ``SYK chains" \cite{Gu:2016oyy}.

\section{Systems with power-law interactions}\label{sec:power-law}

While the celebrated Lieb-Robinson bounds have found profound applications in the past decades, one fundamental caveat is that many physical systems are not, strictly speaking, local but with a power-law decay tail (Table~\ref{table:powerlawsystems}).

For simplicity, our discussion will focus on models defined on some $d$-dimensional graph of vertex set $\mathsf{V}$ (Section \ref{sec:graphreview}).   However, we emphasize that the interactions will not only couple sites connected by an edge of the graph.  For this section only, we define a Hamiltonian with two-body (2-local) interactions to have power-law interactions of exponent $\le \alpha$ whenever \begin{equation}
    H = \sum_{i\ne j\in \mathsf{V}} h_{ij}(t) H_{ij}
\end{equation}
such that \begin{equation}
    h_{ij}(t) \le \frac{C}{\mathsf{d}(i,j)^\alpha} \label{eq:hijpowerlaw}
\end{equation}
where $H_{ij}$ acts non-trivially only on sites $i$ and $j$, with $\lVert H_{ij}\rVert = 1$, and $0<C<\infty$ is an absolute constant.
When we say that a system has power-law interactions of exponent $\alpha$, we mean that $\alpha$ is the largest possible exponent for which this definition holds.

\begin{table}[t]
    \centering
    \begin{tabular}{|l|l|}\hline
       charged particles  & $\alpha=1$ (Coulomb) \\\hline
       electrical dipoles, polar molecules \cite{Yan2013}  & $\alpha=3$ \\\hline
       neutral atoms, Rydberg atom arrays \cite{Saffman2010} & $\alpha=6$ (van der Waals) \\\hline
       trapped ion crystals \cite{Britton2012} & $0\lesssim \alpha \lesssim 3$ (approximate) \\\hline
    \end{tabular}
    \caption{A summary of physically realized systems with power-law interactions of various exponents $\alpha$. All are related to the nature of electromagnetism in three spatial dimensions!}
    \label{table:powerlawsystems}
\end{table}

\subsection{Lieb-Robinson light cone}
While an understanding of quantum dynamics with power-law interactions will thus be relevant for a broad range of physical systems and quantum information processing platforms, it is surprisingly challenging to reasonably extend the Lieb-Robinson Theorem to power-law interacting systems. Let us briefly review the history of bounds on the commutator $\lVert [A_0(t),B_r]\rVert := C(r,t)$, between two local operators separated by a distance $r$.  In 2005, it was shown \cite{Hastings_koma} that for any $\alpha>d$, we have\begin{equation}
    C(r,t) \le C_0\frac{\mathrm{e}^{\lambda t}-1}{r^\alpha}.\label{eq:hastings_koma_logr}
\end{equation}
The proof essentially follows that of Theorem \ref{thm:LRexpTail0}.
Since $C(r,t) \sim 1$ when $t\sim \log r$, we say that this system is bounded by a ``logarithmic light cone". Consequently, we can disprove that, e.g., a Bell pair can be created between sites separated by distance $r$ before this time. Such a bound is no tighter than what we saw for systems with all-to-all interactions!   A decade later, \cite{fossfeig,tran2019,else} found a much tighter bound:  for $\alpha>2d$, we have\begin{equation}
    C(r,t) \le C_0 \left(\frac{t}{r^\kappa}\right)^\beta, \label{eq:powerlaw05}
\end{equation} 
where $\beta>0$ and $0<\kappa<1$ are constants depending on $\alpha, d$ (see the references for the precise values). These bounds demonstrated an exponential improvement, but do not yield a linear light for any large values of $\alpha$.

In the past four years, the shape of the ``light cone" has been essentially tightly established \cite{alpha_3_chenlucas,strictlylinear_KS,tran2021}, as summarized below.
\begin{theor}[Lieb-Robinson bounds for power-law interacting systems]
\label{thm:lrpowerlaw}
Consider $r,t$ large with $t/r^{\min(1,\alpha-2d)}$ fixed.  For $\alpha>2d+1$, the shape of the light cone is linear~\cite{strictlylinear_KS} \begin{equation}
    C(r,t) \le \frac{C}{(r-v_{\mathrm{LR}}t)^\alpha}
\end{equation} 
For $2d<\alpha<2d+1$, one instead finds an algebraic light cone~\cite{tran2021} \begin{equation}
    C(r,t) \le C_0 \left(\frac{t}{r^{\alpha-2d}}\right)^{\beta}.
\end{equation}
The constants $C_0, \beta, v_{\mathrm{LR}}$ depend only on $\alpha, d$. If $d<\alpha<2d$, then the shape of the light cone is polylogarithmic, as bounded by  (\ref{eq:hastings_koma_logr}).
\end{theor}
The proofs of these results are exceedingly technical and take up a few hundred pages. Very briefly, the sharpest results of~\cite{strictlylinear_KS,tran2021} implement a multiscale decomposition of the unitary, which relies on alternating two technical ingredients, one in space and one in time: spatially, one adds interaction with longer ranges via the interaction picture; temporally, one connects short-time Lieb-Robinson bounds into long-time Lieb-Robinson bounds. In contrast, the naive recursive Lieb-Robinson bounds in terms of summing over paths (Theorem~\ref{thm:sum_of_paths}), albeit simple, seem to lose some essential physics of the power-law interaction system.

In 1D, a much shorter alternative proof is possible~\cite{alpha_3_chenlucas} (and historically precedes the sharpest results in Theorem~\ref{thm:lrpowerlaw}); this requires regrouping the Hamiltonian by scales and expanding the exponential into a carefully chosen interaction picture (as a variant of the self-avoiding path construction in Theorem~\ref{thm:self-avoiding}). Unfortunately, this approach does not seem to generalize to higher dimensions.

\subsection{Frobenius light cone}

The Frobenius light cone in power-law interacting systems  \textit{qualitatively} differs from the Lieb-Robinson bounds~\cite{hierarachy}: see Figure~\ref{fig:powerlaw_summary}. This means that quantum dynamics in power-law interaction systems drastically depends on whether one considers the fine-tuned (worst-case) or random (average-case) states; there are even regimes where the distinction is exponentially large!

\begin{figure}
    \centering
    \includegraphics[width=0.8\textwidth]{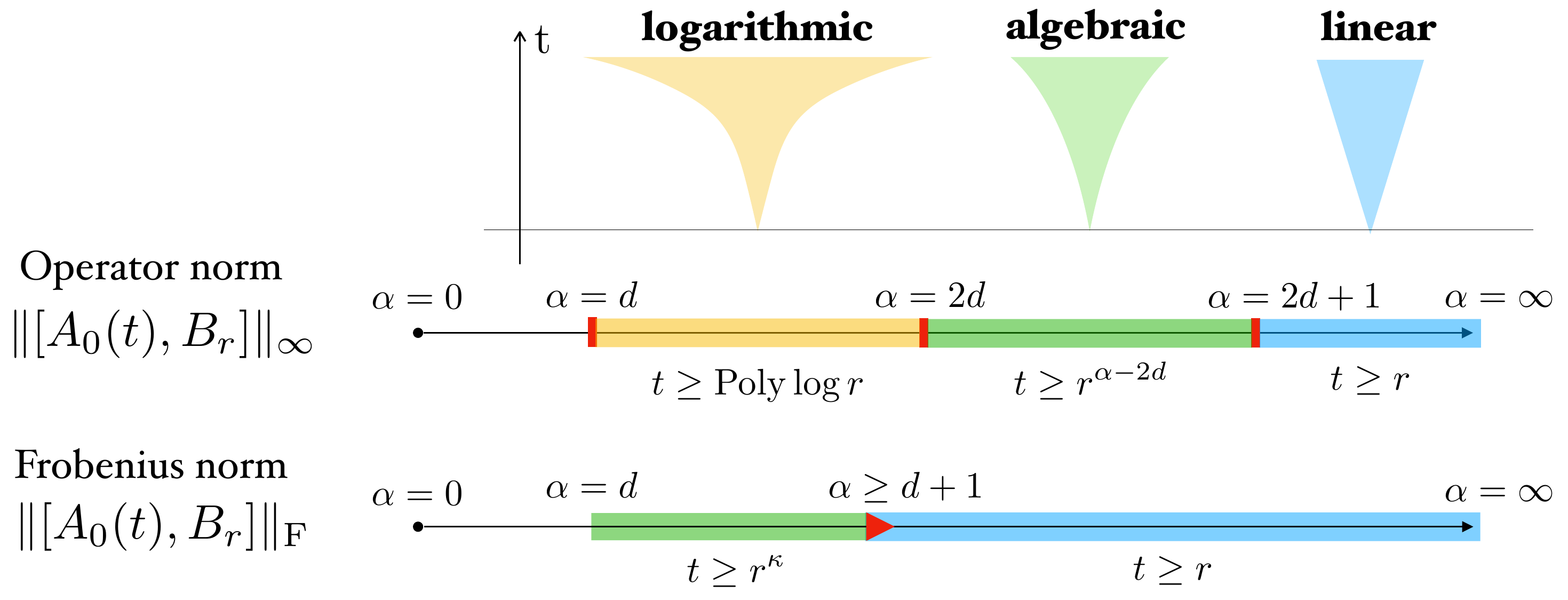}
    \caption{
    A summary of known speed limits in $d$-dimension lattices with power-law exponent $\alpha$. The Lieb-Robinson bounds in spectral norm have been completely classified in a series of works~\cite{fossfeig,tran2019, alpha_3_chenlucas,strictlylinear_KS} with saturating protocols~\cite{tran_optimalstatetransfer}. The Frobenius bounds have not been completely classified, but certainly, they exhibit qualitatively different physics (average case versus worst case)~\cite{hierarachy}. Scrambling in logarithmic time is impossible in the Frobenius norm at $\alpha >d$. So far, the algebraic-to-linear transition is known in 1d~\cite{Chen2021OptimalFL} at exactly $\alpha = 1+1$, but the higher dimensional case remains open with a conjectured value $\alpha = d+1$. 
    }
    \label{fig:powerlaw_summary}
\end{figure}

Right now, the complete shape of the Frobenius is only fully understood in $d=1$:
\begin{theor}[Frobenius bounds for 1d power-law interacting systems~{\cite{Chen2021OptimalFL}}]\label{thm:1d_powerlaw}
For $d=1$, $\alpha >1$, and any $\delta >0$, we have
\begin{equation}
    \lVert [A_0,B_r] \rVert_{\mathrm{F}} \ge \delta \quad \text{requires}\quad |t| \ge \delta^2 C \times \left\lbrace \begin{array}{ll} r/\ln r &\ \alpha>2  \\ r/\ln^2 r &\ \alpha =  2 \\ r^{\alpha-1} &\ 1<\alpha<2 \end{array} \right. \label{eq:main41}
\end{equation}
where the constant $0<C<\infty$ depends only on $\alpha$.
\end{theor}
Concretely, we see the transition to the algebraic light cone $(\alpha =2)$ is distinct from the Lieb-Robinson bounds $(\alpha= 3)$. In higher-dimension ($d\ge 2$), the critical value for this transition remains unknown, but is conjectured to be $d+1$.  However, we do at least know that the possibility for exponentially fast scrambling (a la Section \ref{sec:lyapunov}) is forbidden until $\alpha=d$:

\begin{theor}[Algebraic Frobenius bounds for power-law interacting systems~{\cite{Kuwahara_OTOC}}]\label{thm:higerd_powerlaw}
For any $\alpha >d$ and $t,r$ large with $t/r^\kappa$ fixed,
\begin{equation}
    \lVert [A_0,B_r] \rVert_{\mathrm{F}} \le C \left( \frac{t}{r^{\kappa}}\right)^{\beta}
\end{equation}
 where the constants $0<C <\infty,   0<\kappa <1, 0<\beta$ depend only on $\alpha, d$. 
\end{theor}

Another setting where the (conjectured, in $d>1$) Frobenius light cone is reached is when studying single-particle quantum walks (which also apply to free fermion models), where single-particle hopping rates are suppressed by power-law decay with distance:  \begin{prop}[Single particle quantum walk with power-law interactions \cite{hierarachy}] \label{prop:freepowerlaw}
Consider a single-particle walk on a $d$-dimensional lattice with power-law interactions, generated by \begin{equation}
    H(t) = \sum_{u,v\in\mathsf{V}} h_{uv}(t)|u\rangle\langle v|
\end{equation}
where the matrix $h_{uv}(t)$ is Hermitian and obeys (\ref{eq:hijpowerlaw}). Then for any $\epsilon>0$, so long as $\alpha>d$, there exist constants $0<v_{\mathrm{sp}},K<\infty$ such that \begin{equation}
    |\langle u|\mathrm{e}^{-\mathrm{i}Ht}|v\rangle| \le \frac{Kt}{(\mathsf{d}(u,v)-v_{\mathrm{sp}}t)^{\alpha-d-\epsilon}}
\end{equation}
at sufficiently short time.  For $\alpha<d+1$, we may take $v_{\mathrm{sp}}=0$.
\end{prop}

\subsection{Applications}
As we have emphasized in this review, Lieb-Robinson bounds constrain the minimal time required for various practical tasks.  An important question becomes whether these bounds can ever be saturated.  Remarkably, it turns out to be the case for many of the bounds in this section!  Hence we know that these bounds are optimal.  

In particular, it is possible \cite{tran_optimalstatetransfer} to use power-law interactions to prepare GHZ states (\ref{eq:GHZ}) at the maximal rate allowed by Theorem \ref{thm:lrpowerlaw}. The algorithm is an iterative process (\ref{fig:tran}), with the parameters of the algorithm (such as the number of clusters and number of iterations) chosen efficiently, so as to essentially saturate the limit allowed by Theorem \ref{thm:lrpowerlaw}.
\begin{theor}[Optimal state transfer protocol with power-law interactions~\cite{tran_optimalstatetransfer}]
The shapes of the light cone in Theorem~\ref{thm:lrpowerlaw} are achieved by the state transfer protocols~\cite{tran_optimalstatetransfer}.
\end{theor}

\begin{figure}
    \centering
    \includegraphics{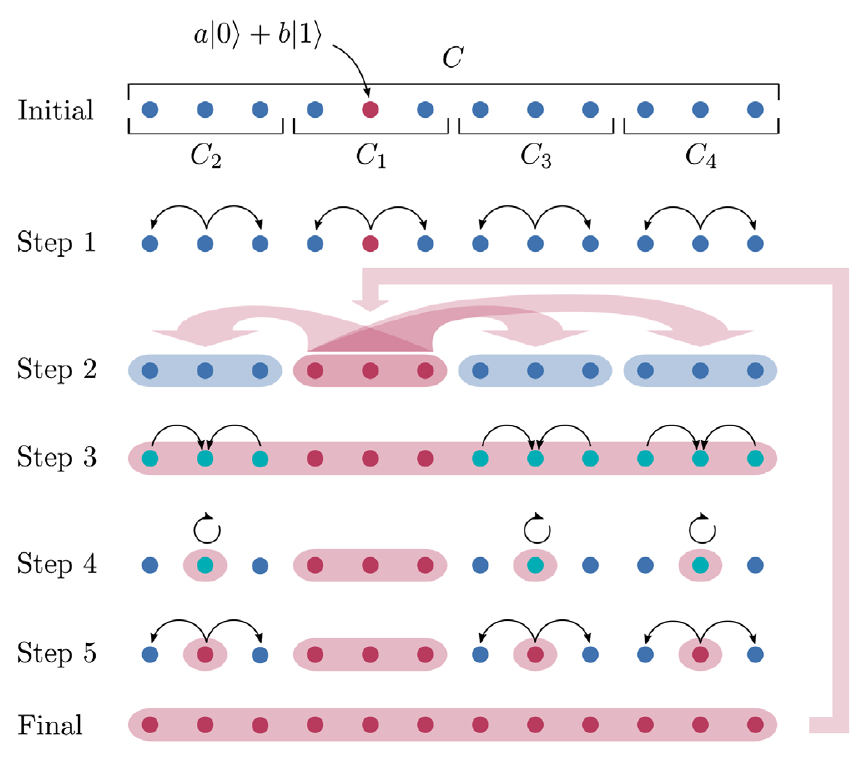}
    \caption{An algorithm to efficiently prepare GHZ states using power law interactions. The steps are: (\emph{1}) prepare GHZ states on clusters of O(1) sites using local interactions; (2) entangle the smaller GHZ clusters using long-range interactions; (3-5) in each cluster, undo and then re-build the entanglement in order to prepare a large GHZ state.  The algorithm then iterates to larger and larger scales.  The shaded regions depict entangled qubits at each step.  Figure taken from \cite{tran_optimalstatetransfer} with permission.}
    \label{fig:tran}
\end{figure}

\begin{exam}[Fast W-state preparation with power-law interactions \cite{hierarachy,yifan2021}]
The basic strategy is to use a free-fermion quantum walk, with all-to-all interactions of power-law strength between sets $B_n$ of increasing size $D_n \sim 2^n$ leading to a W-state expanding from one bubble to the next: see Figure \ref{fig:yifan}.  In the $n^{\mathrm{th}}$ time step, the Hamiltonian takes the form of \begin{equation}
    H = \sum_{u \in B_n,v\in B_{n+1}-B_n}\frac{1}{D_{n+1}^\alpha} \left(c^\dagger_u c_v + c_v^\dagger c_u \right)
\end{equation}  At each time step, due to constructive interference, the time to grow a W-state in $B_n$ into a W state in $B_{n+1}$ scales as \begin{equation}
    \tau_n \sim D_{n+1}^{\alpha-d},
\end{equation}
meaning that the total runtime of the algorithm scales as $R^{\min(1,\alpha-d)}$.  For $\alpha>d+1$, nearest neighbor hopping algorithms can be employed as an alternative.  Note that this quantum walk thus also asymptotically saturates Proposition \ref{prop:freepowerlaw}.  Two interesting features of this protocol are that it leads to efficient transmission of multiple qubits (with only $\sqrt{m}$ increased runtime needed to send $m$ qubits), and that it is robust to certain kinds of coherent errors in the protocol \cite{yifan2021}.
\end{exam}

\begin{figure}
    \centering
    \includegraphics[width=3.5in]{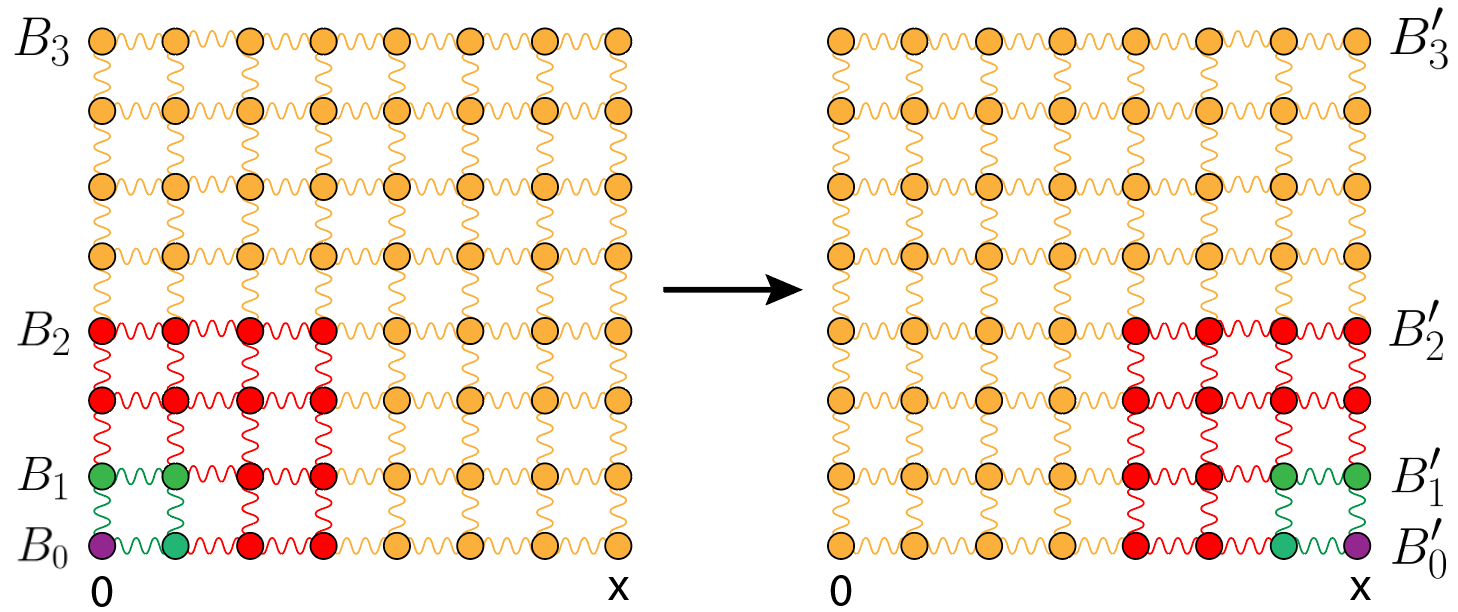}
    \caption{A W state generation algorithm with power-law interactions builds increasingly large W states on sets $B_n$ of side length $D_n\sim 2^n$ as illustrated.  In the figure, it is also shown that the W state can be collapsed in such a way as to perform single qubit state transfer. Figure taken from \cite{yifan2021} with permission.}
    \label{fig:yifan}
\end{figure}

Both the GHZ and W state preparation algorithms can be used to perform state transfer, by simply reversing the unitary protocol that generates it, but ``ending" on a site on the other side of the domain.  An earlier state transfer protocol based on power-law interactions is given in \cite{eldredge}.

In addition to preparing a particular entangled state, one may seek to generate as much bipartite entanglement as possible between two regions exploiting power-law interaction. Naturally, one strategy is to repeat single-qubit state transfer~\cite{tran_optimalstatetransfer}, which turns out to be optimal using resource theory arguments (Section~\ref{sec:entanglement_resource}).
\begin{exam}[Optimal generation of bipartite entanglement with power-law interactions]
In a $d$-dimensional model with power-law interactions of exponent $\alpha$, generating the maximally entangled state between two adjacent cubes with side length $r$ takes time at least \cite{gong2017}\begin{align}
t = \Omega(r^{\min(1,\alpha-d)}),
\end{align}
starting from an arbitrary initial state without bipartite entanglement. This is asymptotically saturated (up to sub-algebraic corrections) by sequentially sending 1-qubit at a time using the protocol of~\cite{tran_optimalstatetransfer}. 
\end{exam}
The above is a direct application of bipartite entanglement generating bounds (Proposition~\ref{prop:entanglement_rate_general}) by summing over all interactions.

Recent work \cite{Kuwahara1d,Wang:2022apr} has also discussed bounds on correlations and entanglement in gapped ground states of systems with power-law interactions.  Some results on prethermal behavior in power-law interacting systems can be found in \cite{yaoheating,tranheating}. Metrology is discussed in \cite{Chu2023StrongQM}.

\section{Bounds at finite energy and finite density}
\label{sec:density}
The Frobenius light cone perspective is closely related to the study of dynamics at finite charge density or finite energy density.  Indeed, while any formal Lieb-Robinson bound would hold for any state at low energy (e.g.), we usually anticipate a much stronger bound holding in practice.  In this section, we will describe a variety of methods that are used to try and provide stronger bounds on dynamics at low temperature or low density.
\subsection{Dynamics at low charge density}
We begin with a discussion of dynamics in a system of  $N$ qubits governed by time evolution that commutes with the total $Z$-spin (or magnetization): \begin{equation}
    Q = \sum_{i=1}^N \frac{1+Z_i}{2} = \sum_{i=1}^N n_i.
\end{equation}  If time evolution is generated by a continuous Hamiltonian, \begin{equation}
    [H(t),Q] = 0. \label{eq:HcommuteQ}
\end{equation}
Note that in this section it will be more convenient for us to define $n_i=0,1$ rather than $Z_i = -1,1$.  

\subsubsection{Formalism for low-density quantum walk bounds}

We are interested in ``low density" states, or those where $\langle Q\rangle  \ll N$.  While it follows from (\ref{eq:HcommuteQ}) that $Q$ is exactly conserved if we restrict the Hilbert space to states with $Q$ ``particles", the Hilbert space does not have a product space structure.  Notions like operator size become challenging to define, and more importantly, it is far from obvious how to exploit the low-density nature of the state to derive an improvement on a Lieb-Robinson-like light cone, restricted to such a low-density state.

For low-density states, the situation seems relatively mild, due to the following proposition: 
\begin{prop}[Grand canonical density matrix]
Consider the ``grand canonical" density matrix \begin{equation}
    \rho = \mathrm{e}^{-\mu Q} \times \text{normalization}. \label{eq:rho101}
\end{equation}
In the $N\rightarrow \infty$ limit, correlation functions $\mathrm{tr}(\rho A)$ of operators $A$ that can only create or destroy O(1) particles are determined by states at density \begin{equation}
    \bar n = \frac{Q}{N} = \frac{1}{1+\mathrm{e}^\mu}.
\end{equation}
The allowed family of $A$ includes all local $k$-point functions for fixed $k$.
\end{prop}
\begin{proof}
Since $\rho$ is diagonal in the $n_i$ basis, this proposition is a statement about the probability that a state is drawn from $\rho$ has a given density of particles $n = Q/N$?  From the form of $\rho$, the probability that we would draw $n_i=1$ is independent and identically distributed with respect to any $n_j$ ($j\ne i$).  Using a modification of (\ref{eq:hoeff}), we find:
\begin{equation}
    \mathbb{P}\left[ \left| \frac{Q}{N}-\bar n \right| > \frac{c}{\sqrt{N}}\right] < \exp\left[ - \frac{c^2}{2\bar n(1-\bar n)}\right].
\end{equation}
We say that $\rho$ concentrates on states of density $\bar n$.
\end{proof}

This proposition is very helpful because it turns out to be quite natural to calculate expectation values with $\rho$ using a mild modification of the formalism of Section \ref{sec:frobenius}.  In particular, using the density matrix (\ref{eq:rho101}), consider redefining the operator inner product (\ref{eq:innerproduct}) to \cite{Chen:2020bmq} \begin{equation}
    (A|B) = \mathrm{tr}\left(\sqrt{\rho} A^\dagger \sqrt{\rho} B\right). \label{eq:innerproductQ}
\end{equation}
This inner product concentrates in the thermodynamic limit onto states at density $\bar n$.  

Ref.~\cite{Chen:2020bmq} proposed a natural generalization of operator size to a (low-density) system.  Defining the operators \begin{subequations}\label{eq:ccdagger}\begin{align}
    c_i = | 0\rangle \langle 1|_i \quad \text{and}\quad  c_i^\dagger = |1\rangle \langle 0|_i,
\end{align}
\end{subequations}
we notice that the space of operators is \begin{equation}
    \mathrm{End}(\mathcal{B}) = \mathrm{span}\left( \lbrace I, c_i, c_i^\dagger, n_i \rbrace^{\otimes N} \right).
\end{equation}
There is then a natural definition of operator size: \begin{equation}
    \mathcal{S} = \sum_{j=1}^N \mathcal{S}_j,
\end{equation}
where on a single site \begin{align}
    \mathcal{S}|I) = 0|I), \quad 
    \mathcal{S}|c) = |c), \quad
    \mathcal{S}|c^\dagger) = |c^\dagger), \quad
    \mathcal{S}|\tilde n) = 2|\tilde n),
\end{align}
where \begin{equation}
    |\tilde n) := |n) - \bar n |I).
\end{equation}
Since $n=c^\dagger c$, one can think of this $\mathcal{S}$ as counting the number of $c$ and $c^\dagger$ in an operator string.  $|\tilde n)$ is defined to be orthogonal to the identity, which is why we have subtracted off its average.
Note that this definition does not reproduce, when $\mu=0$, the standard definition of operator size: the Pauli $Z$ has ``twice" the length of $X$ or $Y$, at $\mu=0$.  After a short calculation, one finds that \begin{subequations}\label{eq:sec11sizelengths}\begin{align}
    (I|I) &= 1, \\
    (c_i|c_i) &= (c_i^\dagger | c_i^\dagger) = \mathrm{e}^{\mu/2}\bar n = \sqrt{\bar n (1-\bar n)}, \\
    (\tilde n_i |\tilde n_i) &= \bar n (1-\bar n).
\end{align}\end{subequations}
When $\bar n\ll 1$, one can \emph{heuristically} think that the inner product of an operator $A$ of size $s$ has an inner product $(A|A) \sim \bar n^{s/2}$. 

\subsubsection{Operator growth and decay}

With these definitions, we are now ready to discuss two useful results that allow us to generalize the many-body quantum walk bounds to low-density systems. 
\begin{prop}[Density-dependent inner product is norm preserving]
Under the inner product (\ref{eq:innerproductQ}) and charge-conserving dynamics, $(A(t)|A(t)) = (A|A)$.
\end{prop}
The proof of this follows immediately from (\ref{eq:HcommuteQ}) and Proposition \ref{prop:Aconslength}.    The Lyapunov exponent $\lambda_{\mathrm{L}}$ can be defined for the $\mu$-dependent notion of size, analogously to Section \ref{sec:lyapunov}.  We then find:
\begin{theor}[Bounds on Lyapunov exponent at finite density]
For some $0<C<\infty$, in a model subject to the assumptions of Theorem \ref{theor:fastscrambling}, together with (\ref{eq:HcommuteQ}),  \begin{equation}
    \lambda_{\mathrm{L}} \le C \sqrt{\bar n}. \label{eq:lambdafinitedensity}
\end{equation}
\end{theor}
\begin{proof}[Proof sketch] 
The idea is rather similar to Theorem \ref{theor:fastscrambling}: we expand \begin{equation}
    |A) = \sum_{s} |A_s)
\end{equation}
into operators of fixed size: $\mathcal{S}|A_s)=s|A_s)$, such that \begin{equation}
    \frac{\mathrm{d}}{\mathrm{d}t}(A(t)|\mathcal{S}|A(t)) = (A(t)|[\mathcal{S},\mathcal{L}]|A(t)) = \sum_{s_1<s_2}2(s_2-s_1)\times (A_{s_2}(t)|\mathrm{i}[H,A_{s_1}(t)]). \label{eq:thm113main0}
\end{equation}
The key observations are that: (\emph{1}) using antisymmetry of $\mathcal{L}$ we can always do this, and (\emph{2}), as this inner product vanishes unless $[H,A_{s_1}]$ has size $s_2>s_1$, we can bound (under similar $k$-locality assumptions to Theorem \ref{theor:fastscrambling}): \begin{equation}
    (A_{s_2}(t)|\mathrm{i}[H,A_{s_1}(t)]) \le Cs_1 \times  \left[\bar n (1-\bar n)\right]^{(s_2-s_1)/4} \times \sqrt{(A_{s_2}|A_{s_2})(A_{s_1}|A_{s_1})}. \label{eq:thm113main}
\end{equation}
The constant $C$ will be related by a constant prefactor to the value at $\bar n=0$, but due to the fact that the analog of Pauli $Z$ has both size-0 and size-2 components under our current counting, it may have quantitative (though not qualitative) $\bar n$ dependence.  The new factor of interest in (\ref{eq:thm113main}) is the $\bar n^{(s_2-s_1)/4}$ scaling.   The reason we have this prefactor is that the inner product $(A_{s_2}(t)|\mathrm{i}[H,A_{s_1}(t)])$ is between two operators of size $s_2$, whereas $(A_{s_1}|A_{s_1})$ is normalized to a size $s_1$ operator, so we need additional factors from (\ref{eq:sec11sizelengths}).

Lastly, due to $[H,Q]=0$, if $[A,Q] = kA$, then $[A(t),Q]=kA(t)$.   In particular, this means that under time evolution, we can only modify $A$ by operators that do not contribute an odd number of ``$c$ or $c^\dagger$", as defined in (\ref{eq:ccdagger}). Therefore, in (\ref{eq:thm113main0}), we must have $s_2-s_1 \in 2\mathbb{Z}^+$.  The smallest value that can contribute to the Lyapunov exponent is $s_2-s_1=2$.  Hence we obtain (\ref{eq:lambdafinitedensity}).
\end{proof}

The fact that $\lambda_{\mathrm{L}}$ must vanish at low density is qualitatively similar to the proposal of \cite{Maldacena:2015waa} that dynamics must slow down at low temperature, although it is far from clear how to generalize this argument to finite temperature states.  Using the same technique, one can prove the following corollary, which states that in typical states, the time it takes for a single-particle operator to decay into a many-particle operator is slow.  This is reminiscent of a conjecture of \cite{Lucas:2018wsc} about the slow decay of ``small operators" in finite temperature dynamics. 
\begin{corol}[Decay time bound at finite density]\label{cor:decaydensity}
For times $t < \tau = C \bar n^{-1/2}$, the time evolution of a size-1 operator (``creation/annihilation operator"), is well approximated by single-particle dynamics.  Namely, if $\mathcal{S}|A)=|A)$, then there exists $|B(t)) = \mathcal{S}|B(t))$ such that  for $t<\tau$, $(A(t)|B(t)) \ge \frac{1}{2}\sqrt{(A|A)(B|B)}$.
\end{corol}
Using Fermi's Golden Rule, we might expect to have the stronger scaling $\tau \sim \bar n^{-1}$, but such a result has not yet been proved in generality (as far as we know).  Intuitively, this scaling would then be optimal, since a particle traveling at constant velocity (see the next subsetion) will encounter another particle in a time $\sim \bar n^{-1}$, at which point we expect single-particle correlation functions to decay.

\subsubsection{Butterfly velocity}
One might expect that, similarly to how the Lyapunov exponent decreases with increasing $\bar n$, the butterfly velocity also decreases as $v_{\mathrm{B}}\sim \bar n$.  Indeed this is expected on heuristic grounds \cite{xiaorahul}.  However, we emphasize that this is quite non-trivial to prove (in fact an outstanding question). The challenge can be seen by considering just \emph{one} particle (the $Q=1$ subspace). In this setting, the quantum dynamics reduces to a single particle quantum walk, which we have seen has a \emph{finite velocity} in Section \ref{sec:babyQW}.  So there cannot, in a strict sense, be any bound on signal propagation at low density where the velocity vanishes with $\bar n$.

What was conjectured in \cite{xiaorahul} is that for short times $\tau \lesssim \bar n^{-1}$, an operator indeed behaves as if it propagates in this $Q=1$ Hilbert space, with a finite velocity.  After $t\sim\tau$, however, the operator grows more complex, and \emph{destructive interference} sets in between operators of different sizes, in such a way that the effective operator growth front propagates with a velocity $v_{\mathrm{B}} \sim \bar n$.  It would be interesting if this can be shown rigorously in any context.  The fact that there is a fast velocity for operator growth at short times, and a slow velocity at late times, poses a challenge to attempts \cite{Blake:2016wvh,Hartman:2017hhp} to universally bound diffusion constants in terms of $v_{\mathrm{B}}$ and $\lambda_{\mathrm{L}}$.

\subsection{Models of interacting bosons}
\label{sec:bosons}
We now turn to a very similar story, where we study the dynamics of interacting bosons.  Starting for the moment with a single boson, the Hilbert space is that of the quantum harmonic oscillator, labeled by an infinite tower of states $|n\rangle$, for $n\in\mathbb{Z}^+$.  We define the operators \begin{subequations}\begin{align}
    b|n\rangle &:= \sqrt{n}|n-1\rangle, \\
    b^\dagger |n\rangle &:= \sqrt{n+1}|n+1\rangle.
\end{align}\end{subequations}

Next, we consider an interacting boson model on a graph $\mathsf{G}=(\mathsf{V},\mathsf{E})$.   On every vertex of the graph, we consider a bosonic Hilbert space as above.  The creation and annihilation operators obey \begin{equation}
    [b_u,b^\dagger_v] = \mathbb{I}(u=v).
\end{equation}
For convenience, we will restrict our study to (time-dependent) Hamiltonians of the form \begin{equation}
    H(t) = \sum_{\lbrace u,v\rbrace \in \mathsf{E}} J_{uv}(t) \left(b^\dagger_u b_v + b^\dagger_v b_u\right) + \sum_{v\in\mathsf{V}}f\left(b^\dagger_v b_v\right), \label{eq:bosonH}
\end{equation}
where $f$ is an arbitrary function.  The results we describe below can be generalized to slightly more complex $H$; for pedagogy, we stick to the above.   However, the essential features of (\ref{eq:bosonH}) are that: (\emph{1}) only one boson can move at a time, and (\emph{2}) the energy in $f$ will in general be unbounded; however it only depends on the boson number operator $b^\dagger_v b_v$.   For $u\ne v$, we have $[b^\dagger_u b_u, b^\dagger_v b_v] = 0$, so the interaction terms in the Hamiltonian commute.  (\ref{eq:bosonH}) generalizes the classic Bose-Hubbard model \cite{gersch}, in which $J_{uv}=J$ does not depend on time, while for some constant $U$, \begin{equation}
    f(n) = Un(n-1).
\end{equation}

Even with these restrictions, however, it is non-trivial to imagine a Lieb-Robinson bound for such a problem.  Every single operator in (\ref{eq:bosonH}) is \emph{unbounded} in operator norm, so the conventional proof of a Lieb-Robinson bound as in Section \ref{sec:LRbounds} will not work.  Moreover, there is a physical construction \cite{bhspeed1,bhspeed2,bhspeed3} that leads to propagating velocity $v\sim \bar n$, if $\bar n$ is the mean number of bosons per site in the system.  Therefore, for some time, Lieb-Robinson bounds with bosons were only found in quite restrictive settings, such as classical models \cite{Raz_2009}, harmonic (non-interacting) models or models with bounded interaction strength \cite{Nachtergaele_2008}, or models with a bounded boson number \cite{schuch}.

Recently, the nature of Lieb-Robinson light cones in interacting boson models has been rigorously established.  A nearly linear light cone was presented in \cite{kuwahara2021liebrobinson}, while the exact linear light cone (in a sense described below) was proved in \cite{Yin:2021uio}.  The proofs of all results are rather involved so we leave them to the literature, and just highlight the conclusions here with some prototypical examples.

\begin{theor}[Linear light cone for interacting bosons in typical states]
Let $A_u$ and $B_v$ denote local operators consisting of a finite number of raising and lowering operators, acting on single sites $u,v\in\mathsf{V}$.  Then for some constants $0<C,v,a<\infty$, \cite{Yin:2021uio} \begin{equation}
    \tr\left(\sqrt{\rho}[A_u(t),B_v]^\dagger\sqrt{\rho}[A_u(t),B_v]\right) \le C \left(\frac{vt}{r}\right)^{ar} 
\end{equation}
When $A_u$ and $B_v$ are single boson creation or annihilation operators, there exist $0<c_{1,2}<\infty$ such that \begin{equation}
    v \le c_1 + c_2 \bar n. \label{eq:optspeedboson}
\end{equation}
This scaling is known to be asymptotically optimal \cite{bhspeed1}.
\end{theor}

\begin{theor}[Lieb-Robinson light cone in finite density states in one dimension]
Consider two arbitrary states $|\psi_{1,2}\rangle$ of bosons on a one-dimensional lattice, restricted to the form \begin{equation}
    |\psi_{1,2}\rangle = \sum_{n_i = 0}^M c_{1,2}(\mathbf{n}) |\cdots n_1n_2\cdots\rangle \label{eq:thm116psi}
\end{equation}
Namely, there are in the initial state, no more than $M$ bosons on any given site.   Then for any $M$ and some $0<C,v,a<\infty$ which can depend on $M$, \cite{Yin:2021uio} \begin{equation}
    |\langle \psi_1|[A_u(t),B_v]|\psi_2\rangle | \le C \left(\frac{vt}{r}\right)^{ar} .
\end{equation}
All matrix elements of commutators are small in one dimension.
\end{theor}

One can prove, as a consequence of this theorem, that gapped ground states have exponential correlations \cite{Yin:2021uio}, and that bosons are not too much more difficult to simulate than spins, at low density.  However, the analogy between the locality in the Bose-Hubbard model and in interacting spin models does not persist in higher dimensional settings.

\begin{theor}[No Lieb-Robinson light cone beyond one dimension \cite{Kuwahara:2022hlg}] \label{thm117}
There exist states $|\psi_{1,2}\rangle$ of the form (\ref{eq:thm116psi}), and time-dependent $H$ of the form (\ref{eq:bosonH}), on a $d$-dimensional lattice ($d>1$), for which $|\langle \psi_1|[A_u(t),B_v]|\psi_2\rangle | \ge 1$ for $r=bt^d$, for O(1) constant $b$.  Hence, there is an effective velocity $v_{\mathrm{eff}}\sim t^{d-1}$.  However, up to logarithmic corrections, one cannot find any protocol for which $v_{\mathrm{eff}}$ is faster. 
\end{theor}

The basic idea behind this theorem is to imagine propagating information in the $x_1$-direction, and starting with a state with one boson on every site at time $t=0$.  For time $0<t<T/2$, we pile up bosons in $x_1=\text{const.}$ planes, as much as possible, along the axis $x_2=\cdots = x_d=0$.  At time $T/2$, we have $\sim T^{d-1}$ bosons on every site.  Then we use the protocol of \cite{bhspeed1} (see (\ref{eq:optspeedboson})) to send signals at a speed $v\sim T^{d-1}$ along the $x_1$-axis.  However, it is also known that while ``signals" can be sent very quickly in higher dimensions, a large number of bosons cannot be: \begin{theor}
[No fast particle transport (informal) \cite{Faupin:2021nhk,Faupin:2021qpp}]
Let $\mathsf{R}\subset \mathsf{V}$ have diameter $r$.  Then in any state with $N$ bosons in region $\mathsf{R}$ at $t=0$, the time $\tau$ needed to move a finite fraction (e.g. $N/2$ bosons) out of $\mathsf{R}$ obeys $\tau \ge Cr$ for some $0<C<\infty$. 
\end{theor}

Even in one dimension, it is known that the finite speed limit at low density arises only due to the particular nature of the Hamiltonian (\ref{eq:bosonH}).  One way to violate any Lieb-Robinson-like bound is to include number-conservation-violating terms such as $b_i^\dagger b_{i+1}^\dagger + \mathrm{H.c.}$ \cite{eisert2008}, which can cause the effective velocity to grow exponentially with time: $v\sim \exp(t)$.   But even if we conserve number, a density-dependent hopping term $f_i(n_i+n_{i+1}) b_i^\dagger b_{i+1}+\mathrm{H.c.})$ will destroy many notions of locality.  Recall for example that with $f=1$, the time it would take to get from a state where all bosons are piled up on the left end of a 1d chain ($|n0\cdots 0\rangle$), to a state where they are all piled up on the right ($|0\cdots 0n\rangle$) would scale with the length of the chain, and independently of $n$.  However, by simply choosing $f_i=n_i+n_{i+1}$, we can make the time arbitrarily small by decreasing $n$.  We expect that using similar strategies to Theorem \ref{thm117}, especially in higher dimensions, one will find little remains of a notion of locality once the interacting terms incorporate boson hopping.

Lastly, we remark that special Lieb-Robinson bounds have been derived \cite{LRion} for Hamiltonians of the schematic form \begin{equation}
    H = b^\dagger b + b A + b^\dagger A^\dagger
\end{equation}
where $A$ is a $k$-local many-body operator, and $b$ is the boson annihilation operator.  One can generalize to multipole boson modes as well. Such Hamiltonians  arise when modeling trapped ion crystals \cite{Britton2012}.  

\subsection{Lieb-Robinson bounds in continuous space at finite energy}\label{sec:continuous}
Now, we turn to bounds that only hold at finite energy (density).   This is an extremely hard problem, albeit one of significant physical importance.   We begin by \emph{briefly} reviewing the efforts to extend the Lieb-Robinson Theorem into the continuum.   This literature has mainly remained restricted to the mathematics community \cite{jensen1985,skibsted,hunziker,hunziker2,Arbunich_2021}; however, we will give an intuitive sketch of the ideas in the context of the one-dimensional Schr\"odinger equation (which seems to us to capture the essence of many of the ideas developed there).  The key insight is, as we explain shortly, such bounds are only plausible at finite energy.  Hence, this provides a nice gateway into the study of bounds on finite temperature dynamics more generally.  We will avoid technical discussions of functional spaces and smoothness, which are of course discussed at length by mathematicians, and focus on aspects of the literature that seem relevant for finite-dimensional quantum systems as well.   

We now prove the following theorem, which represents the minimal model of a continuum Lieb-Robinson-style bound, applicable to the motion of a single particle.
\begin{theor}[Continuous space Lieb-Robinson bound at low energy]
\label{theor:schrodinger}
Consider the time-independent Hamiltonian operator \begin{equation}
    H = -\frac{1}{2}\partial_x^2 + V(x)\quad \text{such that} \quad V(x) \ge 0   \label{eq:Vbound}.
\end{equation}
Define a projector onto low-energy states \begin{equation}
    \mathbb{P}_E := \mathbb{I} \left(H \le E \right).
\end{equation} Then, \begin{equation}
    \lVert \mathbb{P}_E [H,x] \mathbb{P}_E \rVert \le \sqrt{2E}, \label{eq:PEbound}
\end{equation}
and for any state obeying $|\psi\rangle = \mathbb{P}_E|\psi\rangle$, \begin{equation}
    \langle \psi| \mathrm{e}^{\mathrm{i}Ht} x \mathrm{e}^{-\mathrm{i}Ht}|\psi\rangle \le \langle \psi | x | \psi\rangle + \sqrt{2E}t. \label{eq:velconres}
\end{equation}
\end{theor}
\begin{proof}
Since $H$ does not depend on time, $[\mathrm{e}^{\mathrm{i}Ht},\mathbb{P}_E]=0$, so (\ref{eq:velconres}) follows immediately from (\ref{eq:PEbound}).   To prove (\ref{eq:PEbound}), observe that \begin{equation}
    [H,x] = -\frac{1}{2}\left[\partial_x^2, x\right] + \left[V(x),x\right] = -\partial_x.
\end{equation}
Now using (\ref{eq:Vbound}), we see that (for any normalized states $\langle \varphi|\varphi\rangle = \langle \varphi^\prime | \varphi^\prime\rangle = 1$): \begin{align}
   \left| \langle \varphi | \mathbb{P}_E \partial_x \mathbb{P}_E|\varphi^\prime\rangle \right|^2 \le  \langle \varphi | \mathbb{P}_E \partial_x \mathbb{P}_E^2 \partial_x^\dagger \mathbb{P}_E|\varphi\rangle \le \langle \varphi | \mathbb{P}_E \left(- \partial_x^2\right) \mathbb{P}_E|\varphi\rangle &= \langle \varphi| \mathbb{P}_E \left(2(H-V)\right)\mathbb{P}_E|\varphi\rangle \notag \\
   &\le 2\langle \varphi|\mathbb{P}_EH\mathbb{P}_E|\varphi\rangle \le 2E.
\end{align}
Since we have bounded the largest possible matrix element of $\mathbb{P}_E[H,x]\mathbb{P}_E$, we obtain (\ref{eq:PEbound}).  
\end{proof}

What this theorem implies is that, if we can restrict to low energy states, there is a Lieb-Robinson-like bound that depends on energy $E$.  Note that in the present proof, we have assumed $V$ does not depend on time $t$.  This assumption was important to ensure that the potential does not inject energy into the system.\footnote{It is likely that certain time-dependent potentials are mild, however: see e.g. \cite{Arbunich:2023aqm} for a nonlinear Schr\"odinger equation where an effective time-dependent potential does not qualitatively change the nature of (\ref{eq:PEbound}).}  

This bound has been generalized to higher dimensions and more recently to a weakly nonlinear Schr\"odinger equation \cite{Arbunich:2023aqm}.  This technique can also be straightforwardly extended to certain many-body fermion models \cite{Gebert:2019cou,osbornefuture}, similar to the Bose-Hubbard model.  It has been extended to open systems in \cite{Breteaux:2022qdp}.

\subsection{Butterfly velocity and finite temperature bounds}\label{sec:finiteT}
An idealistic goal would be to generalize Theorem \ref{theor:schrodinger} to many-body systems.  Of course in principle this is possible formally: one could study for example a first-quantized model of $N$ particles interacting on a lattice, or in the continuum, and apply this bound.  However, there is a crucial problem.   If we have $N$ particles interacting on the line, then $E\sim N$ in a physically reasonable state (e.g. one drawn from the finite temperature ensemble, with high probability).   Our energy-dependent velocity bound would then diverge in the thermodynamic limit.  This problem arises because we cannot rule out the possibility that all of the energy in the system is dumped into a single particle's kinetic energy, which then can propagate very fast.  

While this scenario is almost certainly unphysical, ruling it out has been a mathematical challenge.  It is noted in \cite{osbornefuture}, following earlier work \cite{Faupin:2021nhk,Faupin:2021qpp} on particle transport in the Bose-Hubbard model, that if one asks about \emph{average} particle velocity, this effect disappears and one can show that the typical particle cannot move faster than allowed by the average energy per particle.

Conjecturing, however, that one can achieve this goal of finding an energy-dependent velocity bound in many-body systems, what should we expect to find?  We conjecture that the correct bound would be of the form \begin{equation}
    \lVert \mathbb{P}_E [H,x_i]\mathbb{P}_E \rVert \lesssim v_{\mathrm{B}} \approx T\xi, \label{eq:Txi}
\end{equation}
where $T$ is the temperature and $\xi$ is the thermal correlation length.    For single-particle motion, we can ``derive" (\ref{eq:Txi}) from Theorem \ref{theor:schrodinger} as follows: the energy $E$ of the particle is roughly the temperature $T$, while the thermal correlation length $\xi \sim 1/\langle |\partial_x|\rangle \sim 1/\sqrt{T}$.  In lattice models, this scaling of $\xi$ can be derived rigorously using a Chebyshev expansion \cite{sachdeva} of the thermal density matrix $\mathrm{e}^{-\beta H}$ \cite{kuwaharathermal}.   

For single particle motion, (\ref{eq:Txi}) was proved in \cite{osbornefuture} subject to some rather special constraints on $H$.  Roughly speaking, suppose one has $H \sim \partial_x^z$, where $z=2,4,6,\ldots$ is an even integer.  Then one expects \begin{equation}
    \xi \sim T^{-1/z}. \label{eq:xiTz}
\end{equation}
Using methods analogous to the proof of Theorem \ref{theor:schrodinger}, \cite{osbornefuture} has demonstrated a family of spatially inhomogeneous lattice and continuum models where (\ref{eq:Txi}) can be derived with (\ref{eq:xiTz}).   Using non-rigorous methods, this scaling has been demonstrated for the dynamics of free fermions in \cite{Roberts:2016wdl}, and in holographic models (see Section \ref{sec:qg}) of quantum matter \cite{Blake:2016wvh}.

We believe that (\ref{eq:Txi}) is quite intuitive: at low temperatures, one might expect an \emph{effective description} of the dynamics with a Hamiltonian where local terms have energy of at most $T$: the price we must pay is that this effective Hamiltonian will be non-local on a scale set by the thermal correlation length $\xi$.

Note that the thermal correlation length $\xi$ has been bounded to always be less than $T^{-1}$ in local lattice models with a finite-dimensional local Hilbert space dimension \cite{kuwaharathermal3}; unfortunately, with this scaling, (\ref{eq:Txi}) simply reproduces the Lieb-Robinson bound.  Indeed, the proof that $\xi \lesssim T^{-1}$ uses Lieb-Robinson bounds! Other recent work \cite{kuwaharathermal3, kuwaharathermal,kuwaharathermal2} has studied entanglement in thermal states.  It is possible that these methods could be helpful in proving rigorous bounds on the butterfly velocity.

The authors of \cite{Han:2018bqy} have taken a somewhat different attempt to bound the butterfly velocity; however, their approach is not known to reproduce the desired scaling of (\ref{eq:Txi}).

\section{Open problems}
We have given a fairly detailed overview of Lieb-Robinson bounds, along with their known extensions and applications.  Now, we conclude  with a brief discussion and recap of interesting open problems and directions.

(\emph{1}) As discussed in Section \ref{sec:thermalization}, we do not have strong bounds the decay times of simple correlation functions, especially at low temperature.  Solving this problem may lead to qualitatively new methods for locality bounds in quantum dynamics, and seems like an important issue to address in the near future.  There are many conjectured bounds on transport coefficients \cite{Hartman:2017hhp,Kovtun:2004de,Hartnoll:2014lpa,Hartnoll:2021ydi} that would be interesting to revisit after rigorous results on decay times are better understood: see also \cite{Nussinov:2021fgc}.

(\emph{2}) As we just discussed, bounds on finite temperature dynamics are both notoriously difficult (as are statements about even equilibrium properties such as entanglement!), yet they have been the inspiration for many of the important developments in this field in recent years.  We hope that it is possible to generalize the Frobenius light cones described in Section \ref{sec:density} to generic finite temperature Hamiltonian systems, but this is likely to be an extraordinary challenge.  A particularly motivating problem would be a proof of the chaos bound of \cite{Maldacena:2015waa} in all-to-all interacting models, without ergodicity assumptions (see also \cite{Lucas:2019aif}).  

(\emph{3}) It is not known how to prepare the W state (\ref{eq:W}) as quickly as the GHZ state (\ref{eq:GHZ}), using both unitary dynamics and local measurement.  We expect that it is provably harder to prepare W states.  Whether or not a clever application of Lieb-Robinson bounds can prove this result or one needs a more fine-grained bound on multipartite entanglement, is an important open problem.  More generally, we expect that deeper incorporation of Lieb-Robinson bounds into error-corrected quantum dynamics with measurement can be a fruitful endeavor.

(\emph{4}) It would be interesting if the techniques of Section \ref{sec:preth} can rigorously address actively debated questions about the existence and robustness of many-body localization \cite{Basko_2006,vadim2007,Nandkishore:2014kca,Imbrie_2016,abaninreview,Suntajs:2019lmb}.

(\emph{5}) There is rather extensive literature on quantum walks \cite{Venegas-Andraca:2012zkr}.  Sometimes, this literature makes deep connections to the classical theory of random walks and Markov chains \cite{levinbook}, although certain notions such as hitting time have been delicate to extend to the quantum setting.  It might be fruitful to revisit some of these questions with the recent developments of Section \ref{sec:frobenius} in mind.

(\emph{6}) Recent work \cite{Wild:2022owx} uses cluster decomposition techniques (more well-known in equilibrium statistical mechanics) to study quantum dynamics.  Combining these techniques with Lieb-Robinson bounds may be a fruitful direction for future research.

(\emph{7}) Preparing squeezed states \cite{Ma_2011,sensing_rmp} is notoriously difficult in experiments, and we expect that Lieb-Robinson bounds, and the corresponding notions of locality, could help find efficient ways to generate highly squeezed states, likely using error correction as part of the protocol \cite{Zhou_LOCC20,Zhou:2023fxe}.

(\emph{8}) The linear-to-algebraic transition in the Frobenius light cone for power-law interactions (Section~\ref{sec:power-law}) remains open for $d\ge 2$ dimensions.

(\emph{9}) Understanding whether there are any more general classes of interacting boson models where sharp Lieb-Robinson bounds can be derived in thermodynamically reasonable states is an important problem, especially since many future quantum computers (whose operation speed might be constrained by a Lieb-Robinson light cone) will contain bosonic degrees of freedom.

\addcontentsline{toc}{section}{Acknowledgments}
\section*{Acknowledgments}
We thank Oliver Janssen for pointing out an error in the draft of this review.

This work was supported in part by a Research Fellowship from the Alfred P. Sloan Foundation under Grant FG-2020-13795 (AL), and by the U.S. Air Force Office of Scientific Research under Grant FA9550-21-1-0195 (AL, CY).

\begin{appendix}

\end{appendix}
 
\bibliographystyle{JHEP}
\addcontentsline{toc}{section}{References}
\bibliography{thebib}

\end{document}